\documentclass{article}

\RequirePackage{amsthm,amsmath,amsfonts,amssymb}
\RequirePackage{mathrsfs}
\RequirePackage{natbib}
\RequirePackage{hyperref}
\RequirePackage{booktabs, array}
\RequirePackage{lipsum}
\RequirePackage{graphicx}
\RequirePackage{subcaption}
\RequirePackage{multirow}
\RequirePackage{float}
\RequirePackage[ruled,lined,boxed,linesnumbered]{algorithm2e}
\RequirePackage{pifont}
\RequirePackage{makecell}
\RequirePackage{subcaption}
\RequirePackage[dvipsnames, table]{xcolor}

\setcitestyle{numbers,square}

\RequirePackage[ruled,lined,boxed,linesnumbered]{algorithm2e}

\RequirePackage[utf8]{inputenc}
\RequirePackage{setspace}
\usepackage{lscape}
\RequirePackage[a4paper,
    left=1.3in,
    right=1in]{geometry}

\allowdisplaybreaks
\providecommand{\keywords}[1]{\textbf{Keywords:} #1}

\newtheorem{definition}{Definition}[section]

\newtheorem{theorem}{Theorem}[section]

\newtheorem{assumption}{Assumption}[section]
\newtheorem{lemma}{Lemma}[section]
\newtheorem{corollary}{Corollary}[section]

\RequirePackage{stackengine}

\definecolor{transpred}{RGB}{255, 152, 152}
\definecolor{transpblue}{RGB}{0, 255, 128}

\newcommand{\cmark}{\ding{51}}
\newcommand{\xmark}{\ding{55}}

\newcommand\cleq{\mathrel{\stackunder{$<$}{$\sim$}}}
\newcommand\cgeq{\mathrel{\stackunder{$>$}{$\sim$}}}
\DeclareMathOperator*{\argmax}{arg\,max}
\DeclareMathOperator*{\argmin}{arg\,min}

\title{
\Large
Fast and Optimal Inference for Change Points in Piecewise Polynomials via Differencing
}

\author{Shakeel Gavioli-Akilagun\footnote{London School of Economics and Political Science; s.a.gavioli-akilagun@lse.ac.uk} \qquad Piotr Fryzlewicz\footnote{London School of Economics and Political Science; p.fryzlewicz@lse.ac.uk}
}

\begin{document}

\maketitle

\begin{abstract}
\noindent
We consider the problem of uncertainty quantification in change point regressions, where the signal can be piecewise polynomial of arbitrary but fixed degree. That is we seek disjoint intervals which, uniformly at a given confidence level, must each contain a change point location. We propose a procedure based on performing local tests at a number of scales and locations on a sparse grid, which adapts to the choice of grid in the sense that by choosing a sparser grid one explicitly pays a lower price for multiple testing. The procedure is fast as its computational complexity is always of the order $\mathcal{O} (n \log (n))$ where $n$ is the length of the data, and optimal in the sense that under certain mild conditions every change point is detected with high probability and the widths of the intervals returned match the mini-max localisation rates for the associated change point problem up to log factors. A detailed simulation study shows our procedure is competitive against state of the art algorithms for similar problems. Our procedure is implemented in the R package \texttt{ChangePointInference} which is available via \href{https://github.com/gaviosha/ChangePointInference}{GitHub}. 
\end{abstract}


\section{Introduction} \label{section: introduction}

We study the setting in which an analyst observes data $ \boldsymbol{Y}  = \left ( Y_1, \dots, Y_n \right )'$ on an equi-spaced grid which can be written as the sum of a signal component and a noise component:
\begin{equation}
Y_t = f_\circ \left ( t / n \right ) + \zeta_t \hspace{2em} t = 1, \dots, n 
\label{equation: signal + noise}
\end{equation}
The signal component $f_\circ : [0,1] \mapsto \mathbb{R}$ is known to be a piecewise polynomial function of arbitrary but fixed degree $p$. That is, associated with $f_\circ (\cdot)$ are $N$ integer valued change point locations $\Theta = \left \{ \eta_1, \dots, \eta_N \right \}$, whose number is possibly diverging with $n$, such that for each $k = 1, \dots, N$ the function can be described as a degree $p$ polynomial on the sub-interval $\left [ (\eta_{k} - p - 1)/n, \eta_{k} / n  \right ]$ but not on $\left [ (\eta_{k} - p)/n, (\eta_{k} +1) / n  \right ]$. Examples of such signals are shown in the left column of Figure \ref{figure: piecewise polynomials examples}. Both $N$ and $\Theta$ are unknown. Several algorithms exist for estimating $N$ and $\Theta$ in specific instances of model (\ref{equation: signal + noise}), such as when $f_\circ (\cdot)$ is piecewise constant \cite{killick2012optimal, fryzlewicz2014wild, eichinger2018mosum, frick2014multiscale} or when $f_\circ (\cdot)$ is piecewise linear \cite{fearnhead2019detecting, baranowski2019narrowest, anastasiou2022detecting, maeng2019detecting}. While the piecewise constant and piecewise linear change point regression are well studied, the generic piecewise polynomial model has attracted less attention. Nevertheless, the piecewise polynomial model has practical applications in areas as diverse as finance \cite{qiu1998local, liu2018jump, mczgee1970piecewise}, aerospace engineering \cite{cunis2019piecewise}, protein folding \cite{butt2005force}, light transmittance \cite{abramovich2007estimation}, climatology \cite{aue2008testing}, and data compression \cite{gao2020fast}. 

Our goal in this paper is to simultaneously quantify the level of uncertainty the around the existence and location of each putative change point in the generic piecewise polynomial model. This is a worthwhile task since estimates of the change point locations are not consistent in the usual sense: the best rate at which a change point can be localised on the domain $\left \{ 1, \dots, n \right \}$ is $\mathcal{O}_\mathbb{P} (1)$ \citep{verzelen2020optimal, wang2020univariate}, however this can be as high as $\mathcal{O}_\mathbb{P} \left ( n^{\varpi_k} \right )$ for some $\varpi_k \in [0,1)$ depending on the smoothness of $f_\circ (\cdot)$ at each change point location $\eta_k$ \cite{yu2020review,yu2022localising}. Moreover, since most algorithms do not quantify uncertainty around the change points they recover, it is difficult to say whether these change points are real or spuriously estimated. 

\begin{figure}[!htbp]
\centering
\caption{the piecewise constant \texttt{blocks} signal, piecewise linear \texttt{waves} signal, and piecewise quadratic \texttt{hills} signal each contaminated with i.i.d. Gaussian noise (left column). Intervals of significance with uniform $90\%$ coverage returned by our procedure (right column). Black dashed lines (\textbf{- - -}) represent underlying piecewise polynomial signal, light grey lines (\textcolor{gray}{\textbf{---}}) represent the observed data sequence, red shaded regions (\textcolor{transpred}{$\blacksquare$}) represent intervals of significance returned by our procedure, red dotted lines (\textcolor{red}{$\cdot\cdot\cdot$}) represent split points within each interval associated with the piecewsie polynomial fit providing the lowest sum of squared residuals.} 
\label{figure: piecewise polynomials examples}
\begin{subfigure}[b]{0.4\textwidth}
\centering
\includegraphics[width=\textwidth]{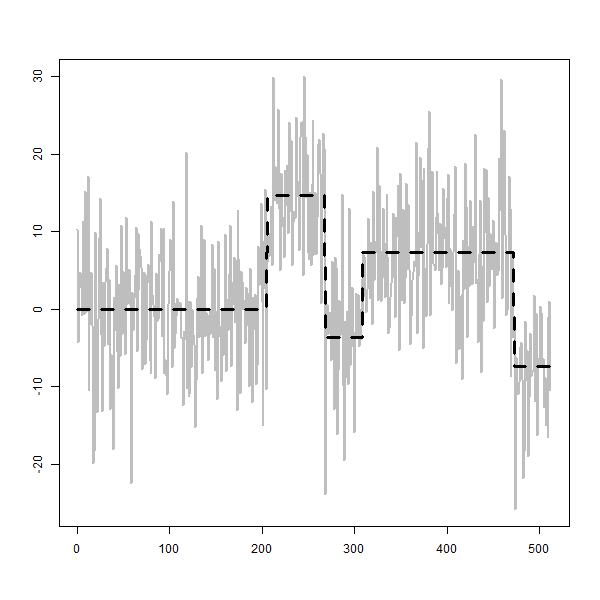}
\caption{the \texttt{blocks} signal}
\label{subfigure: blocks}
\end{subfigure}
\begin{subfigure}[b]{0.4\textwidth}
\centering
\includegraphics[width=\textwidth]{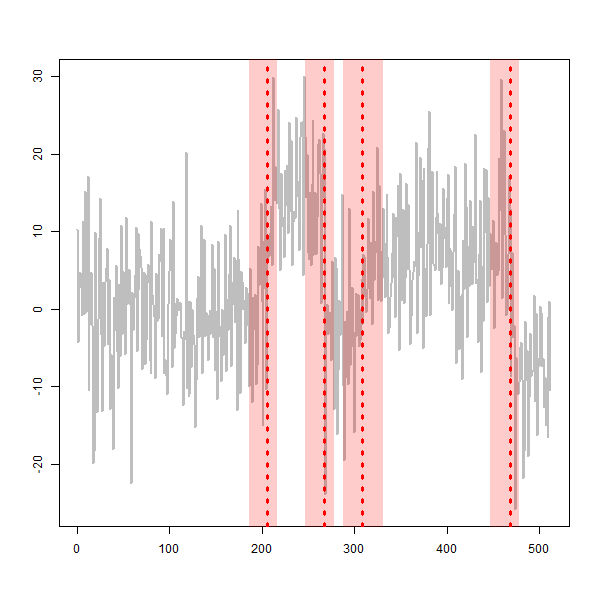}
\caption{intervals recovered}
\end{subfigure}

\begin{subfigure}[b]{0.4\textwidth}
\centering
\includegraphics[width=\textwidth]{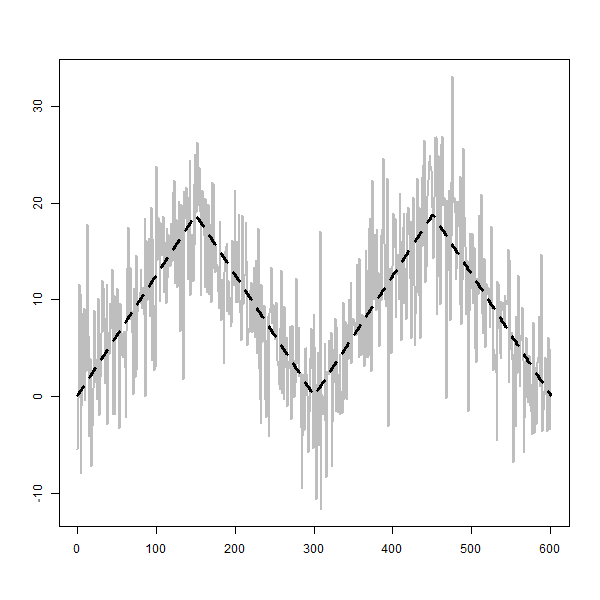}
\caption{the \texttt{waves} signal}
\label{subfigure: waves}
\end{subfigure}
\begin{subfigure}[b]{0.4\textwidth}
\centering
\includegraphics[width=\textwidth]{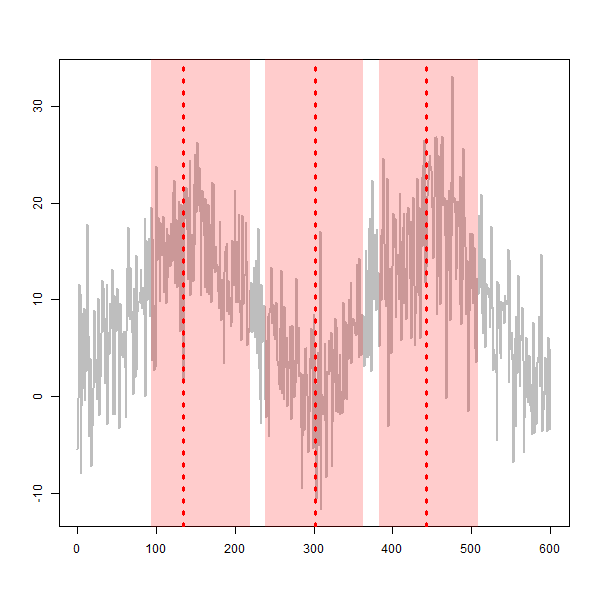}
\caption{intervals recovered}
\end{subfigure}

\begin{subfigure}[b]{0.4\textwidth}
\centering
\includegraphics[width=\textwidth]{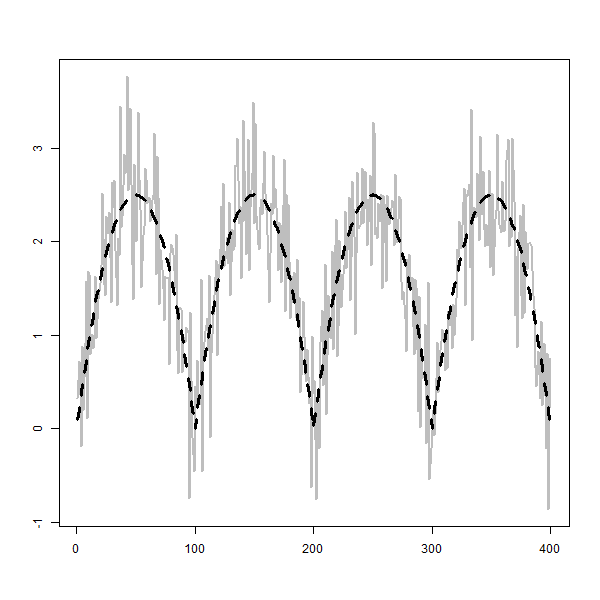}
\caption{the \texttt{hills} signal}
\label{subfigure: hills}
\end{subfigure}
\begin{subfigure}[b]{0.4\textwidth}
\centering
\includegraphics[width=\textwidth]{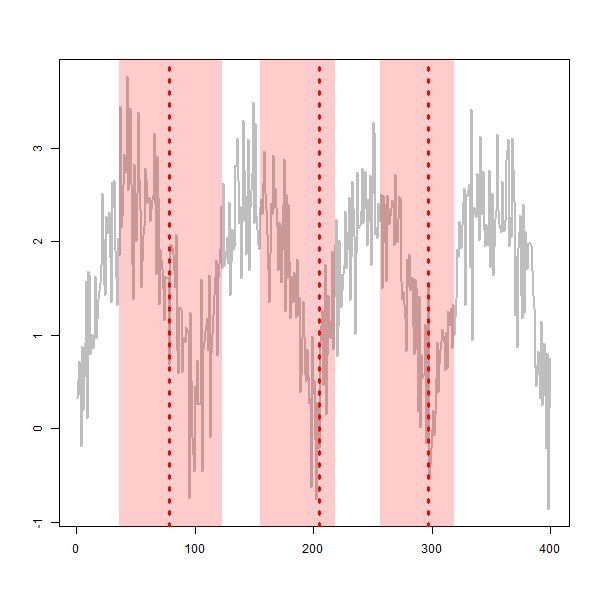}
\caption{intervals recovered}
\end{subfigure}
\end{figure}

We propose a procedure which aims to return the narrowest possible disjoint sub-intervals of the index set $\left \{ 1, \dots, n \right \}$ in such a way that each must contain a change point location uniformly at some confidence level chosen by the user. Examples of such intervals are shown in the right column of Figure \ref{figure: piecewise polynomials examples}. This is done by testing for a change at a range of scales and locations belonging to a sparse grid, and tightly bounding the supremum of local test statistics over the same grid which guarantees sharp unconditional family-wise error control. An advantage of this approach is that once can post-search each unconditional interval for the best location(s) of the change-point(s) with appropriate statistics without worrying about significance testing.  We initially study the setting in which the noise components are independent with marginal $\mathcal{N} \left ( 0 , \sigma^2 \right )$ distribution and later in Section \ref{section: extension to dependent and non-Gaussian noise} extend our results to dependent and non-Gaussian noise. Motivated by the fact that taking $(p+1)$-th differences will eliminate a degree $p$ polynomial trend \citep{chan1977note}, we consider tests based on differences of (standardised) local sums of the data sequence. There are several advantages to working with tests based on local sums as opposed to for example likelihood ratio or Wald statistics, which we list below. 

\begin{itemize}
    \item Each of our local test can be completed in $\mathcal{O} (1)$ time in a straightforward manner, regardless of the degree of the underlying polynomial or the scale at which the test is performed, leading to a procedure with worst case complexity $\mathcal{O} \left ( n \log (n) \right )$ when test are carried out on a sparse grid.  
    \item Local averaging brings the contaminating noise closer to Gaussianity, which is a feature we exploit in Section \ref{section: extension to dependent and non-Gaussian noise} when studying the behaviour of the procedure under non-Gaussian and possibly dependent noise. 
    \item Unlike procedures based on differencing the raw data, which are known to be sub-optimal, as we show in Theorem \ref{theorem: optimality of change point detection} the combination of local averaging followed by differencing leads to a procedure which is optimal in a mini-max sense.
    \item The asymptotic analysis is by design uncomplicated, as it boils down to analysing the high excursion probability of a stationary Gaussian field whose local structure depends on the polynomial degree in a straightforward way. 
\end{itemize}

We now review existing procedures in the literature for change point inference in specific instances of model (\ref{equation: signal + noise}). If one is able to localise all change points at a fast enough rate it is possible to construct asymptotically valid confidence intervals for the change point locations. This is done by \cite{meier2021mosum, cho2022bootstrap} for piecewise constant $f_\circ (\cdot)$ and by \cite{bai1998estimating, bai2003computation} for regressions with piecewise constant coefficients. A crucial limitation of these approaches is that confidence intervals are only valid conditional on the number of change points being correctly estimated. Since there is no guarantee this will happen in a finite sample these intervals are problematic to interpret in practice. We further note that the piecewise constant regression model considered by \cite{bai1998estimating, bai2003computation} does not actually cover generic the piecewise polynomial model (\ref{equation: signal + noise}), as it is necessary to assume the regressors are stationary or satisfy certain regularity conditions which are necessarily violated by polynomial functions of $t$. The SMUCE estimator and its many variants \cite{frick2014multiscale, pein2017heterogeneous, dette2020multiscale, jula2022multiscale} estimates a piecewise constant signal subject to the constraint that empirical residuals pass a multi-scale test, and produces a confidence set for the signal from which uniform confidence intervals for the change point locations can be extracted. However, the multi-scale tests have been observed to be poorly calibrated \citep{fryzlewicz2021robust}. Moreover, letting the size of the test determine the estimated signal leads to larger nominal coverage actually reducing coverage \citep{chen2014discussion}. In \cite{fang2020detection, fang2020segmentation} approximations of the tail probability of the supremum of local likelihood ratio tests for constant means and constant slopes calculated at all possible scales and locations are derived, and an algorithm is provided which returns uniform confidence intervals for the change point locations. However, the approach in these papers does not extend to the case of generic piecewise polynomials, and the algorithm propose has cubic time complexity in the worst case. The Narrowest Significance Pursuit algorithm \cite{fryzlewicz2023narrowest, fryzlewicz2021robust} tests for local deviations from a linear model using the multi-resolution norm \citep{mildenberger2008geometric}, and bounds the supremum of local tests by the multi-resolution norm of the unobserved noise which in turn can be controlled using the results in \cite{kabluchko2007extreme, kabluchko2014limiting, ravckauskas2003invariance}. However, computing each local test requires solving a linear program, which makes the procedure slow in practice. Moreover, other than for piecewise constant and continuous piecewise linear signals contaminated with Gaussian noise, it has not been shown that the procedure can detect change points optimally. 

We finally review some approaches for problems closely related to the one studied in this paper. The problem of testing for the presence of a single change point in piecewise polynomials has previously been considered by \cite{jaruvskova1999testing, aue2008testing, aue2009extreme} by studying the supremum of likelihood ratio tests, and by \cite{jandhyala1997iterated, macneill1978properties} by studying partial sums of residuals from a least squares fit. These tests however do not extend to the case of multiple change points, which is the focus of this paper. Given estimates $\widehat{\Theta}$ and $\widehat{N}$ some authors focus on testing whether a change did in fact occur at each estimated change point location. This is a post selection inference problem as it requires conditioning on the estimation of $\widehat{\Theta}$, and has been studied by \cite{hyun2021post, jewell2022testing, carrington2023improving} for piecewise constant $f_\circ(\cdot)$ and by \cite{mehrizi2021valid} for generic piecewise polynomial $f_\circ(\cdot)$. However, the goal of these methods is to quantify uncertainty about the size of each change, whereas our goal is to simultaneously quantify uncertainty about the existence and the location of the change point. The piecewise polynomial problem is closely related to the problem of detecting changes in the smoothness of the regression function in nonparametric regressions \cite{raimondo1998minimax, hall1992edge}. The distribution of certain estimators for the location of a single change have been derived for instance by \cite{muller1992change}; such results allow for inference on the location of the change. Our focus on the piecewise polynomial problem is motivated by practical considerations: a parametric model is often preferable to practitioners due to ease of interpretability. Finally,  Bayesian approaches to change point detection provide an alternative approach to uncertainty quantification, via credible intervals derived from the posterior. However, choosing sensible priors and sampling from the posterior remain non-trial tasks. Methods for evaluating the posterior have been studied by \cite{rigaill2012exact, fearnhead2006exact, nam2012quantifying}.

The remainder of the paper is structured as follows. In Section \ref{section: local tests} we introduce local tests for the presence of a change based on differences of local sums of the data, and study their behaviour under the null of no change points in terms of the family-wise error when the test are applied over a sparse grid. In Section \ref{section: fast algorithm} we introduce a fast algorithm for turning our local tests into a collection of disjoint intervals which each must contain a change at a prescribed significance level, and show the algorithm's consistency and optimality in terms of recovering narrow intervals which each contain a change point location. In section \ref{section: numerical illustrations} we compare the performance of our algorithm with that of existing procedures when applied to simulated data. Finally in Section \ref{section: real data examples} we show the practical use of our algorithm via two real data examples. 

\section{Difference based tests with family-wise error control} \label{section: local tests}

\subsection{Local tests for a change point}

We begin by describing tests for the presence of a change on a localised segment of the data. Motivated by the fact that a polynomial trend will be eliminated by differencing, if it were suspected that a segment of the data contained a change point location one could divide the segment into $p+2$ chunks of roughly equal size and take the $(p+1)$-th difference of the sequence of local sums on each chunk. Since summing boosts the signal from the change point, and differencing eliminates the polynomial trend, one could then declare a change if the resulting quantity coming from the summed and differenced sequence, appropriately scaled, was large in absolute value. By contrast, simply differencing the data on the segment would reduce the signal from the change, and any statistic based on the differenced data only would be sub-optimal for detecting the change. 

For each local test we write $l$ for the location of the data segment being inspected for a change point and $w$ for the width of the data segment. Following the reasoning above, to test for the presence of a change point on the interval $\left \{ l, \dots, l + w - 1 \right \}$ we first compute the following non-overlapping local sums:
\begin{equation*}
\bar{Y}^j_{l,w} = Y_{l + j \left \lfloor \frac{w}{p+2} \right \rfloor} + \dots + Y_{l + (j+1) \left \lfloor \frac{w}{p+2} \right \rfloor -1} \hspace{2em} j = 0, \dots, p + 1 
\label{equation: local sum}
\end{equation*}
We then declare a change if the test statistic defined below, which corresponds to the the $(p+1)$-th differences of the sequence $\bar{Y}^0_{l,w}, \dots, \bar{Y}^{p+1}_{l,w}$ scaled so that its variance is constant independent of $l$ and $w$ when the noise is homoskedastic and independently distributed, is large in absolute value. 
\begin{equation}
D_{l,w}^p \left ( \boldsymbol{Y} \right ) = \left \{ \left \lfloor \frac{w}{p+2} \right \rfloor \sum_{i=0}^{p+1} \binom{p+1}{i}^2 \right \}^{-1/2} \sum_{j=0}^{p+1} \left ( - 1 \right ) ^ {p+1-j} \binom{p+1}{j} \bar{Y}^j_{l,w}
\label{equation: local test statistic}
\end{equation}
The functional introduced in (\ref{equation: local test statistic}) enjoys the following properties, which make it well suited for the task of change change point testing on piecewise polynomials: 
\begin{itemize}
\item \underline{Additivity:} for any two vectors $\boldsymbol{f},\boldsymbol{g} \in \mathbb{R}^n$ it holds that $D_{l,w}^p \left ( \boldsymbol{f + g} \right ) = D_{l,w}^p \left ( \boldsymbol{f} \right ) + D_{l,w}^p \left ( \boldsymbol{g} \right )$ for all admissible $l$'s and $w$'s.
\item \underline{Annihilation of polynomials:} if the entries of $\boldsymbol{f} \in \mathbb{R}^n$ are from a polynomial of degree no larger than $p$ then $D_{l,w}^p \left ( \boldsymbol{f} \right ) = 0$ for all admissible $l$'s and $w$'s. 
\item \underline{Large for discontinuous functions:} if the entries of $\boldsymbol{f} \in \mathbb{R}^n$ are from a piecewise monomial with a single discontinuity at location $\eta$ then $ | D_{l,w}^p \left ( \boldsymbol{f} \right ) | > 0$ for all $l$'s and $w$'s such that $\eta \in \{ l, \dots, l + w - 1 \}$. 
\end{itemize}
The first two properties ensure (\ref{equation: local test statistic}) is small under the local null of no change, whereas the third property can be used to show that for some admissible $(l,w)$ pair the the statistic will be large in absolute value in the presence of a change.

Consequently, for some $\lambda > 0$ to be chosen later on, each local test for the presence of a change on an interval $\left \{ l, \dots, l + w - 1 \right \}$ takes the following form:
\begin{equation}
T_{l,w}^\lambda \left ( \mathbf{Y} \right ) = \boldsymbol{1} \left \{ |D_{l,w}^p \left ( \boldsymbol{Y} \right )| > \lambda \right \}.
\label{equation: local tests}
\end{equation}
When $p = 0$ the statistic (\ref{equation: local test statistic}) recovers the moving sum filter used for change point detection in the piecewise constant model \cite{eichinger2018mosum}. This also corresponds to the (square root of) the likelihood ratio statistic for testing the null of a constant mean on the segment under Gaussian noise, as well as the Wald statistic for the same problem. Typical approaches for generalising to higher order polynomial change point problems involve local likelihood-ratio or Wlad statistics for testing the null of a polynomial mean on the segment \cite{baranowski2019narrowest, fang2020segmentation, anastasiou2022detecting, kim2022moving}, which however are hard to stochastically control. We show that simply extending the order of differencing leads to simple and powerful tests. 

\subsection{Local tests on a sparse gird}
For the purpose of making inference statements about an unknown number of change point locations, we would like to apply the local tests (\ref{equation: local tests}) over a grid which is both dense enough to cover all potential change point locations well and sparse enough to allow all local tests to be computed quickly. Given a suitable grid $\mathcal{G}$ of $(l,w)$ pairs, if $\lambda$ were chosen to control the family-wise error of the collection of tests
\begin{equation}
\mathcal{T}_{\mathcal{G}}^\lambda \left ( \mathbf{Y} \right ) = \left \{ T_{l,w}^\lambda \left ( \mathbf{Y} \right ) \mid (l,w) \in \mathcal{G} \right \}
\label{equation: collection of local tests}
\end{equation}
at some level $\alpha$, we could be sure that with probability $1 - \alpha$ every $(l,w)$ pair on which a test rejects corresponds to a segment of the data containing at least one change point location. 

We propose to use the following grid, which is parameterised by a minimum grid scale parameter $W$, controlling the minimum support of the detection statistic (\ref{equation: local test statistic}), and a decay parameter $a > 1$, controlling the density of the grid:
\begin{gather}
\mathcal{G} \left ( W, a \right ) = \left \{ \left ( l , w \right ) \in \mathbb{N}^2 \mid w \in \mathcal{W} (W,a) , 1 \leq l \leq n - w \right \} \label{equation: a-adic grid} \\ 
\mathcal{W} \left ( W , a \right ) = \left \{ w = \left \lfloor a^k \right \rfloor \mid \left \lfloor \log_a (W) \right \rfloor \leq k \leq \left \lfloor \log_a (n/2) \right \rfloor \right \}. \nonumber
\end{gather}
Associated with the grid is the collection of sub-intervals of $\left \{ 1, \dots, n \right \}$ whose length is larger than $W$ and can be written as an integer power of $a$. For example, the collection of intervals $\left \{ l, \dots, l + w - 1 \right \} $ associated with the $(l,w)$ pairs in the grid obtained when $n = 20$ and setting $W = 2$ and $a = 2$ is shown in Figure \ref{figure: example gird} below. For this configuration of $a$ and $W$, the associated collection of intervals consists of all contiguous sub-interval of $\left \{ 1, \dots, 20 \right \}$ having dyadic length. 

\begin{figure}[h!]
\centering
\includegraphics[width=\textwidth]{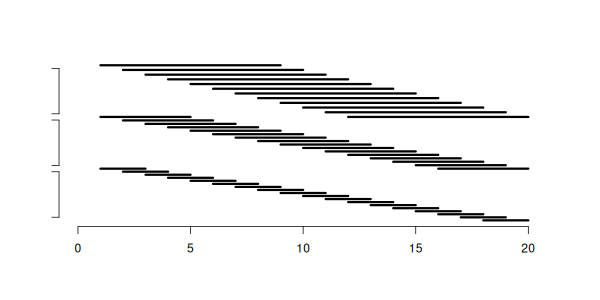}
\caption{intervals associated with $\mathcal{G} \left ( W , a \right )$ when $n = 20$, $W = 2$, and $a = 2$}
\label{figure: example gird}
\end{figure}

The grid defined by (\ref{equation: a-adic grid}) is similar to several grids already proposed for different change point detection problems \cite{kovacs2020seeded, chan2013detection, pilliat2023optimal}, in that the size of scales decays exponentially. Two key difference are first that for any scale $w$ all possible locations $l$ are considered, and second that all scales with $w = o \left ( W \right )$ are excluded from the grid. Regarding the minimum grid scale, if the noise were known to be independently distributed and Gaussian we could take $W = \mathcal{O}(1)$ and still retain family-wise error control using our proof technique. However, under dependent or non-Gaussian noise letting the minimum grid scale diverge at an appropriate rate with $n$ is necessary for controlling the family-wise error, as this allows local sums of the noise to be treated as approximately uncorrelated and Gaussian. 

\subsection{Family-wise error control under Gaussianity} \label{section: Family-wise error control under Gaussianity}

As a starting point for family-wise error analysis in more general noise settings, we first show how to control the family-wise error of the local tests (\ref{equation: local tests}) over the grid (\ref{equation: a-adic grid}) when the noise terms are independently distributed and Gaussian. The approach is to tightly bound the maximum of the local test statistics (\ref{equation: local test statistic}) under the null of no change points, and use this bound to select an appropriate threshold $\lambda$ for (\ref{equation: collection of local tests}). We impose the following assumptions on the minimum grid scale and on the noise components. 

\begin{assumption}
The noise terms $\zeta_1, \dots, \zeta_n$ are mutually independent with marginal $\mathcal{N} (0, \sigma^2)$ distribution for some $\sigma > 0$. 
\label{assumption: gaussian noise}
\end{assumption}

\begin{assumption}
The minimum grid scale $W$ satisfies $W / \log (n) \rightarrow d$ for some $d \in (0, \infty)$.
\label{assumptioin: logarithmic segment length}
\end{assumption}

With these assumption in place we have the following result on the behaviour of the maximum of local test statistics (\ref{equation: local test statistic}) under the null of no change points. 

\begin{theorem}
Let $\boldsymbol{Y} = \left ( Y_1, \dots, Y_n \right )'$ be from model (\ref{equation: signal + noise}) with signal component having no change points and grant Assumptions \ref{assumption: gaussian noise} - \ref{assumptioin: logarithmic segment length} hold. For fixed $a > 1$ introduce the following quantity: 

\begin{equation*}
M_{\mathcal{G}(W,a)}^\sigma \left ( \boldsymbol{Y} \right ) = \max_{(l,w) \in \mathcal{G} (W,a)} \left \{ \frac{1}{\sigma} D_{l,w}^p \left ( \boldsymbol{Y} \right ) \right \}.
\end{equation*}

(i) Putting $\mathfrak{a}_{n} = \sqrt{2 \log (n)}$ and $\mathfrak{b}_{n} = 2 \log(n) - \frac{1}{2} \log \log (n) - \log( 2 \sqrt{\pi})$ the sequence of random variables $\left \{ \mathfrak{a}_{n} M_{\mathcal{G}(W,a)}^\sigma \left ( \boldsymbol{Y} \right ) - \mathfrak{b}_{n} \mid n \in \mathbb{N} \right \}$ is tight, and there are constants $H_{1,1}$ and $H_{1,2}$ depending only on $a$, $p$, and $d$ such that for fixed $x \in \mathbb{R}$ the following holds
\begin{equation*}
o(1) + \exp \left ( - H_{1,1} e^{-x} \right ) \leq \mathbb{P} \left ( \mathfrak{a}_{n} M_{\mathcal{G}(W,a)}^\sigma \left ( \boldsymbol{Y} \right ) - \mathfrak{b}_{n} \leq x \right ) \leq \exp \left ( - H_{1,2} e^{-x} \right ) + o(1). 
\end{equation*}

(ii) Moreover the result in (i) continues to hold if $\sigma$ is replaced with any consistent estimator $\widehat{\sigma}$ which satisfies $\left | \widehat{\sigma} / \sigma -  1 \right | = o_\mathbb{P} \left ( \log ^{-1} (n)  \right )$. 

\label{theorem: Gaussian tightness of normalized maximum}
\end{theorem}

Note that for large values of $n$ the quantity 
\begin{equation*}
L_{\mathcal{G}(W,a)}^\sigma \left ( \boldsymbol{Y} \right ) = \max_{(l,w) \in \mathcal{G} (W,a)} \left \{ \frac{1}{\sigma} \left | D_{l,w}^p \left ( \boldsymbol{Y} \right ) \right | \right \}
\end{equation*}
behaves asymptotically like the maximum of two independent copies of $M_{\mathcal{G}(W,a)}^\sigma \left ( \boldsymbol{Y} \right )$. We do not give a formal proof of this statement, however the statement can be understood intuitively by writing $L_{\mathcal{G}(W,a)}^\sigma \left ( \boldsymbol{Y} \right ) = M_{\mathcal{G}(W,a)}^\sigma \left ( \boldsymbol{Y} \right ) \vee M_{\mathcal{G}(W,a)}^\sigma \left ( -\boldsymbol{Y} \right )$ and then using the well known fact that order statistics are asymptotically independent \citep{falk1988independence, kabluchko2014limiting}. Therefore, in light of Theorem \ref{theorem: Gaussian tightness of normalized maximum} it follows that under Assumptions \ref{assumption: gaussian noise} - \ref{assumptioin: logarithmic segment length}, for any $\alpha \in (0,1)$, choosing $\lambda = \widehat{\sigma} \lambda_\alpha$ with $\widehat{\sigma}$ satisfying the condition given in part (ii) and $\lambda_\alpha$ defined as follows
\begin{equation}
\lambda_\alpha = \sqrt{2 \log(n)} + \frac{-\frac{1}{2} \log\log (n) - \log \left ( 2 \sqrt{\pi} / H_{1,2} \right ) + \log \left ( -2 \log ^ {-1} \left ( 1 - \alpha \right ) \right )}{\sqrt{2 \log(n)}}
\label{equation: Gaussian FWE control lambda}
\end{equation}
will result in the collection of tests $\mathcal{T}^\lambda_{\mathcal{G} \left ( W , a \right )} \left ( \mathbf{Y} \right )$ having family-wise error asymptotically no larger than $\alpha$. In Section \ref{section: varaince and LRV estimation} we given an example of an estimator $\widehat{\sigma}$ which satisfies condition (ii) in Theorem \ref{theorem: Gaussian tightness of normalized maximum} above, even if the data contains change points, provided the number of change points does not grow too quickly with the $n$. 

Importantly, the threshold (\ref{equation: Gaussian FWE control lambda}) explicitly accounts for the grid used, in the sense that if one chooses a coarser gird a lower price is paid for multiple testing. More specifically, if one chooses a coarser grid by increasing $a$ the constant $H_{1,2}$ adjusts which reduces the size of (\ref{equation: Gaussian FWE control lambda}). As a result, each local test performed will have higher power with the same family-wise error guarantee. Naturally, on a coarser grid the collection of tests may overall have lower power for detecting a change, since fewer tests are carried out in total. 

The constants $H_{1,1}$ and $H_{1,2}$ are defined explicitly below, where we put $b_1 = 1/a$ and $b_2 = 1$, and $\bar{\Phi} \left ( \cdot \right )$ for the tail function of a standard Gaussian random variable. 
\begin{align}
& H_{1,i} = \sum_{j=0}^\infty p_\infty^2 \left ( \frac{2 C_p}{a^j b_i d} \right ) \hspace{2em} i = 1,2 \label{equation: gaussian limit constants} \\
& p_\infty \left ( x \right ) = \exp \left ( - \sum_{k=1}^\infty \frac{1}{k} \bar{\Phi} \left ( \sqrt{kx / 4} \right ) \right ) \nonumber \\ 
& C_p = \left ( p+2 \right ) \left ( 1 + \sum_{j=1}^{p+1} \binom{p+1}{j}\binom{p+1}{j-1} \big / \sum_{i=0}^{p+1}\binom{p+1}{i}^2 \right ) \nonumber
\end{align}
The effect of the decay parameter $a$ on $H_{1,1}$ and $H_{1,2}$ can now be understood via (\ref{equation: gaussian limit constants}) using the additional fact that \cite[Corollary 3.18]{kabluchko2007extreme} for any $C > 0$ the quantity $p_\infty^2 \left ( C / x \right )$ behaves like $C/(2x)$ when $x$ is large. 

We now explain the origin of the double inequality in Theorem \ref{theorem: Gaussian tightness of normalized maximum}, and why it is sufficient for strong family-wise error control. In Theorem \ref{theorem: Gaussian tightness of normalized maximum} we are only able to establish tightness of the normalised maximum, as opposed to convergence to an extreme value distribution, for the following reason: the maximum over standardised increments of a sequence of Gaussian variables will be achieved on scales of the order $\mathcal{O} \left ( \log (n) \right )$ as was shown by \cite{kabluchko2007extreme, kabluchko2014limiting}, but scales of this order cannot necessarily be expressed as integer powers of $a$. Consequently the choice of grid introduces small fluctuations in the maximum, which persist in the limit, and correspond to the difference between $\log (n)$ and the closest integer power of $a$. However for a sub-sequence of $n$'s on which the quantity $b_{n} = a^{\left \lfloor \log_a (W) \right \rfloor} / W$ converges the normalised maximum does converge. The constants $H_{1,1}$ and $H_{1,2}$ therefore correspond to the largest and smallest constants which may appear in the extreme value limit on a sub-sequence of $n$'s on which $b_{n}$ converges to some constant. Such fluctuations can arise even for maxima of sequences of i.i.d. random variables \cite{anderson1970extreme}; for instance the maximum of $n$ i.i.d. Poisson random variables with fixed rate fluctuates between two integers \cite{kimber1983note}. 

\subsection{Extension to dependent and non-Gaussian noise} \label{section: extension to dependent and non-Gaussian noise}

We now extend the result of Theorem \ref{theorem: Gaussian tightness of normalized maximum} to dependent and non-Gaussian noise. This is done through the standard approach \cite{huvskova2001permutation, kirch2021moving, eichinger2018mosum} of computing local tests only on scales large enough such that partial sums of the data can be replaced by increments of a Wiener process without affecting the asymptotics. Therefore, we impose the following assumptions on the minimum grid scale and the noise component. 

\begin{assumption}
The noise terms are mean zero and weakly stationary, with auto-covariance function $\gamma_h = \text{Cov} \left ( \zeta_0, \zeta_h \right )$ and strictly positive long run variance $\tau^2 = \gamma_0 + 2 \sum_{h > 0} \gamma_h$. 
\label{assumption: positive long run varaince}
\end{assumption}

\begin{assumption}
There exists a Wiener process $\left \{ B (t) \right \}_{t>0}$ such that for some $\nu > 0$, possibly after enlarging the probability space, it holds $\mathbb{P}$-almost surely that $\sum_{t=1}^n \zeta_t - \tau B (n) = \mathcal{O} \left ( n^\frac{1}{2+\nu} \right )$. 
\label{assumption: strong approximation}
\end{assumption}

\begin{assumption}
With the same $\nu$ as in Assumption \ref{assumption: strong approximation}, the minimum grid scale $W$ satisfies $n / W \rightarrow \infty$ and $n^\frac{2}{2+\nu} \log (n) / W \rightarrow 0$. 
\label{assumption: minimum segment}
\end{assumption}

Assumption \ref{assumption: strong approximation} holds under a wide range of common dependence conditions such as $\beta$-mixing, functional dependence, and auto-covaraince decay \cite{berkes2014komlos, philipp1975almost, kuelbs1980almost}; these dependence conditions in turn hold for a range of popular time series models such as ARMA, GARCH, and bilinear models \cite{doukhan2012mixing,wu2005nonlinear}. If the noise terms are independently distributed Assumption \ref{assumption: strong approximation} holds as long as their $(2+\nu)$-th moment is bounded \cite{komlos1975approximation, csorgo2014strong}. With these assumption in place we have the following result on the behaviour of the maximum of local test statistics (\ref{equation: local test statistic}) under the null of no change points. 

\begin{theorem}
Let $\boldsymbol{Y} = \left ( Y_1, \dots, Y_n \right )'$ be from model (\ref{equation: signal + noise}) with signal component having no change points and grant Assumptions \ref{assumption: positive long run varaince} - \ref{assumption: minimum segment} hold. For fixed $a > 1$ introduce the following quantity: 

\begin{equation*}
M_{\mathcal{G}(W,a)}^\tau \left ( \boldsymbol{Y} \right ) = \max_{(l,w) \in \mathcal{G}(W,a)} \left \{ \frac{1}{\tau} D_{l,w}^p \left ( \boldsymbol{Y} \right ) \right \}.
\end{equation*}

(i) Putting $\mathfrak{a}_{n,W} = \sqrt{2 \log (n / W)}$ and $\mathfrak{b}_{n,W} = 2 \log(n / W) + \frac{1}{2} \log \log (n / W) - \log(\sqrt{\pi})$ the sequence of random variables $\left \{ \mathfrak{a}_{n,W} M_{\mathcal{G}(W,a)}^\tau \left ( \boldsymbol{Y} \right ) - \mathfrak{b}_{n,W} \mid n \in \mathbb{N} \right \}$ is tight, and there are constants $H_{2,1}$ and $H_{2,2}$ depending only on $a$ and $p$ such that for fixed $x \in \mathbb{R}$ the following holds

\begin{equation*}
o(1) + \exp \left ( - H_{2,1} e^{-x} \right ) \leq \mathbb{P} \left ( \mathfrak{a}_{n,W} M_{\mathcal{G}(W,a)}^\tau \left ( \boldsymbol{Y} \right ) - \mathfrak{b}_{n,W} \leq x \right ) \leq \exp \left ( - H_{2,2} e^{-x} \right ) + o(1)
\end{equation*}.

(ii) Moreover the result in (i) continues to hold if $\tau$ is replaced with any consistent estimator $\widehat{\tau}$ 
 which satisfies $\left | \widehat{\tau} / \tau -  1 \right | = o_\mathbb{P} \left ( \log ^ {-1} (n / W) \right )$. 

\label{theorem: tightness of non-gaussian dependent maximum}
\end{theorem}

By the same reasoning used in Section \ref{section: Family-wise error control under Gaussianity} under assumptions \ref{assumption: positive long run varaince} - \ref{assumption: minimum segment} Theorem \ref{theorem: tightness of non-gaussian dependent maximum} guarantees that choosing $\lambda = \widehat{\tau} \lambda_\alpha$, with $\widehat{\tau}$ satisfying the condition given in part (ii), and with $\lambda_\alpha$ defined as follows 
\begin{equation}
\lambda_\alpha = \sqrt{2 \log(n/W)} + \frac{\frac{1}{2} \log\log (n/W) - \log \left ( \sqrt{\pi} / H_{2,2} \right ) + \log \left ( -2 \log ^ {-1} \left ( 1 - \alpha \right ) \right )}{\sqrt{2 \log(n/W)}}
\label{equation: generic FWE control lambda}
\end{equation}
will result in the collection of tests $\mathcal{T}^\lambda_{\mathcal{G} \left ( W , a \right )} \left ( \mathbf{Y} \right )$ having family-wise error asymptotically no larger than $\alpha$. In Section \ref{section: varaince and LRV estimation} we give examples of variance and long run variance estimators which satisfy condition (ii) in Theorem \ref{theorem: tightness of non-gaussian dependent maximum}, even in the presence of change points, provided the number of change points does not grow too quickly with $n$.  

By the same mechanism as in Theorem \ref{theorem: Gaussian tightness of normalized maximum} the threshold (\ref{equation: generic FWE control lambda}) is adaptive to the chosen grid. The constants $H_{2,1}$ and $H_{2,2}$ in Theorem \ref{theorem: tightness of non-gaussian dependent maximum} are as shown below, where $C_p$ and $b_i$ are as in Section \ref{section: Family-wise error control under Gaussianity}. 

\begin{equation*}
H_{2,i} = \frac{b_i^{-1 }C_p}{1 - a^{-1}} \hspace{2em} i = 1,2
\end{equation*}

The proofs of Theorems \ref{theorem: Gaussian tightness of normalized maximum} and \ref{theorem: tightness of non-gaussian dependent maximum} reveal that maxima achieved over different scales in the grid (\ref{equation: a-adic grid}) will be asymptotically independent. This combined with the tightness of the normalised maximum shows that the thresholds (\ref{equation: Gaussian FWE control lambda}) and (\ref{equation: generic FWE control lambda}) are the sharpest possible for each scale in the grid, under their respective sets of assumptions. That is, if one were to restrict tests to a single scale of the order $\mathcal{O} \left ( W \right )$ the threshold needed to control the family-wise error of the collection of tests would be asymptotically equivalent to the thresholds presented for controlling the family-wise error of test conducted on the whole grid. 

\section{A fast algorithm for change point inference} \label{section: fast algorithm}

\subsection{The algorithm}

We now present an algorithm, based on the tests introduced in Section \ref{section: local tests}, for efficiently recovering disjoint sub-intervals of the index set $\left \{ 1, \dots n \right \}$ in such a way that each must contain a change point uniformly at some prescribed significance level $\alpha$. The algorithm is motivated by the Narrowest Significance Pursuit proposed by \cite{fryzlewicz2023narrowest}, in that the focuses is on recovering theses intervals through a series of local tests so that each interval is the narrowest possible. However, there are several important differences between our approach and the approach in \cite{fryzlewicz2023narrowest}, which we outline below before presenting the algorithm. 
\begin{itemize}
    \item Each of our local tests can be computed in constant time as a function of the sample size and independently of the scale of the computation. This is not the case for \cite{fryzlewicz2023narrowest}, where each local test requires solving a linear program. 
    \item We compute local tests over the sparse grid defined in (\ref{equation: a-adic grid}), whereas \cite{fryzlewicz2023narrowest} uses a two stage procedure where local tests are initially performed over a coarse grid and intervals flagged in the first stage are exhaustively sub-searched. In the worst case the former leads to $\mathcal{O} \left ( n \log (n) \right )$ tests being carried out, whereas the latter may lead to $\mathcal{O} \left ( n^2 \right )$ test being performed. 
    \item The thresholds used used in our local tests are designed to adapt to the chosen grid, which accounts for the statistical-computational trade off in large scale problems change point problems. However, the threshold used in \cite{fryzlewicz2023narrowest} does not depend on the chosen grid. 
\end{itemize}

Given a grid of $(l,w)$ pairs $\mathcal{G} \left ( W , a \right )$ constructed according to (\ref{equation: a-adic grid}) our approach is to greedily search for a pair on which the associated local test (\ref{equation: local tests}) declares a change, starting from the finest scale in the grid. When such a pair is found the associated interval $\left \{ l , \dots, l + w - 1 \right \}$ is recorded and the search is recursively repeat to the left and right of this interval. Pseudo code for the procedure is given below in Algorithm \ref{algorithm: binary segmentation over grids}. In the pseudo code given integers $s$ and $e$ which satisfy $1 \leq s < e \leq n$ we write $\mathcal{G}_{s,e} \left ( W, a \right )$ for the set of $(l,w)$ pairs in $\mathcal{G} \left ( W , a \right )$ which can be associated with an interval satisfying $\left \{ l, \dots, l + w - 1\right \} \subseteq \left \{ s, \dots, e \right \}$. We write $\lambda_\alpha$ for either of the thresholds (\ref{equation: Gaussian FWE control lambda}) or (\ref{equation: generic FWE control lambda}), depending on whether we are operating under Assumptions \ref{assumption: gaussian noise} - \ref{assumptioin: logarithmic segment length} or Assumptions \ref{assumption: positive long run varaince} - \ref{assumption: minimum segment}. Finally we write $\widehat{\tau}$ for a generic estimator of the (long run) standard deviation of the noise which satisfies either the of the conditions in of part (ii) of Theorem \ref{theorem: Gaussian tightness of normalized maximum} or in part (ii) of Theorem \ref{theorem: tightness of non-gaussian dependent maximum}, depending on the set of assumptions we are operating under. 

\begin{algorithm}[H]
\SetKwFunction{greedyIntervalSearch}{greedyIntervalSearch}
\SetKwFunction{RecordInterval}{RecordInterval}
\SetKwProg{Fn}{function}{:}{\KwRet}

\Fn{\greedyIntervalSearch{$\boldsymbol{Y}, s, e$}}{

\If{$e-s < \min \left ( W, p+1 \right )$}
{STOP}
detection $\leftarrow$ \texttt{False} \\
\For{$(l,w)$ in $\mathcal{G}_{s,e} \left ( W, a \right )$}
{
\If{$\left | D_{l,w}^p \left ( \boldsymbol{Y} \right ) \right | > \widehat{\tau} \lambda_\alpha$}
{
\RecordInterval{$l,w$} \\
\greedyIntervalSearch{$Y, s, l$} \\ 
\greedyIntervalSearch{$Y, l+w-1, e$} \\
detection $\leftarrow$ \texttt{True}
}
\If{detection}
{
BREAK
}
}
}
\caption{The greedy interval search algorithm for change point inference in piecewise polynomials. Given an appropriate threshold, the algorithm returns a collection of mutually disjoint intervals which each must contain a change point uniformly with probability at least $1-\alpha + o(1)$.}
\label{algorithm: binary segmentation over grids}
\end{algorithm}

A consequence of using thresholds (\ref{equation: Gaussian FWE control lambda}) and (\ref{equation: generic FWE control lambda}) in Algorithm \ref{algorithm: binary segmentation over grids} is that with no assumptions on the number of change points in the data or their spacing, with high probability, every interval returned is guaranteed to contain at least one change point. The number of intervals returned therefore functions as an assumption free lower bound on the number of change points in the data. This behaviour is summarised in Corollary \ref{corollary: coverage guarantee} below. 

\begin{corollary}
Let $\hat{I}_1, \dots, \hat{I}_{\hat{N}}$ be intervals returned by Algorithm \ref{algorithm: binary segmentation over grids}. On a set with probability asymptotically larger than $1 - \alpha$ the following events occur simultaneously:
\begin{gather*}
E_1^* = \left \{ \widehat{N} \leq N \right \} \\
E_2^* = \left \{ \hat{I}_k \cap \Theta \neq \emptyset \mid k = 1, \dots, \hat{N} \right \}.
\end{gather*}
\label{corollary: coverage guarantee}
\end{corollary}

Although the coverage guarantee provided by Corollary \ref{corollary: coverage guarantee} is asymptotic in nature, in practice we find that Algorithm \ref{algorithm: binary segmentation over grids} provides accurate coverage in finite samples, and in fact tends to deliver over coverage; see the simulation results in Section \ref{section: coverage simulations} and in Section \ref{section: additional simulation study}. The thresholds proposed for Algorithm \ref{algorithm: binary segmentation over grids} rely on extreme value asymptotics. For such results the convergence rate is often slow, and indeed we conjecture that the convergence rate for our procedure is no better than $\mathcal{O} \left ( \log^{-1} (n)\right ) $. However, it is also known that that the extreme value approximation for upper quantiles works well even for moderate sample sizes; confer for instance \cite[Figure 1.3.1]{csorgo1997limit}, \cite[Section 2.4]{leadbetter2012extremes}, and \cite{aue2009extreme, aue2008testing}.  

The worst case run time of Algorithm \ref{algorithm: binary segmentation over grids} is always of the order $\mathcal{O} \left ( n \log (n) \right )$, independent of the number of change points in the data, their spacing, and the polynomial degree of the signal. This is because the worst case run time will be attained when a test has to be carried out for every $(l,w)$ pair in the grid $\mathcal{G} \left ( W, a \right )$. However, for any fixed $a > 1$ the the grid contains at most of the order $\mathcal{O} \left ( n \log (n) \right )$ such pairs, and by first calculating all cumulative sums of the data, which can be done in $\mathcal{O} \left ( n \right )$ time, each local test can be carried out in constant time.

We finally remark that many existing procedures for change point detection make use of thresholds which involve unknown constants other than the scale of the noise. In general these constants are either chosen sub-optimally, or calibrated via Monte Carlo. See for instance the implementation of \cite{verzelen2020optimal} by \cite{liehrmann2023ms} for an example in in the piecewise constant setting, and the discussion on the practical selection of tuning parameters in \cite{kim2022moving} for an example in the piecewise linear setting. Meanwhile, the thresholds used in Algorithm \ref{algorithm: binary segmentation over grids} are the sharpest possible, and do not rely on any unknown constants other than the scale of the noise. 

\subsection{Variance and long run variance estimation} \label{section: varaince and LRV estimation}

In general the scale of the noise will not be known, and to make Algorithm \ref{algorithm: binary segmentation over grids} operational the (long run) standard deviation of the noise will need to be estimated consistently, according to the conditions given in part (ii) of either Theorem \ref{theorem: Gaussian tightness of normalized maximum} or Theorem \ref{theorem: tightness of non-gaussian dependent maximum}. In this section we give several strategies for consistently estimating the noise level in the presence of an unknown piecewise polynomial signal. 

\subsubsection{Variance estimation under Gaussian noise}

In change point problems where the noise is independently distributed, homoskedastic, and Gaussian the standard deviation is commonly estimated using the median absolute deviation (MAD) estimator \cite{hampel1974influence}. To account for the unknown piecewise polynomial signal we propose to use the following generalisation of the MAD estimator based on the $(p+1)$-th difference of the data. Letting $X_{p+2},\dots, X_n$ be the $(p+1)$-th difference of the sequence $Y_1, \dots, Y_n$ the estimator is defined as follows:
\begin{equation}
\widehat{\sigma}_{\text{MAD}} = \frac{\text{median} \left \{ \left | X_{p+2} \right |, \dots, \left | X_n \right | \right \}}{\Phi^{-1} \left ( 3 / 4 \right ) \sqrt{\sum_{j=0}^{p+1}\binom{p+1}{j}^2}}.
\label{equation: MAD estimator}
\end{equation}

As shown by the following lemma, when the assumptions of Theorem \ref{theorem: Gaussian tightness of normalized maximum} hold the modified MAD estimator satisfies the condition in part (ii) of the Theorem \ref{theorem: Gaussian tightness of normalized maximum} as long as the number of change points grows more slowly than $n / \log (n)$. 

\begin{lemma}
If the noise terms are independently distributed and Gaussian with common variance $\sigma^2$ it holds that $\left | \widehat{\sigma}_\text{MAD} - \sigma \right | = \mathcal{O}_\mathbb{P} \left ( \frac{1}{\sqrt{n}} \vee \frac{N}{n} \right )$. 
\label{lemma: MAD consistency}
\end{lemma}

\subsubsection{Variance estimation under non-Gaussian noise}

For variance estimation under independently distributed light tailed homoskedastic noise, difference based estimators are often used \citep{dumbgen2001multiscale, rice1984bandwidth, gasser1986residual}. To account for the unknown piecewise polynomial signal we propose to use the following estimator based on the $(p+1)$-th difference of the data sequence. The estimator is defined as follows:
\begin{equation}
\widehat{\sigma}^2_{\text{DIF}} = \frac{1}{n-(p+1)}\sum_{t=p+2}^n \left \{ \frac{X_t^2}{\sum_{j=0}^{p+1} \binom{p+1}{j}^2} \right \}.
\label{equation: difference based var estimator}
\end{equation}

As shown by the following lemma, under some mild conditions on signal component the difference based estimator satisfies condition (ii) in Theorem \ref{theorem: tightness of non-gaussian dependent maximum} as long as the number of change points again grows more slowly than $n / \log (n)$. 

\begin{lemma}
If the function $f_\circ \left ( \cdot \right )$ is bounded and the noise terms are independently distributed with common variance $\sigma^2$ and bounded fourth moments it holds that $\left | \widehat{\sigma}_\text{DIF}^2 - \sigma^2 \right | = \mathcal{O}_\mathbb{P} \left ( \frac{1}{\sqrt{n}} \vee \frac{N}{n} \right )$. 
\label{lemma: consistency of diff-sd estimator}
\end{lemma}

\subsubsection{Long-run variance estimation}

For estimating the long run variance we extend the estimator proposed in \cite{wu2007inference}, based on first order differences of local sums of the data, to $(p+1)$-th differences. To form the estimator we choose a scale $W'$, which is not necessarily related to any of the scales in the grid (\ref{equation: a-adic grid}), and form the following local sums:
\begin{equation}
\bar{Y}_{t,W'} = Y_{(t-1)W' + 1} + \dots + Y_{tW'}, \hspace{2em} t = 1, \dots, \left \lfloor n / W' \right \rfloor
\label{equation: LRV local sums}
\end{equation}
Then, putting $\bar{X}_{p+2,W'}, \dots, \bar{X}_{\left \lfloor n / W' \right \rfloor,W'}$ for the $(p+1)$-th difference of the sequence of $\bar{Y}_{W'}$'s, the estimator is defined as follows:
\begin{equation}
\widehat{\tau}^2_{\text{DIF}} = \frac{1}{\left \lfloor n / W' \right \rfloor - (p+1)} \sum_{t=p+2}^{\left \lfloor n / W' \right \rfloor} \left \{ \frac{\bar{X}_{t,W'}^2}{W' \sum_{i=0}^{p+1} \binom{p+1}{i}^2} \right \}.
\label{equation: long run varaince estimator}
\end{equation}
In order to show consistency of our long run variance estimator we need to impose the following assumption, which states that the sequence of auto-covariances for the noise decay sufficiently fast and can be estimated well from a finite sample. 
\begin{assumption}
The auto-covariances decay fast enough that $\sum_{h>1} h \left | \gamma_h \right | < \infty$, and for any fixed integer $h$ and any ordered subset of $\left \{ 1, \dots, n - h \right \}$, say $M$, it holds that $\left | M \right | ^ {-1} \sum_{t \in M} \zeta_t \zeta_{t+h} = \gamma_h + \mathcal{O}_\mathbb{P} \left ( 1 / \sqrt{\left | M \right |} \right )$. 
\label{assumtion: ACVF}
\end{assumption}
With the above assumption in place, we have the following guarantee on the consistency of the estimator. 
\begin{lemma}
If the function $f_\circ \left ( \cdot \right )$ is bounded and the noise terms satisfy Assumption \ref{assumption: positive long run varaince} and Assumption \ref{assumtion: ACVF} it holds that $\left | \hat{\tau}_{\text{DIF}}^2 - \tau^2 \right | = \mathcal{O}_{\mathbb{P}} \left ( \frac{W'}{\sqrt{n}} \vee \frac{1}{W'} \vee \frac{NW'^2}{n} \right )$. 
\label{lemma: consistency of LRV estimator}
\end{lemma}
Lemma \ref{lemma: consistency of LRV estimator} shows that if, for example, $W'$ is chosen to be of the order $W' = \mathcal{O} \left ( n ^ \theta \right )$ for some $\theta < 1/2$ then (\ref{equation: long run varaince estimator}) satisfies the condition in part (ii) of Theorem \ref{theorem: tightness of non-gaussian dependent maximum} as long as the number of change points grows more slowly than $n^{1-2\theta} \log^{-1} \left ( n / W \right )$. In practice we follow \cite{wu2007inference} in setting $W' = n ^ {1/3}$. 

\subsection{Consistency of the algorithm}

We now investigate the conditions under which algorithm Algorithm \ref{algorithm: binary segmentation over grids} is consistent, in the sense that with high probability it is able to detect all change points and returns no spurious intervals. It is useful to parameterise the signal in model (\ref{equation: signal + noise}) between change point locations as follows:
\begin{equation}
f_\circ \left ( t / n \right ) = 
\begin{cases}
\sum_{j=0}^p \alpha_{j,k} \left ( t/n - \eta_k / n \right )^j & \text{ if } \eta_{k-1} < t \leq \eta_k \\
\sum_{j=0}^p \beta_{j,k} \left ( t/n - \eta_k / n \right )^j & \text{ if } \eta_{k} < t \leq \eta_{k+1} 
\end{cases}
\hspace{2em} k = 1, \dots, N.
\label{equation: signal between change}
\end{equation}

Therefore, the absolute change in the $j$-th derivative of $f_\circ(\cdot)$ at the $k$-th change point location can be written as $\Delta_{j,k} = \left | \alpha_{j,k} - \beta_{j,k} \right |$. Putting $\eta_0 = 0$ and $\eta_{N+1} = n$ we write $\delta_k = \min \left ( \eta_{k} - \eta_{k-1}, \eta_{k+1} - \eta_{k} \right )$ for the effective sample size associated with the $k$-th change location. The most prominent change in derivative at each change point location can therefore be defined as follows:

\begin{equation}
p^*_k \in \argmax_{0 \leq j \leq p} \left \{ \Delta_{j,k} \left ( \frac{\delta_k}{n} \right ) ^ j \right \} \hspace{2em} k = 1, \dots, N.
\label{equation: most prominent change}
\end{equation}

In order to show the consistency of Algorithm \ref{algorithm: binary segmentation over grids} we impose two restriction on the signal. The first states that the changes in derivative at each change point location are bounded. The second states that although multiple changes in the derivatives of $f_\circ \left ( \cdot \right )$ can occur at each change point location, there is always one dominating change. This excludes the possibility of signal cancellation occurring. 

\begin{assumption}
There is a constant $C_\Delta > 0$ such that $\left | \Delta_{jk} \right | < C_\Delta$ for each $j,k$. 
\label{assumption: bounded changes}
\end{assumption}

\begin{assumption}
For each $k = 1, \dots, N$ the quantity $p^*_k$ is uniquely defined, and for any sequence $\left ( \rho_{k,n} \right )_{n \geq 1}$ with the property $\rho_{k,n} \leq \delta_k / n$ for all $n \geq 1$ it holds that $\left | \Delta_{j,k} \right | \rho_{k,n}^j \leq C_{p^*_k} \left | \Delta_{p^*_k,k} \right | \rho_{k,n}^{p^*_k} $ for all $j \neq p_k^*$, where $C_{p^*_k} = \frac{1}{2^{p_k^*+2}(p^*+1)p}$. 
\label{assumption: One prominent jump}
\end{assumption}

For example, Assumption \ref{assumption: One prominent jump} would be violated by the piecewise linear signal shown in (\ref{equation: cancelation example}) for which $n = 8$ and $\eta = 4$, and the scaled difference in slopes between the first four entries and the last four had the same magnitude but the opposite sign to the corresponding difference in levels. That is: $\Delta_0 = \Delta_1 \left ( \delta / n \right )$.  
\begin{equation}
\mathbf{f} = \left ( -7/8, -6/8, -5/8, -4/8, 3/8, 2/8, 1/8, 0 \right ) ' 
\label{equation: cancelation example}
\end{equation}
In practice, in situations when Assumption \ref{assumption: One prominent jump} is violated our procedure is still able to detect the corresponding change point. This is because although signal cancellation such as in (\ref{equation: cancelation example}) may occur on a particular interval considered by Algorithm \ref{algorithm: binary segmentation over grids}, it is unlikely to occur on every interval considered. In the above example, if we were to look at the sub-vector $\left ( -5/8, -4/8, 3/8, 2/8 \right )'$ no cancellation would occur. See also Remark \ref{remark: assumption 3} in the Proofs section, where we show how the assumption can be relaxed for piecewise linear functions, and show good practical performance via simulation on higher order piecewise polynomials which violate the assumption. With these assumptions in place we have the following result.

\begin{theorem}
Let $\hat{I}_1, \dots, \hat{I}_{\hat{N}}$ be intervals returned by Algorithm \ref{algorithm: binary segmentation over grids} run on data $\boldsymbol{Y} = \left ( Y_1, \dots, Y_n \right )'$ from model (\ref{equation: signal + noise}), with parameters $a>1$, $W$, and $\alpha \in (0,1)$. Grant Assumptions \ref{assumption: bounded changes}-\ref{assumption: One prominent jump} and either of Assumptions \ref{assumptioin: logarithmic segment length}-\ref{assumption: gaussian noise} or \ref{assumption: positive long run varaince}-\ref{assumption: minimum segment} hold, and let the threshold $\lambda_\alpha$ chosen according to (\ref{equation: Gaussian FWE control lambda}) or (\ref{equation: generic FWE control lambda}) accordingly. If the the effective sample size at each change point location satisfies
\begin{equation}
\delta_k > C_1 \left ( W \vee n^{\frac{2p^*_k}{2p^*_k + 1}} \left ( \frac{\tau^2 \log (n)}{\Delta_{p^*_k,k}^2} \right ) ^ {\frac{1}{2p^*_k + 1}} \right ) \hspace{2em} k = 1, \dots, N
\label{equation: spacing condition}
\end{equation}
then on a set with probability $a - \alpha + o(1)$ the following events occur simultaneously:
\begin{gather*}
E^*_3 = \left \{ \hat{N} = N \right \} \\ 
E^*_4 = \left \{ \hat{I}_k \cap \Theta = \left \{ \eta_k \right \} \mid k = 1, \dots, N \right \} \\
E^*_5 = \left \{ \left | \hat{I}_k \right | \leq C_2 \left ( W \vee n^{\frac{2p^*_k}{2p^*_k + 1}} \left ( \frac{\tau^2 \log (n)}{\Delta_{p^*_k,k}^2} \right ) ^ {\frac{1}{2p^*_k + 1}} \right ) \big \vert k =  1, \dots, N \right \}. 
\end{gather*}
Here $C_1$ and $C_2$ depend only on $\alpha$, $a$ and $p$.
\label{theorem: optimality of change point detection}
\end{theorem}

Theorem \ref{theorem: optimality of change point detection} states that on a set with probability asymptotically larger than $1 - \alpha$, where $\alpha$ can be tuned by the user, the number of intervals returned by Algorithm \ref{algorithm: binary segmentation over grids} coincides with the true number of change points (event $E_3^*$), and every interval returned contains exactly one change point (event $E_4^*$). Event $E_5^*$ provides bounds on the widths of intervals returned, which in turn implies a bound on the localisation rate of any change point estimator which lies within a given interval returned by the algorithm. 

Theorem \ref{theorem: optimality of change point detection} leads to the following large sample consistency result. 

\begin{corollary}
Let $\hat{I}_1, \dots, \hat{I}_{\hat{N}}$ be intervals returned by Algorithm \ref{algorithm: binary segmentation over grids} under the same conditions as Theorem \ref{theorem: optimality of change point detection} but with threshold $\lambda = \left ( 1 + \varepsilon \right ) a_{W,n}$ for some fixed $\varepsilon > 0$, where $a_{W,n}$ is as defined in Theorem \ref{theorem: tightness of non-gaussian dependent maximum}. Then on a set with probability $1 - o(1)$ the events $E_1^*$, $E_2^*$, and $E_3^*$ occur simultaneously. 
\label{corollary: large sample consistency}
\end{corollary}

An important consequence of Theorem \ref{theorem: optimality of change point detection} and Corollary \ref{corollary: large sample consistency} is that any point-wise estimator $\hat{\eta}_k$ for the $k$-th change point location which lies in an interval $\hat{I}_k$ will inherit the localisation rate implied by event $E^*_5$. As explained in Section \ref{section: Optimality of the algorithm} this rate is unimprovable in a minimax sense. This extends to the naive estimator formed by setting $\hat{\eta}_k$ to the midpoint of the interval $\hat{I}_k$. However, more sophisticated estimators can be used; for example one may choose $\hat{\eta}_k$ to be the split point which results in the lowest sum of squared residuals when a piecewise polynomial function is fit over $\hat{I}_k$ (see for example Figure \ref{figure: piecewise polynomials examples}).

\subsection{Optimality of the algorithm} \label{section: Optimality of the algorithm}

In \cite{yu2022localising, yu2020review} it was shown that, under independent sub-Gaussian noise with Orlicz-$\psi_2$ norm bounded from above by some $\omega^2$, the mini-max localisation rate for each change point in the generic piecewsie polynomial model is of the order
\begin{equation}
\mathcal{O} \left ( n^{\frac{2p^*_k}{2p^*_k + 1}} \left ( \frac{\omega^2}{\Delta_{p^*_k,k}^2} \right ) ^ {\frac{1}{2p^*_k + 1}} \right ), \hspace{2em} k = 1, \dots, N.
\label{equation: minimax localisation rate}
\end{equation}
Examining the proof of Lemma 2 in \cite{yu2022localising} one can see that the same rate holds for weakly dependent noise by replacing the sub-Gaussian parameter $\omega^2$ with the long run variance $\tau^2$. Therefore, under Assumptions \ref{assumption: gaussian noise}- \ref{assumptioin: logarithmic segment length} where $W$ is of the order $\mathcal{O} \left ( \log (n) \right )$, the bounds guaranteed by $E_5^*$ can be seen to be optimal up to to log terms. That is, the width of each interval returned matches (up to log terms) the best possible rate at which the corresponding change point can be localised. Under Assumptions \ref{assumption: positive long run varaince}-\ref{assumption: minimum segment}, where $W$ grows slightly faster than $n^{2/(2+\nu)}$, the bounds provided by event $E_5^*$ are again optimal as long as $\nu > 1$ and the most prominent change occurs in derivatives of order $1$ or higher. However, whenever $p^*_k = 0$ comparing to (\ref{equation: minimax localisation rate}) it is clear the bounds are no longer optimal. 

The aforementioned lack of optimality is due to Assumption \ref{assumption: minimum segment}, which requires the minimum support of our detection statistic to be relatively larger. This is needed in order that a strong approximation result may be invoked for a range of noise distributions. However, the requirement that $W$ grows at a polynomial rate with $n$ can be overly conservative. For example, if the noise terms are independently distributed with finite moment generating function in a neighbourhood of zero, which is the setting studied by \cite{yu2022localising, yu2020review}, then Theorem 1 in \cite{komlos1975approximation} states that after enlarging the probability space
\begin{equation*}
\sum_{t=1}^n \zeta_t - \tau B (n) = \mathcal{O} \left ( \log (n) \right ), \hspace{2em} \mathbb{P}\text{-almost surely}.
\end{equation*}
Consequently, in this setting the results of Theorem \ref{theorem: optimality of change point detection} continue to hold with $W$ of the order $o \left ( \log^3 (n) \right )$. In which case, setting $\lambda_\alpha$ accordingly, the bound provided by event $E^*_5$ again results optimal up to the log factors.

The width of the $k$-th interval depends (up to constants) only on the order of the derivative at which the most prominent change occurs, and not on the overall polynomial degree of the signal. This shows the intervals adapt locally to the smoothness of the signal. Interestingly the rate $\mathcal{O} \left ( n^{2p^*/(2p^*+1)} \right )$ is the same as the mini-max bound on the sup-norm risk for $p^*$-smooth Holder regression functions \cite[Theorem 2.10]{tsybakov2004introduction}. The error probability $\alpha$ does not appear explicitly in Theorem \ref{theorem: optimality of change point detection} as it is absorbed into the constants $C_1$ and $C_2$. Indeed for different but fixed choices of $\alpha$ all thresholds constructed according to the rules discussed in Sections \ref{section: Family-wise error control under Gaussianity} and \ref{section: extension to dependent and non-Gaussian noise} will be asymptotically equivalent. However in finite samples there is a clear price to pay for requesting higher coverage since as $\alpha \downarrow 0$ we have that $-2 \log ^{-1} \left ( 1 - \alpha \right ) \sim 2 / \alpha$. 

Finally, the effect of the degree of serial dependence on the lengths of the intervals is explicit, as the long run variance of the noise appears in the upper bound on the interval lengths. The nature of this dependence is similar to that found by \cite{enikeeva2020bump} in the simpler problem of detecting a bump in the mean function of a stationary Gaussian process. 

\subsection{On the polynomial order of the signal}

We emphasise that in the problem statement $p$ refers to the maximum polynomial order of the signal on any stationary segment, and that the polynomial order of the signal is permitted to vary between segments. If $p$ is unknown, it should be considered as an input to our algorithm. However, provided this input is chosen large than or equal to the maximum polynomial order, it only affects the output in terms of constants and not rates. Of course, in a finite sample there is a price to pay: choosing $p$ larger leads to longer intervals through inflating the constant $C_2$, and changes the change point detection condition through inflating the constant $C_1$.

We observe that in applications analysts usually have in mind a reasonable idea of $p$ motivated by knowledge of the problem at hand. However, it may be unreasonable to assume that the maximum polynomial order is known exactly. Therefore, we present two methods for determining $p$ from data given upper and lower bounds $\underline{p}$ and $\overline{p}$ such that $p \in \left \{ \underline{p}, \dots, \overline{p} \right \}$. The methods are designed for the setup in Sections \ref{section: Family-wise error control under Gaussianity} and \ref{section: extension to dependent and non-Gaussian noise} respectively. 

\subsubsection{Estimating $p$ via the  strengthened Schwarz Information Criterion}

\cite{fryzlewicz2014wild} introduced the  strengthened Schwarz Information Criterion (sSIC) for consistently estimating the number of change points in the canonical change point model for which the signal is piecewise constant and the contaminating noise is independently distributed and Gaussian. The same approach can be extended to estimating $p$ in the piecewise polynomial model. 

Given data $\boldsymbol{Y} = \left ( Y_1, \dots, Y_n \right )'$ from model (\ref{equation: signal + noise}) and some $p' \in \left \{ \underline{p}, \dots, \overline{p} \right \}$ let $\hat{I}_1, \dots, \hat{I}_{\hat{N}_{p'}}$ be the output of Algorithm \ref{algorithm: binary segmentation over grids} under the assumption that the maximum polynomial degree is $p'$, run with threshold $\lambda = (1+\varepsilon) \mathfrak{a}_{W,n}$ for some fixed $\varepsilon > 0$. Let $\hat{\eta}_1, \dots, \hat{\eta}_{\hat{N}_{p'}}$ be the split points within each interval associated with the piecewsie polynomial fit providing the lowest sum of squared residuals and let $\hat{f}_{p'} \left ( \cdot \right )$ be the function estimated via least squares between these knots. Following Section 3.4 in \cite{fryzlewicz2014wild} for some arbitrary but fixed $\alpha > 1$ the sCIC at $p'$ is defined as
\begin{equation*}
\text{sSIC} \left ( p' \right ) = \frac{n}{2} \log \left ( \hat{\sigma}^2_{p'} \right ) + ( \hat{N}_{p'} + 1 ) \left ( p' + 1 \right ) \log^\alpha \left ( n \right ), 
\end{equation*}
where in particular 
\begin{equation*}
\hat{\sigma}^2_{p'} = \frac{1}{n} \sum_{t=1}^n \left ( Y_t - \hat{f}_{p'} (t/n) \right )^2.
\end{equation*}
Then, the maximum polynomial degree of the signal can be estimated as
\begin{equation}
\hat{p} = \argmin_{\underline{p} \leq p' \leq \overline{p}} \text{sSIC} \left ( p' \right ). 
\label{equation: sSIC estimator}
\end{equation}
Regarding the large sample consistency of $\hat{p}$, we have the following result. 

\begin{lemma}
Let $\hat{p}$ be the estimator defined in (\ref{equation: sSIC estimator}). Grant Assumptions \ref{assumptioin: logarithmic segment length} and \ref{assumption: gaussian noise} as well as condition (\ref{equation: spacing condition}) hold, and moreover assume moreover that: (i) $\underline{p} \leq p \leq \overline{p}$ and $(\overline{p} - \underline{p}) = \mathcal{O} (1)$, (ii) $N = \mathcal{O}(1)$, and (iii) the coefficients in (\ref{equation: signal between change}) are all of the order $\mathcal{O}(1)$. Then $\mathbb{P} \left ( \hat{p} = p \right ) \rightarrow 1$ as $n \rightarrow \infty$. 
\label{lemma: sSIC consistency}
\end{lemma}

\subsubsection{Estimating $p$ via recursive testing on null intervals}

The finite difference functional which has so far been used to test for the presence of a change point can itself be used to estimate the maximum degree of the signal. For some $p' \in \left \{ \underline{p}, \dots, \overline{p} \right \}$ let $K$ be a contiguous subset of $\left \{ 1, \dots, n \right \}$ for which $\left | K \right |$ is a multiple of $(p'+2)$. Therefore, introduce the statistic
\begin{equation}
D_{K}^{p'} \left ( \boldsymbol{Y} \right ) = \left \{ \left \lfloor \frac{ \left | K \right |}{p'+2} \right \rfloor \sum_{i=0}^{p'+1} \binom{p'+1}{i}^2 \right \}^{-1/2} \sum_{j=0}^{p'+1} \left ( - 1 \right ) ^ {p'+1-j} \binom{p'+1}{j} \bar{Y}^j_{K}
\label{equation: K local statistic}
\end{equation}
where in particular letting $K$ have elements $\left \{ k_1, \dots, k_{\left | K \right |} \right \}$ we write
\begin{equation*}
\bar{Y}_K^j = Y_{k_1 + j \frac{\left | K \right |}{p'+2}} + \dots + Y_{ (j+k_1) \frac{\left | K \right |}{p'+2}}, \hspace{2em} j = 0, \dots, p'+1
\end{equation*}
for non-overlapping sums of the data over the $(p'+2)$ equally sized contiguous partitions of $K$. Note that if $K$ corresponds to a stretch of data which contains no change points and $p' < p$ then (\ref{equation: K local statistic}) will be large in (absolute) expectation, whereas if $p' \geq p$ then (\ref{equation: K local statistic}) will be small. 

Using the above intuition, to estimate $p$ we first run Algorithm \ref{algorithm: binary segmentation over grids} with threshold $\lambda = (1 + \varepsilon) \mathfrak{a}_{W,n}$ for some small but fixed $\varepsilon > 0$ assuming the maximum polynomial order of the signal is $\overline{p}$. We then obtain sets $\widehat{\mathbb{K}} = \{ \hat{K}_1, \hat{K}_2, \dots \}$ by retaining indices \textit{between} each interval returned, and trimming either the first or last few indices so that the number of elements in each $\hat{K}$ is a multiple of $(\overline{p}+1)$. Note that since $\overline{p} \geq p$ by Corollary \ref{corollary: large sample consistency} with high probability each $\hat{K}$ corresponds to a stretch of data which contains no change points. Finally we test whether $ | D_{\hat{K}}^{\overline{p}-1} \left ( \boldsymbol{Y} \right ) | > (1 + \varepsilon) \mathfrak{a}_{W,n}$ for each $\hat{K}$. If any such test is not passed we conclude that $p = \overline{p}$. Else, we repeat the procedure with $\overline{p} -1$. The procedure automatically ends once $\underline{p}$ is reached, since we assume $p \geq \underline{p}$, and by this point we have concluded that $p < p'$ for all $p' > \underline{p}$. The procedure is sumarized in Algorithm \ref{algorithm: p estimation}. Regarding the large sample consistency of the output of Algorithm \ref{algorithm: p estimation} we have the following result.  

\begin{algorithm}[!htbp]
\SetKwFunction{maxDegreeEstimation}{maxDegreeEstimation}
\SetKwProg{Fn}{function}{:}{\KwRet}

\Fn{\maxDegreeEstimation{$\boldsymbol{Y}, \overline{p}, \underline{p}$}}{
$p' \leftarrow \overline{p}$ \\
Detection $\leftarrow$ \texttt{False} \\

\While{$p' > \underline{p}$}{
Obtain intervals $\hat{\mathbb{K}} = \{ \hat{K}_1, \hat{K}_2, \dots \}$ from Algorithm \ref{algorithm: binary segmentation over grids} 
using \\ threshold $\lambda = (1 + \varepsilon) \mathfrak{a}_{W,n}$ and assuming maximum degree $p'$.\\
\For{$\hat{K} \in \hat{\mathbb{K}}$}
{
\If{$ | D_{\hat{K}}^{p'-1} \left ( \boldsymbol{Y} \right ) | > (1 + \varepsilon) \mathfrak{a}_{W,n}$}
{
Detection $\leftarrow$ \texttt{True}
}
}
\If{{\upshape Detection}}
{
BREAK
}
$p' \leftarrow (p'-1)$
}
}
\caption{An algorithm for determining the maximum polynomial order of the signal by progressively estimating intervals of significance and testing null intervals for the presence of a change points in a lower degree polynomial.}
\label{algorithm: p estimation}
\end{algorithm}

\begin{lemma}
Let $\hat{p}$ be the output of Algorithm \ref{algorithm: p estimation}. Grant Assumptions \ref{assumption: positive long run varaince}, \ref{assumption: strong approximation}, and \ref{assumption: minimum segment} as well as condition (\ref{equation: spacing condition}) hold, and moreover assume moreover that: (i) $\underline{p} \leq p \leq \overline{p}$ and $(\overline{p} - \underline{p}) = \mathcal{O} (1)$, (ii) $N = \mathcal{O}(1)$, and (iii) the coefficients in (\ref{equation: signal between change}) are all of the order $\mathcal{O}(1)$. Then $\mathbb{P} \left ( \hat{p} = p \right ) \rightarrow 1$ as $n \rightarrow \infty$. 
\label{lemma: p algo consistency}
\end{lemma}

\section{Simulation studies} \label{section: numerical illustrations}

\subsection{Alternative methods for change point inference} \label{section: alternative methods}

We will compare our proposed methodology with existing algorithms with publicly available implementations, which each promise to return intervals containing true change point locations uniformly at a significance level chosen by the user. These are: the Narrowest Significance Pursuit (NSP) algorithm of \cite{fryzlewicz2023narrowest}, its self-normalised variant (NSP-SN), and its extension to auto-regressive signals (NSP-AR); the bootstrap confidence intervals for moving sums (MOSUM) of \cite{cho2022bootstrap} using a single bandwidth (uniscale) and multiple bandwidths (multiscale); the simultaneous multiscale change point estimator (SMUCE) of \cite{frick2014multiscale}, as well as its extension to heterogeneous noise (H-SMUCE) developed by \cite{pein2017heterogeneous}, and its extension to dependent noise (Dep-SMUCE) developed by \cite{dette2020multiscale}. We also consider the conditional confidence intervals of \cite{bai1998estimating} (B\&P) with significance level Bonferroni-corrected for the estimated number of change-points. For our own procedure we write DIF1 for Algorithm \ref{algorithm: binary segmentation over grids} run under the assumptions of Theorem \ref{theorem: Gaussian tightness of normalized maximum} and DIF2 for the algorithm run under the assumption of Theorem \ref{theorem: tightness of non-gaussian dependent maximum}. Additionally we write MAD if the scale of the noise is estimated using the median absolute deviation estimator (\ref{equation: MAD estimator}), SD if the scale is estimated using the difference based estimator of the standard deviation (\ref{equation: difference based var estimator}), and LRV if the long run variance is estimated using (\ref{equation: long run varaince estimator}). Each of the methods considered is designed for different noise types and different change point models, and we summarise this information in Table \ref{table: methods proprties} below. 

\begin{table}[!htbp]
\centering
\caption{Suitability of each method to non-Gaussian noise, dependent noise, and change point detection in higher order polynomial signals. The the letter \textbf{e} indicates that no theoretical guarantees are given but the authors observe good empirical performance of the method.}
\begin{tabular}{|l|c|c|c|}
\hline
Method & \makecell{non-Gaussian \\ noise} & \makecell{dependent \\ noise} & \makecell{higher order \\ polynomials} \\ 
\hline
DIF1-MAD & \xmark & \xmark & \cmark \\
DIF2-SD & \cmark & \xmark & \cmark \\
DIF2-LRV & \cmark & \cmark & \cmark \\
NSP & \xmark & \xmark & \cmark \\
NSP-SN & \cmark & \xmark & \cmark \\
NSP-AR & \xmark & \cmark & \cmark \\ 
B\&P & \cmark & \xmark & \xmark \\ 
MOSUM (uniscale) & \cmark & \xmark & \xmark \\
MOSUM (multiscale) & \cmark & \xmark & \xmark \\
SMUCE & \xmark & \xmark & \xmark \\
H-SMUCE & \textbf{e} & \xmark & \xmark \\
Dep-SMUCE & \cmark & \cmark & \xmark \\
\hline 
\end{tabular}    
\label{table: methods proprties}
\end{table}

Throughout the simulation studies, whenever a method requires the user to specify a minimum support parameter we set this to $W = 0.5 n ^ {1/2}$. Exceptions occur for Dep-SMUCE for which we follow the authors' recommendation in setting $W = n ^ {1/3}$, for DIF1-MAD in which we set $W = \log (n)$ following the results of Theorem \ref{theorem: Gaussian tightness of normalized maximum}, and for the multiscale MOSUM procedure for which we generate a grid of bandwidths using the \texttt{bandwidths.auto} function in the MOSUM package \cite{meier2021mosum}. For our own procedure we set the decay parameter regulating the density of the grid to $a = \sqrt{2}$ as was done in \cite{kovacs2020seeded} for the grid proposed therein. 

\subsection{Coverage on null signals} \label{section: coverage simulations}

We first investigate empirically the coverage provided by our algorithm and the alternatives introduced in Section \ref{section: alternative methods}. To investigate coverage we apply each method to a vector of pure noise with length $n = 750$ generated according to each of the noise types listed below, setting the noise level to $\sigma = 1$ , and over $100$ replications record the proportion of times no intervals of significance are returned. For each procedure we set appropriate tuning parameters in order that the family-wise error is nominally controlled at the level $\alpha = 0.1$. Where applicable we ask each procedure to test for change points in polynomial signals of degrees $0$, $1$, and $2$.

\begin{itemize}
    \item (\texttt{N1}): $\zeta_t \sim \mathcal{N} ( 0, \sigma^2 ) \text{ i.i.d.}$
    \item (\texttt{N2}): $\zeta_t \sim t_5 \times \sigma \sqrt{0.6} \text{ i.i.d.}$
    \item (\texttt{N3}): $\zeta_t \sim \sigma \times \text{Laplace} (0, 1/\sqrt{2}) \text{ i.i.d.}$ 
    \item (\texttt{N4}): $\zeta_t = \phi \zeta_{t-t} + \varepsilon_t$ with $\phi = 0.8$ and $\varepsilon_t \sim \mathcal{N} ( 0, \sigma^2 / (1 - \phi^2)$ i.i.d.
    \item (\texttt{N5}): $\zeta_t = \phi \zeta_{t-t} + \varepsilon_t$ with $\phi = 0.8$ and $\varepsilon_t \sim t_5 \times \sigma \sqrt{0.6/(1-\phi^2)}$ i.i.d.
    \item (\texttt{N6}): $\zeta_t = \phi_1 \zeta_{t-1} + \phi_2 \zeta_{t-2} + \sum_{j=1}^6 \theta_j \varepsilon_{t-j} + \varepsilon_t$ with $\phi_1 = 0.75$, $\phi_2 = -0.5$, $\theta_j = 0.1 \times (9-j)$ and $\varepsilon_t \sim \mathcal{N} (0, \sigma^2)$ i.i.d.
\end{itemize}

The results of the simulation study are reported in Tables \ref{table: coverage simulation} and \ref{table: coverage simulation dep}. We also highlight whether each method comes with theoretical coverage guarantees for each noise type, where the letter \textbf{c} indicates that the method should give correct coverage conditional on the event that the number of change points is correctly estimated. The majority of methods tested keep the nominal size well for noise types consistent with the assumptions under which they were developed and in general tend to provide over coverage. The only exception occurs for Dep-SMUCE which delivers significant under-coverage on noise types \texttt{N4} and \texttt{N5}. The coverage provided by our procedure is likewise accurate, and in particular under Gaussian noise tends to provide coverage closer to the level requested than that provided by competing methods. This shows that the asymptotic results in Theorems \ref{theorem: Gaussian tightness of normalized maximum} and \ref{theorem: tightness of non-gaussian dependent maximum} hold well in finite samples, and that that our procedure is generally better calibrated than other available methods; see also the additional simulation study in Section \ref{section: additional simulation study} of the appendix, which shows that the same results hold for a range of signal lengths. 

\begin{table}[!htbp]
\caption{Proportion of times out of $100$ replications each method returned no intervals of significance when applied to a noise vector of length $n = 750$, as well as whether each method is theoretically guaranteed to provide correct coverage. The letter \textbf{c} indicates that the method should give correct coverage conditional on the event that the number of change points is correctly estimated. The the letter \textbf{e} indicates that no theoretical guarantees are given but the authors observe good empirical performance of the method.} 
\label{table: coverage simulation}
\centering
\begin{subtable}{\linewidth}\centering
{\scalebox{0.9}{
\begin{tabular}{|l|c|c|c|c|}
  \hline
 & guarantee & degree 0 & degree 1 & degree 2 \\ 
  \hline
DIF1-MAD & \cmark & 0.93 & 0.92 & 0.95 \\ 
  DIF2-SD & \cmark & 0.98 & 1.00 & 1.00 \\ 
  DIF2-LRV & \cmark & 0.97 & 0.99 & 0.97 \\ 
  NSP & \cmark & 0.96 & 0.99 & 0.99 \\ 
  NSP-SN & \cmark & 1.00 & 1.00 & 1.00 \\ 
  NSP-AR & \cmark & 1.00 & 1.00 & 0.99 \\ 
  B\&P & \textbf{c} & 0.99 & - & - \\ 
  MOSUM (uniscale) & \textbf{c} & 0.98 & - & - \\ 
  MOSUM (multiscale) & \textbf{c} & 0.94 & - & - \\ 
  SMUCE & \cmark  & 0.96 & - & - \\ 
  H-SMUCE & \cmark  & 0.95 & - & - \\ 
  Dep-SMUCE & \cmark  & 0.92 & - & - \\ 
   \hline
\end{tabular}
}}
\caption{Coverage on noise type \texttt{N1} with $\sigma = 1$.}
\end{subtable}%

\begin{subtable}{\linewidth}\centering
{\scalebox{0.9}{
\begin{tabular}{|l|c|c|c|c|}
  \hline
 & guarantee & degree 0 & degree 1 & degree 2 \\ 
  \hline
DIF1-MAD & \xmark & 0.46 & 0.45 & 0.38 \\ 
  DIF2-SD & \cmark & 0.98 & 0.97 & 0.95 \\ 
  DIF2-LRV & \cmark & 0.93 & 0.92 & 0.91 \\ 
  NSP & \xmark & 0.05 & 0.04 & 0.04 \\ 
  NSP-SN & \cmark & 1.00 & 1.00 & 1.00 \\ 
  NSP-AR & \xmark & 0.14 & 0.10 & 0.19 \\ 
  B\&P & \textbf{c} & 0.97 & - & - \\ 
  MOSUM (uniscale) & \textbf{c} & 0.99 & - & - \\ 
  MOSUM (multiscale) & \textbf{c} & 0.98 & - & - \\ 
  SMUCE & \xmark & 0.21 & - & - \\ 
  H-SMUCE & \textbf{e} & 1.00 & - & - \\ 
  Dep-SMUCE & \cmark & 0.95 & - & - \\ 
   \hline
\end{tabular}
}}
\caption{Coverage on noise type \texttt{N2} with $\sigma = 1$.}
\end{subtable}%

\begin{subtable}{\linewidth}\centering
{\scalebox{0.9}{
\begin{tabular}{|l|c|c|c|c|}
  \hline
 & guarantee & degree 0 & degree 1 & degree 2 \\ 
  \hline
DIF1-MAD & \xmark & 0.36 & 0.33 & 0.37 \\ 
  DIF2-SD & \cmark & 0.97 & 0.99 & 0.99 \\ 
  DIF2-LRV & \cmark & 0.98 & 0.98 & 0.94 \\ 
  NSP & \xmark & 0.02 & 0.04 & 0.03 \\ 
  NSP-SN & \cmark & 1.00 & 1.00 & 1.00 \\ 
  NSP-AR & \xmark & 0.19 & 0.23 & 0.22 \\ 
  B\&P & \textbf{c} & 0.95 & - & - \\ 
  MOSUM (uniscale) & \textbf{c} & 1.00 & - & - \\ 
  MOSUM (multiscale) & \textbf{c} & 0.98 & - & - \\ 
  SMUCE & \xmark & 0.14 & - & - \\ 
  H-SMUCE & \textbf{e} & 1.00 & - & - \\ 
  Dep-SMUCE & \cmark & 0.90 & - & - \\ 
   \hline
\end{tabular}
}}
\caption{Coverage on noise type \texttt{N3} with $\sigma = 1$.}
\end{subtable}
\end{table}

\begin{table}[!htbp]
\caption{Proportion of times out of $100$ replications each method returned no intervals of significance when applied to a noise vector of length $n = 750$, as well as whether each method is theoretically guaranteed to provide correct coverage. The letter \textbf{c} indicates that the method should give correct coverage conditional on the event that the number of change points is correctly estimated. The the letter \textbf{e} indicates that no theoretical guarantees are given but the authors observe good empirical performance of the method.} 
\label{table: coverage simulation dep}
\centering
\begin{subtable}{\linewidth}\centering
{\scalebox{0.9}{
\begin{tabular}{|l|c|c|c|c|}
  \hline
 & guarantee & degree 0 & degree 1 & degree 2 \\ 
  \hline
DIF1-MAD & \xmark & 0.00 & 0.00 & 0.00 \\ 
  DIF2-SD & \xmark & 0.00 & 0.00 & 0.00 \\ 
  DIF2-LRV & \cmark & 0.90 & 0.90 & 0.89 \\ 
  NSP & \xmark & 0.00 & 0.00 & 0.00 \\ 
  NSP-SN & \xmark & 0.00 & 0.00 & 0.01 \\ 
  NSP-AR & \cmark & 1.00 & 0.99 & 0.98 \\ 
  B\&P & \xmark & 0.00 & - & - \\ 
  MOSUM (uniscale) & \xmark & 0.00 & - & - \\ 
  MOSUM (multiscale) & \xmark & 0.00 & - & - \\ 
  SMUCE & \xmark & 0.00 & - & - \\ 
  H-SMUCE & \xmark & 0.00 & - & - \\ 
  Dep-SMUCE & \cmark & 0.41 & - & - \\ 
   \hline
\end{tabular}
}}
\caption{Coverage on noise type \texttt{N4} with $\sigma = 1$.}
\end{subtable}%

\begin{subtable}{\linewidth}\centering
{\scalebox{0.9}{
\begin{tabular}{|l|c|c|c|c|}
  \hline
 & guarantee & degree 0 & degree 1 & degree 2 \\ 
  \hline
DIF1-MAD & \xmark & 0.00 & 0.00 & 0.00 \\ 
  DIF2-SD & \xmark & 0.00 & 0.00 & 0.00 \\ 
  DIF2-LRV & \cmark & 0.87 & 0.91 & 0.95 \\ 
  NSP & \xmark & 0.00 & 0.00 & 0.00 \\ 
  NSP-SN & \xmark & 0.00 & 0.01 & 0 \\ 
  NSP-AR & \xmark & 0.17 & 0.12 & 0.07 \\ 
  B\&P & \xmark & 0.00 & - & - \\ 
  MOSUM (uniscale) & \xmark & 0.00 & - & - \\ 
  MOSUM (multiscale) & \xmark & 0.00 & - & - \\ 
  SMUCE & \xmark & 0.00 & - & - \\ 
  H-SMUCE & \xmark & 0.00 & - & - \\ 
  Dep-SMUCE & \cmark & 0.32 & - & - \\ 
   \hline
\end{tabular}
}}
\caption{Coverage on noise type \texttt{N5} with $\sigma = 1$.}
\end{subtable}%

\begin{subtable}{\linewidth}\centering
{\scalebox{0.9}{
\begin{tabular}{|l|c|c|c|c|}
  \hline
 & guarantee & degree 0 & degree 1 & degree 2 \\ 
  \hline
DIF1-MAD & \xmark & 0.00 & 0.00 & 0.00 \\ 
  DIF2-SD & \xmark & 0.00 & 0.00 & 0.00 \\ 
  DIF2-LRV & \cmark & 0.99 & 0.95 & 1.00 \\ 
  NSP & \xmark & 0.00 & 0.00 & 0.00 \\ 
  NSP-SN & \xmark & 0.03 & 0.10 & 0.12 \\ 
  NSP-AR & \xmark & 0.83 & 0.87 & 0.93 \\ 
  B\&P & \xmark & 0.00 & - & - \\ 
  MOSUM (uniscale) & \xmark & 0.00 & - & - \\ 
  MOSUM (multiscale) & \xmark & 0.00 & - & - \\ 
  SMUCE & \xmark & 0.00 & - & - \\ 
  H-SMUCE & \xmark & 0.00 & - & - \\ 
  Dep-SMUCE & \cmark & 0.94 & - & - \\ 
   \hline
\end{tabular}
}}
\caption{Coverage on noise type \texttt{N6} with $\sigma = 1$.}
\end{subtable}
\end{table}

\subsection{Coverage in the presence of strong serial dependence} \label{section: strong serial dependence}

Calibrating change point procedures in the presence of serial dependence is a difficult problem, and in practice few available methods work well uniformly; see for instance the numerical comparison in \cite{cho2021multiple}. We remark that in the presence of strong serial dependence the coverage provided by our procedure can break down. To illustrate this, Table \ref{table: ar1 sims} reports the proportion of times over $100$ replications for which DIF2-LRV reported no intervals on significance on the the signal
\begin{align}
& \zeta_t = \phi_j \zeta_{t-1} + \varepsilon_t 
\label{equation: ar-1 noise} \text{ with } \phi_j = 0.8 + j / 100 
\end{align}
with $\varepsilon_t \sim \mathcal{N} (0,1)$ i.i.d. and $j = 0, \dots, 10$. For large $\phi$ the procedure no longer delivers the desired coverage. However, on closer inspection this appears to be a failure of the long run variance estimator proposed in (\ref{equation: long run varaince estimator}) which for values of $\phi$ close to $1$ tends to under-estimate the long run variance, rather than the asymptotic theory. This is because scaling each local test by the true time average variance constant (TAVC, \citep{wu2009recursive}), which for a given scale $W''$ is defined as 
\begin{equation}
\text{TAVC} \left ( W'' \right ) = \mathbb{E} \left [ \left ( \frac{1}{\sqrt{W''}} \sum_{t=1}^{W''} \zeta_t \right ) ^ 2 \right ],
\end{equation}
at a scale proportional to the $W$ supplied to DIF2-LRV our procedure attains the desired level of coverage. The time average variance constant converges to the long run variance as long as $W''$ diverges with the sample size. As argued by \cite{mcgonigle2023robust} it is preferable to scale by the time average variance constant, as opposed to the long run variance, as with a properly chosen scale the latter better accounts for the local variation of each test. In fact, close inspection of the proof of Lemma \ref{lemma: consistency of LRV estimator} reveals that (\ref{equation: long run varaince estimator}) is consistent for the TAVC calculated at the scale $W'$. However, the conditions of Lemma \ref{lemma: consistency of LRV estimator} limit (\ref{equation: long run varaince estimator}) to scales of the order $W' = o \left ( \sqrt{n} \right )$. 

\begin{table}[h] 
\caption{Proportion of times out of $100$ replications DIF2-LRV returned no intervals of significance when applied to a noise vector of length $n = 750$ form (\ref{equation: ar-1 noise}) and normalized with estimated long run standard deviation $\hat{\tau}_{\text{DIF}}$ calculated according to (\ref{equation: long run varaince estimator}) as well as the TAVC calculated at scale $W'' = \frac{2}{5}\sqrt{n} $, the true long run standard deviation.} 
\centering
\label{table: ar1 sims} 
\begin{tabular}{|c|c|c|c|c|c|c|c|c|c|c|c|}
  \hline
$\phi$ & 0.8 & 0.81 & 0.82 & 0.83 & 0.84 & 0.85 & 0.86 & 0.87 & 0.88 & 0.89 & 0.9 \\ 
  \hline
Estimated LRV & 0.94 & 0.86 & 0.89 & 0.89 & 0.86 & 0.81 & 0.84 & 0.83 & 0.62 & 0.72 & 0.51 \\ 
True TAVC & 0.98 & 0.96 & 0.96 & 0.98 & 0.96 & 0.94 & 0.93 & 0.96 & 0.90 & 0.95 & 0.87 \\ 
   \hline
\end{tabular}

\end{table} 
 
In light of the above, there are a number of approaches one may take if strong serial dependence is suspected. For instance, one could slightly pre-whiten the data using the heuristic methods suggested in Section 4.1 of \cite{baranowski2019narrowest}. Alternatively one may run DIF2-LRV with a conservative, but nonetheless consistent, estimator of the long run variance. For instance one may scale by $\hat{\tau}^2_{\text{DIF}} + C / W'''$ for some positive constant $C>0$ and some $W'''$ diverging with $n$. 

\subsection{Performance on test signals} \label{section: performance on test signals}

Next we investigate the performance of our method and its competitors on test signals containing change points. To investigate performance we apply each method to $100$ sample paths from the change point models \texttt{M1}, \texttt{M2}, and \texttt{M3} listed below, contaminated with each of the four noise types introduced in Section \ref{section: coverage simulations} above. On each iteration we record for each method: the number of intervals which contain at least one change point location (no. genuine), the proportion of intervals returned which contain at least one change point location (prop. genuine), the average length of intervals returned (length), and whether all intervals returned contain at least once change point location (coverage). We report the average of these quantities, and again highlight whether each method comes with theoretical coverage guarantees for each noise type (guarantee). 

\begin{itemize}
    \item (\texttt{M1}): the first $n = 512$ values of piecewise constant the \texttt{blocks} signal from \cite{donoho1994ideal}, shown in Figure \ref{subfigure: blocks}, with $N = 4$ change points at locations $\Theta = \left \{ 205, 267, 308, 472 \right \}$
    \item (\texttt{M2}): the first $n = 600$ values of the piecewise linear \texttt{waves} signal from \cite{baranowski2019narrowest}, shown in Figure \ref{subfigure: waves}, with $N = 3$ change points at locations $\Theta = \left \{ 150, 300, 450 \right \}$
    \item (\texttt{M3}): the piecewise quadratic \texttt{hills} signal with length $n = 400$, shown in Figure \ref{subfigure: hills}, with $N = 3$ change points at locations $\Theta = \left \{100, 200, 300 \right \}$
\end{itemize}

The results of the simulation study are reported in Tables \ref{table: blocks performance} - \ref{table: hills performance}. On the piecewise constant \texttt{blocks} function, among the methods which provide correct coverage, our algorithm is generally among the top performing methods in terms the number of change points detected and the lengths of intervals recovered. In fact, is only outperformed by the MOSUM procedure with multiscale bandwidth under noise types \texttt{N1} and \texttt{N2}. The family of SMUCE algorithms, as well as the B\&P procedure, all suffer from under coverage on noise types for which they should give accurate coverage. Among the methods compared to only the family of NSP algorithms is applicable to higher order piecewise polynomial signals. On the piecewise polynomial \texttt{waves} and \texttt{hills} signals our methods deliver correct coverage where theoretical guarantees are available and consistently outperform the only competitor, the family of NSP algorithms.

\begin{table}[!htbp] 
\caption{
Average of the number of intervals which contain at least one change point location (no. genuine), the proportion of intervals returned which contain at least one change point location (prop. genuine), the average length of intervals returned (length), and whether all intervals returned contain at least once change point location (coverage), on the piecewise constant \texttt{blocks} signal contaminated with noise \texttt{N1-N4} over $100$ replications. The noise level was set to $\sigma = 10$ for noise types \texttt{N1-2}and to $\sigma = 5$ for noise types \texttt{N3-4}. We also report whether each method is theoretically guaranteed to provide correct coverage.}
\centering
\label{table: blocks performance} 
\scalebox{0.9}{
\begin{tabular}{|c|c|c|c|c|c|c|c|}
  \hline
 &  & N1 & N2 & N3 & N4 & N5 & N6 \\ 
  \hline
 & no. genuine & 3.69 & 3.75 & 3.87 & 3.46 & 3.45 & 3.42 \\ 
  DIF1-MAD & prop. genuine & 0.99 & 0.89 & 0.85 & 0.13 & 0.11 & 0.09 \\ 
   & length & 34.86 & 27.19 & 23.89 & 9.36 & 8.87 & 8.30 \\ 
   & coverage & 0.97 & 0.61 & 0.42 & 0.00 & 0.00 & 0.00 \\ 
\hline
   & no. genuine & 3.34 & 3.36 & 3.40 & 3.27 & 3.41 & 3.65 \\ 
  DIF2-SD & prop. genuine & 1.00 & 1.00 & 1.00 & 0.17 & 0.19 & 0.17 \\ 
   & length & 43.72 & 43.80 & 43.41 & 16.63 & 16.59 & 16.15 \\ 
   & coverage & 1.00 & 0.99 & 0.99 & 0.00 & 0.00 & 0.00 \\ 
\hline
   & no. genuine & 1.98 & 2.03 & 1.97 & 1.35 & 1.33 & 1.39 \\ 
  DIF2-LRV & prop. genuine & 0.99 & 1.00 & 0.99 & 0.90 & 0.91 & 0.95 \\ 
   & length & 61.35 & 60.67 & 58.03 & 69.27 & 80.50 & 71.57 \\ 
   & coverage & 1.00 & 1.00 & 1.00 & 1.00 & 1.00 & 1.00 \\ 
\hline
   & no. genuine & 3.20 & 3.48 & 3.52 & 3.67 & 3.74 & 3.86 \\ 
  NSP & prop. genuine & 1.00 & 0.63 & 0.60 & 0.15 & 0.12 & 0.10 \\ 
   & length & 59.63 & 36.06 & 32.45 & 11.34 & 9.35 & 7.85 \\ 
   & coverage & 1.00 & 0.14 & 0.12 & 0.00 & 0.00 & 0.00 \\ 
\hline
   & no. genuine & 1.92 & 1.90 & 1.93 & 3.11 & 3.11 & 3.00 \\ 
  NSP-SN & prop. genuine & 1.00 & 1.00 & 1.00 & 0.83 & 0.84 & 0.94 \\ 
   & length & 120.41 & 117.62 & 113.95 & 75.43 & 75.16 & 73.08 \\ 
   & coverage & 1.00 & 1.00 & 1.00 & 0.36 & 0.42 & 0.80 \\ 
\hline
   & no. genuine & 0.12 & 0.87 & 0.88 & 0.68 & 0.99 & 1.56 \\ 
  NSP-AR & prop. genuine & 0.12 & 0.54 & 0.55 & 0.61 & 0.51 & 0.89 \\ 
   & length & 24.58 & 63.29 & 78.52 & 40.76 & 46.77 & 36.83 \\ 
   & coverage & 1.00 & 0.50 & 0.54 & 1.00 & 0.43 & 0.98 \\ 
\hline
   & no. genuine & 3.85 & 3.88 & 3.93 & 3.49 & 3.73 & 3.60 \\ 
  B\&P & prop. genuine & 0.96 & 0.96 & 0.98 & 0.19 & 0.20 & 0.25 \\ 
   & length & 16.78 & 17.14 & 16.32 & 13.71 & 13.68 & 14.84 \\ 
   & coverage & 0.83 & 0.85 & 0.92 & 0.00 & 0.00 & 0.00 \\ 
\hline
   & no. genuine & 1.96 & 2.02 & 2.06 & 3.54 & 3.59 & 3.51 \\ 
  MOSUM (uniscale) & prop. genuine & 0.83 & 0.83 & 0.88 & 0.20 & 0.21 & 0.24 \\ 
   & length & 14.03 & 14.21 & 14.21 & 15.87 & 15.82 & 15.36 \\ 
   & coverage & 0.89 & 0.91 & 0.94 & 0.00 & 0.00 & 0.00 \\ 
\hline
   & no. genuine & 3.90 & 3.98 & 3.97 & 4.90 & 4.96 & 4.69 \\ 
  MOSUM (multiscale) & prop. genuine & 0.96 & 0.98 & 0.98 & 0.23 & 0.23 & 0.25 \\ 
   & length & 22.01 & 21.14 & 20.62 & 21.30 & 21.13 & 22.02 \\ 
   & coverage & 0.86 & 0.93 & 0.92 & 0.00 & 0.00 & 0.00 \\ 
\hline
   & no. genuine & 3.71 & 3.61 & 3.82 & 2.12 & 1.87 & 1.68 \\ 
  SMUCE & prop. genuine & 0.95 & 0.72 & 0.73 & 0.09 & 0.06 & 0.05 \\ 
   & length & 36.02 & 24.21 & 23.99 & 8.59 & 7.03 & 5.91 \\ 
   & coverage & 0.89 & 0.35 & 0.24 & 0.00 & 0.00 & 0.00 \\ 
\hline
   & no. genuine & 3.40 & 3.08 & 3.21 & 2.81 & 2.91 & 2.42 \\ 
  H-SMUCE & prop. genuine & 0.92 & 0.86 & 0.89 & 0.84 & 0.84 & 0.77 \\ 
   & length & 49.42 & 44.92 & 45.63 & 49.32 & 52.98 & 54.19 \\ 
   & coverage & 0.80 & 0.70 & 0.72 & 0.63 & 0.65 & 0.56 \\ 
\hline
   & no. genuine & 2.12 & 2.15 & 2.30 & 3.11 & 2.94 & 3.42 \\ 
  Dep-SMUCE & prop. genuine & 0.80 & 0.78 & 0.82 & 0.49 & 0.46 & 0.61 \\ 
   & length & 73.66 & 69.10 & 72.25 & 32.73 & 29.83 & 33.06 \\ 
   & coverage & 0.57 & 0.59 & 0.61 & 0.00 & 0.01 & 0.06 \\ 
\hline
\end{tabular}
}
\end{table} 

\begin{table}[!htbp] 
\caption{Average of the number of intervals which contain at least one change point location (no. genuine), the proportion of intervals returned which contain at least one change point location (prop. genuine), the average length of intervals returned (length), and whether all intervals returned contain at least once change point location (coverage), on the piecewise linear \texttt{waves} signal contaminated with noise types \texttt{N1-N4} over $100$ replications. The noise level was set to $\sigma = 5$ for all noise types. We also report whether each method is theoretically guaranteed to provide correct coverage.} 
\centering
\label{table: waves performance} 
\begin{tabular}{|c|c|c|c|c|c|c|c|}
  \hline
 &  & N1 & N2 & N3 & N4 & N5 & N6 \\ 
  \hline
 & no. genuine & 2.98 & 2.77 & 2.66 & 1.51 & 1.60 & 2.28 \\ 
  DIF1-MAD & prop. genuine & 0.98 & 0.83 & 0.77 & 0.06 & 0.06 & 0.05 \\ 
   & length & 81.57 & 65.57 & 58.03 & 12.85 & 11.91 & 9.13 \\ 
   & coverage & 0.92 & 0.49 & 0.39 & 0.00 & 0.00 & 0.00 \\ 
   \hline
   & no. genuine & 2.99 & 2.98 & 2.97 & 1.66 & 1.71 & 2.56 \\ 
  DIF2-SD & prop. genuine & 1.00 & 0.99 & 0.99 & 0.08 & 0.09 & 0.08 \\ 
   & length & 94.25 & 92.87 & 92.98 & 16.67 & 17.14 & 15.05 \\ 
   & coverage & 0.99 & 0.98 & 0.96 & 0.00 & 0.00 & 0.00 \\ 
   \hline
   & no. genuine & 3.00 & 2.98 & 2.98 & 1.36 & 1.51 & 1.54 \\ 
  DIF2-LRV & prop. genuine & 1.00 & 0.99 & 0.99 & 0.96 & 0.99 & 0.99 \\ 
   & length & 95.78 & 97.32 & 98.90 & 233.37 & 219.28 & 239.46 \\ 
   & coverage & 0.99 & 0.98 & 0.97 & 0.97 & 0.97 & 0.99 \\ 
   \hline
   & no. genuine & 3.00 & 2.60 & 2.67 & 1.73 & 1.85 & 1.85 \\ 
  NSP & prop. genuine & 1.00 & 0.65 & 0.66 & 0.11 & 0.09 & 0.07 \\ 
   & length & 93.06 & 58.00 & 57.31 & 20.36 & 17.06 & 14.02 \\ 
   & coverage & 1.00 & 0.16 & 0.19 & 0.00 & 0.00 & 0.00 \\ 
   \hline
   & no. genuine & 2.99 & 3.00 & 3.00 & 2.42 & 2.34 & 2.44 \\ 
  NSP-SN & prop. genuine & 1.00 & 1.00 & 1.00 & 0.84 & 0.83 & 0.96 \\ 
   & length & 126.23 & 124.74 & 125.23 & 119.64 & 116.84 & 135.12 \\ 
   & coverage & 1.00 & 1.00 & 1.00 & 0.58 & 0.57 & 0.90 \\ 
   \hline
   & no. genuine & 0.62 & 1.38 & 1.63 & 0.00 & 0.44 & 0.09 \\ 
  NSP-AR & prop. genuine & 0.51 & 0.63 & 0.77 & 0.00 & 0.29 & 0.09 \\ 
   & length & 98.93 & 96.03 & 113.80 & 0.00 & 59.68 & 22.95 \\ 
   & coverage & 1.00 & 0.39 & 0.59 & 1.00 & 0.39 & 0.97 \\ 
   \hline
\end{tabular}

\end{table} 

\begin{table}[!htbp] 
\caption{Average of the number of intervals which contain at least one change point location (no. genuine), the proportion of intervals returned which contain at least one change point location (prop. genuine), the average length of intervals returned (length), and whether all intervals returned contain at least once change point location (coverage), on the piecewise quadratic \texttt{hills} signal contaminated with noise types \texttt{N1-N4} over $100$ replications. The noise level was set to $\sigma = 1$ for all noise types. We also report whether each method is theoretically guaranteed to provide correct coverage.} 
\centering
\label{table: hills performance} 
\begin{tabular}{|c|c|c|c|c|c|c|c|}
  \hline
 &  & N1 & N2 & N3 & N4 & N5 & N6 \\ 
  \hline
 & no. genuine & 3.00 & 2.85 & 2.90 & 1.86 & 1.79 & 1.99 \\ 
  DIF1-MAD & prop. genuine & 0.99 & 0.85 & 0.86 & 0.14 & 0.12 & 0.07 \\ 
   & length & 43.32 & 36.10 & 35.09 & 16.66 & 14.93 & 8.60 \\ 
   & coverage & 0.95 & 0.51 & 0.58 & 0.00 & 0.00 & 0.00 \\ 
   \hline
   & no. genuine & 3.00 & 2.99 & 3.00 & 1.72 & 1.92 & 2.60 \\ 
  DIF2-SD & prop. genuine & 1.00 & 0.99 & 1.00 & 0.16 & 0.18 & 0.12 \\ 
   & length & 51.96 & 51.37 & 51.16 & 19.77 & 20.36 & 16.01 \\ 
   & coverage & 1.00 & 0.97 & 1.00 & 0.00 & 0.00 & 0.00 \\ 
   \hline
   & no. genuine & 3.00 & 3.00 & 3.00 & 1.82 & 1.72 & 1.69 \\ 
  DIF2-LRV & prop. genuine & 1.00 & 1.00 & 1.00 & 0.98 & 0.95 & 0.98 \\ 
   & length & 69.29 & 68.27 & 68.84 & 122.55 & 121.99 & 131.37 \\ 
   & coverage & 1.00 & 1.00 & 1.00 & 0.99 & 0.95 & 1.00 \\ 
   \hline
   & no. genuine & 3.00 & 2.89 & 2.97 & 2.04 & 2.13 & 1.99 \\ 
  NSP & prop. genuine & 1.00 & 0.83 & 0.87 & 0.25 & 0.22 & 0.16 \\ 
   & length & 50.55 & 40.40 & 39.59 & 28.13 & 23.17 & 19.24 \\ 
   & coverage & 1.00 & 0.44 & 0.60 & 0.00 & 0.00 & 0.00 \\ 
   \hline
   & no. genuine & 2.96 & 2.96 & 2.93 & 2.77 & 2.77 & 2.69 \\ 
  NSP-SN & prop. genuine & 1.00 & 1.00 & 1.00 & 1.00 & 0.98 & 1.00 \\ 
   & length & 83.66 & 83.21 & 83.62 & 92.23 & 90.86 & 95.98 \\ 
   & coverage & 1.00 & 1.00 & 1.00 & 1.00 & 0.95 & 1.00 \\ 
   \hline
   & no. genuine & 0.52 & 1.40 & 1.60 & 0.03 & 0.55 & 0.24 \\ 
  NSP-AR & prop. genuine & 0.44 & 0.75 & 0.85 & 0.03 & 0.40 & 0.24 \\ 
   & length & 60.37 & 64.41 & 72.31 & 3.51 & 52.93 & 41.64 \\ 
   & coverage & 1.00 & 0.65 & 0.79 & 1.00 & 0.48 & 0.93 \\ 
   \hline
\end{tabular}

\end{table} 

\section{Real data examples} \label{section: real data examples}

\subsection{Application to bone mineral density acquisition curves}

We analyse data on bone mineral acquisition in 423 healthy males and females aged between $9$ and $25$. The data is available from \href{https://hastie.su.domains/ElemStatLearn/datasets/spnbmd.csv}{\texttt{hastie.su.domains}} and was first analysed in \cite{bachrach1999bone}. The data was originally collected as part of a longitudinal study where four consecutive yearly measurements of bone mass by dual energy x-ray absorptiometry were taken from each subject. We obtain bone density acquisition curves for males and females by grouping measurements by gender and age and averaging over measurements in each grouping. The processed data are plotted in the first row of Figure \ref{figure: bone density}. There is some disagreement over the age at which peak bone mass density is attained in adolescents \cite{kroger1993development, theintz1992longitudinal, lu1996volumetric}. One possible solution is to model the data in Figure \ref{figure: bone density} as following a piecewise linear trend, and to infer this information from any estimated change point locations. 

We apply the procedure DIF2-SD to the data, with the tuning parameters specified in Section \ref{section: numerical illustrations}, because as the data are strictly positive the assumption of Gaussian noise is unlikely to hold. We additionally estimate change point locations using five state of the art algorithms for recovering changes in piecewise linear signals which however do not come with any coverage guarantees. These are: the Narrowest-Over-Threshold algorithm (NOT) of \cite{baranowski2019narrowest}, and the same algorithm run with the requirement that the estimated signal be continuous (NOT-cont), the Isolate Detect algorithm (ID) of \cite{anastasiou2022detecting}, the dynamic programming based algorithm of \cite{bai1998estimating} (BP), and the Continuous-piecewise-linear Pruned Optimal Partitioning algorithm (CPOP) of \cite{fearnhead2019detecting}. When applying each method we use the default parameters in their respective R packages. 

The results of the analysis are shown in in the second row of Figure \ref{figure: bone density}. On both bone density acquisition curves all methods for change point detection estimate a single change point location, save for CPOP. However, on the male bone density acquisition data there is considerable disagreement among the methods regarding the location of the change point detected. Since the methods do not quantify the uncertainty around each estimated change point, it is difficult to say which estimate is closest to the truth. DIF2-SD also returns a single interval of significance when applied to each data set, and each interval returned contains all change point locations recovered by the other methods on each respective data set save the extraneous change point detected by CPOP. By Corollary \ref{corollary: coverage guarantee} one can be certain each interval contains at least one true change point location with high probability. We therefore re-apply the aforementioned change point detectors to this interval only. The results are shown in the third row of Figure \ref{figure: bone density}, where this time there is much greater agreement among the methods. We also note that the corresponding intervals returned by NSP-SN (not shown), which is the only competing method from Section \ref{section: alternative methods} applicable to the data, cover essentially the entire range of the data. 

\begin{figure}[!htbp]
\centering
\caption{black / grey solid lines (\textbf{---} / \textcolor{gray}{\textbf{---}}) represents bone density acquisition curves for males and females between the ages of $9$ and $25$, red shaded regions (\textcolor{transpred}{$\blacksquare$}) represent intervals of significance returned by DIF2-SD, dashed coloured lines represent change point locations recovered by NOT (\textcolor{blue}{\textbf{- - -}}), NOT-cont (\textcolor{blue}{$\cdots$}), ID (\textcolor{orange}{\textbf{- - -}}), BP (\textcolor{green}{\textbf{- - -}}), and CPOP (\textcolor{purple}{\textbf{- - -}}) } 
\label{figure: bone density}
\begin{subfigure}[b]{0.4\textwidth}
\centering
\includegraphics[width=\textwidth]{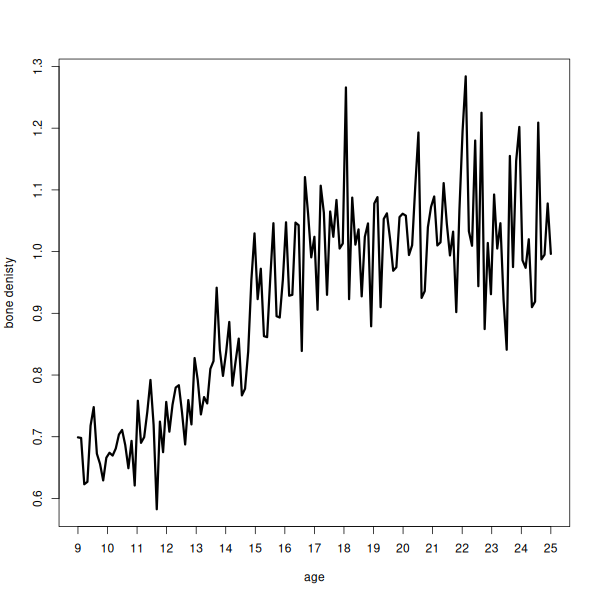}
\caption{male bone density}
\end{subfigure}
\begin{subfigure}[b]{0.4\textwidth}
\centering
\includegraphics[width=\textwidth]{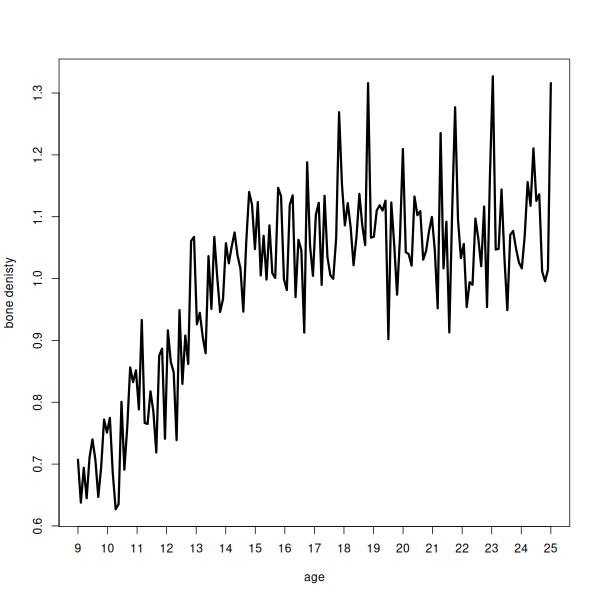}
\caption{female bone density}
\end{subfigure}

\begin{subfigure}[b]{0.4\textwidth}
\centering
\includegraphics[width=\textwidth]{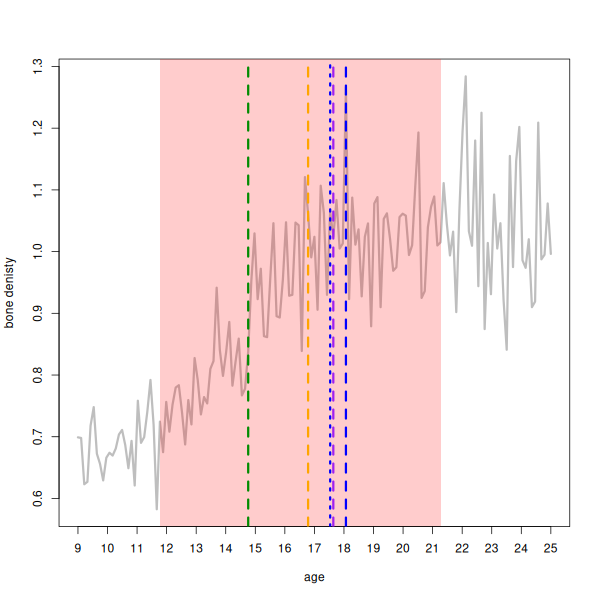}
\caption{estimated change points}
\end{subfigure}
\begin{subfigure}[b]{0.4\textwidth}
\centering
\includegraphics[width=\textwidth]{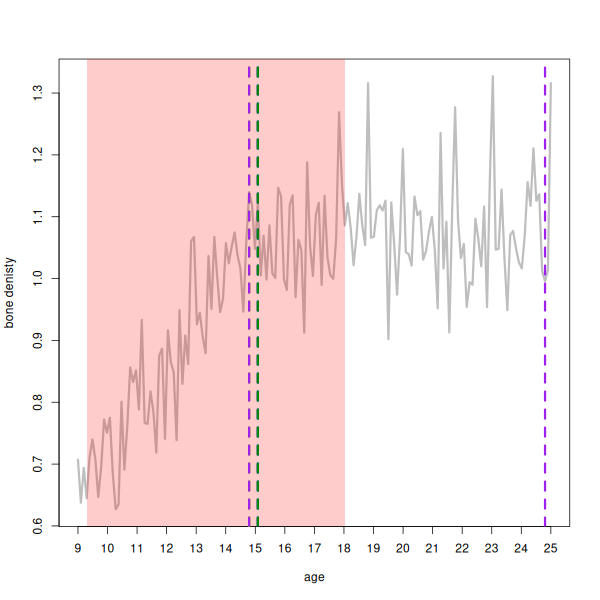}
\caption{estimated change points}
\end{subfigure}

\begin{subfigure}[b]{0.4\textwidth}
\centering
\includegraphics[width=\textwidth]{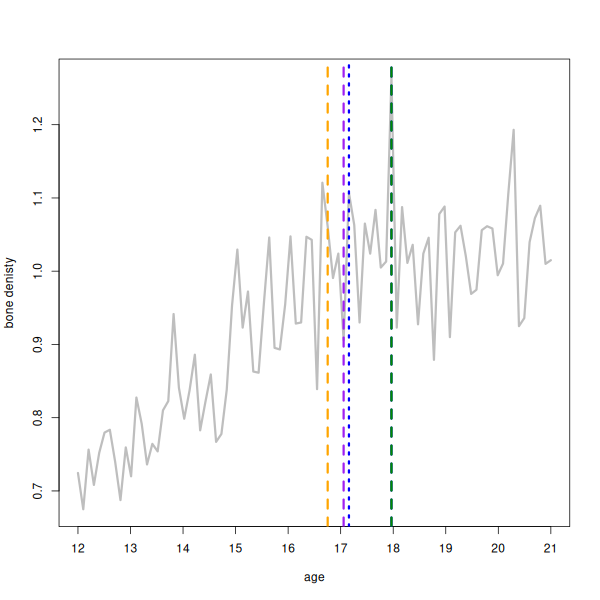}
\caption{change points on sub-interval}
\end{subfigure}
\begin{subfigure}[b]{0.4\textwidth}
\centering
\includegraphics[width=\textwidth]{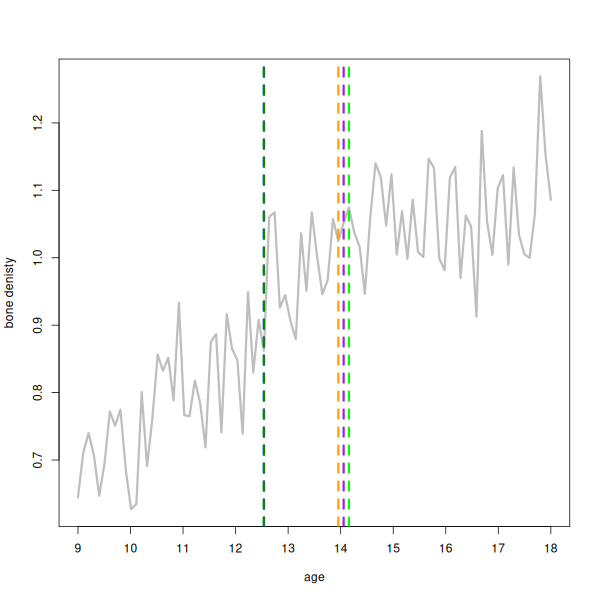}
\caption{change points on sub-interval}
\end{subfigure}
\end{figure}

\subsection{Applications to nitrogen dioxide concentration in London}

We analyse daily average concentrations of nitrogen dioxide ($\text{NO}_2$) at Marylebone Road in London between September 2, 2000 and September 30, 2020. The data are available from \href{https://uk-air.defra.gov.uk/}{\texttt{uk-air.defra.gov.uk}} and were originally analysed from a change point perspective, assuming a piecewsie constant mean, by \cite{cho2021multiple}. We follow their analysis in \cite{Cho2022git} by taking the square root transform of the data and removing seasonal and weekly variation. The processed data is plotted in Figure \ref{figure: London NO2 de-trended}. \cite{cho2021multiple} identify three historical events which are likely to have affected $\text{NO}_2$ concentration levels in London during the period in question, which are summarised below. 

\begin{itemize}
    \item \underline{February 2003:} \emph{installation of particulate traps on most London buses and other heavy duty diesel vehicles.}
    \item \underline{April 8, 2019:} \emph{introduction of Ultra Low Emission zones in central London.}
    \item \underline{March 23, 2020:} \emph{beginning of the nation-wide COVID-19 lockdown.}
\end{itemize}

We apply the procedure DIF2-LRV to the data with tuning parameters specified in Section \ref{section: numerical illustrations}, since time series of $\text{NO}_2$ concentrations are known to be strongly serially correlated. For comparison we additionally estimate change point locations using three state of the art algorithms for recovering changes in piecewise constant signals in the presence of serially correlated noise, which however do not come with coverage guarantees. These are: the algorithm of \cite{romano2022detecting} for Detecting Changes in Autocorrelated and Fluctuating Signals (DeCAFS), the algorithm of \cite{chakar2017robust} for estimating multiple change-points in the
mean of a Gaussian AR(1) process (AR1seg), and the Wild Contrast Maximisation and gappy Schwarz algorithm (WCM.gSa) of \cite{cho2021multiple}. When applying each method we use default parameters in their respective R packages save for the DeCAFS algorithm for which our choice of tuning parameters is guided by the \texttt{guidedModelSelection} function in the DeCAFS R package. 

\begin{figure}[!htbp]
\centering
\caption{daily average concentrations of $\text{NO}_2$ at Marylebone Road after square root transform and  with seasonal variation removed, red dashed lines (\textcolor{red}{\textbf{- - -}}) and dark red shaded region (\textcolor{red}{$\blacksquare$}) represent dates of events which are likely to have affected $\text{NO}_2$ concentration levels}
\label{figure: London NO2 de-trended}
\includegraphics[width=\textwidth]{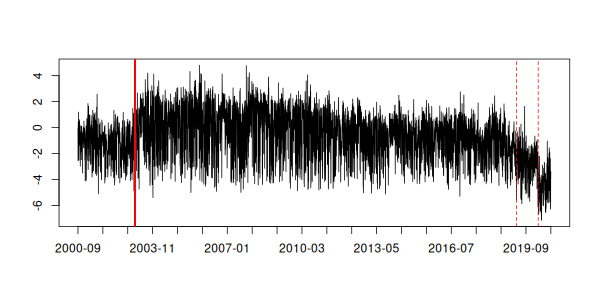}
\end{figure}

The results of the analysis are shown in Figure \ref{figure: LondoN-NO2}. DIF2-LRV returns four intervals, among which the first, third, and fourth cover the dates of important events identified by \cite{cho2021multiple}. Within each of these three intervals AR1seg, DeCAFS, and WCM.gSa each identify one change point, with the exception of WCM.gSa which identifies two change points in the third interval returned. However, when we re-apply WCM.gSa over the third interval only one change point is detected, suggesting the second change point in this interval was spuriously estimated. DeCAFS detects a change point between the first and second intervals returned by DIF2-LRV. However, re-applying the algorithms to data between the two intervals no change points are detected suggesting the original change points were also spuriously estimated. We finally note that the data analysed consists of $n = 7139$ observations, and running DIF2-LRV on a desktop computer with a 3.20GHz Intel (R) Core (TM) i7-8700 CPU took $4.1$ seconds. Running Dep-SMUCE and NSP-AR, which are the only competing methods from Section \ref{section: alternative methods} applicable to the data, on the same machine took $15.1$ seconds and $145.8$ seconds respectively. Dep-SMUCE returns similar intervals to DIF2-LRV, whereas NSP-AR does not detect any change points in the data. 

\begin{figure}[!htbp]
\centering
\caption{grey lines (\textcolor{gray}{---}) represent daily average concentrations of $\text{NO}_2$ at Marylebone Road after square root transform and  with seasonal variation removed, red shaded regions (\textcolor{transpred}{$\blacksquare$}) represent intervals of significance returned by DIF2-LRV, blue dashed lines (\textcolor{blue}{- - -}) represent change points recovered by a given algorithm, blue solid lines (\textcolor{blue}{---}) represent the corresponding fitted piecewise constant signal.} 
\label{figure: LondoN-NO2}
\begin{subfigure}[b]{\textwidth}
\centering
\includegraphics[width=0.9\textwidth]{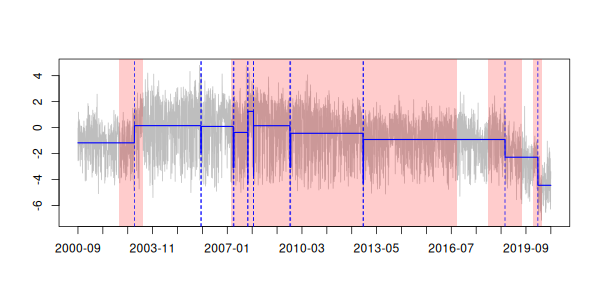}
\caption{change points and piecewise constant signal recovered by DeCAFS}
\label{subfigure: London NO2 decafs}
\end{subfigure}
\begin{subfigure}[b]{\textwidth}
\centering
\includegraphics[width=0.9\textwidth]{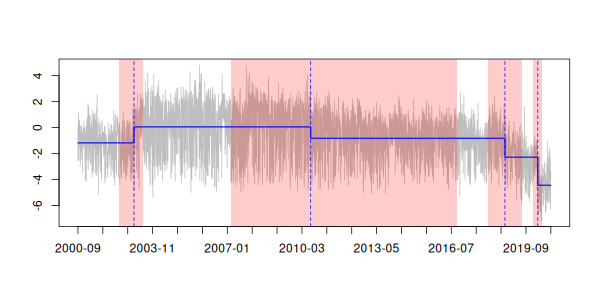}
\caption{change points and piecewise constant signal recovered by AR1seg}
\label{subfigure: London NO2 AR1-seg}
\end{subfigure}
\begin{subfigure}[b]{\textwidth}
\centering
\includegraphics[width=0.9\textwidth]{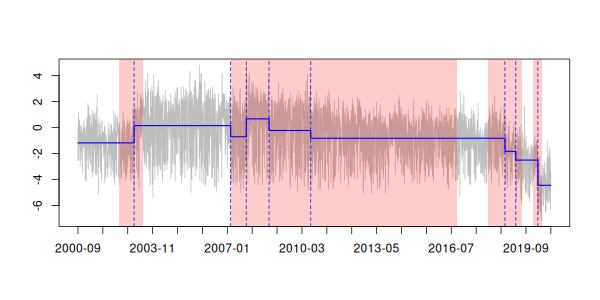}
\caption{change points and piecewise constant signal recovered by WCM.gSa}
\label{subfigure: London NO2 WGS}
\end{subfigure}
\end{figure}

\section{Acknowledgments}
The authors would like to thank Prof. Zakhar Kabluchko for helpful discussions.

\section{Proofs}

For sequences $\left \{ a_{n} \right \}_{n > 0}$ and $\left \{ b_{n} \right \}_{n > 0}$ we write $a_{n} \cleq b_{n}$ if there is a constant $C > 0$ for which $a_{n} \leq C b_{n}$ for every $n > 0$. We write $a_{n} \sim b_{n}$ if $a_{n} / b_{n} \rightarrow 1$ as $n \rightarrow \infty$. We write $|\mathcal{A}|$ for the cardinality of a set $\mathcal{A}$. The density, cumulative density, and tail functions of a standard Gaussian random variable are written respectively as $\phi \left ( \cdot \right )$, $\Phi \left ( \cdot \right )$, and $\bar{\Phi} \left ( \cdot \right )$. 

\subsection{Preparatory results}

\begin{definition}
Let $\left \{ \xi (t) \right \}_{t > 0}$ be a centred Gaussian process with unit variance, then if there are constants $C_\xi > 0$ and $\alpha \in (0, 2]$ such that for all $t > 0$ the following holds
\begin{equation*}
\text{Cov} \left ( \xi (t) , \xi (t + s) \right ) = 1 - C_\xi \left | s \right | ^ \alpha + o \left ( \left | s \right | ^ \alpha \right ), \hspace{2em} \left | s \right | \rightarrow 0,
\end{equation*}
the process is called stationary with index $\alpha$ and local structure $C_\xi$. Moreover, the process has almost surely continuous sample paths and for any compact $K \subset \mathbb{R}^{+}$ the quantity $M_K = \sup_{t \in K} \left  \{ \xi \left ( t \right ) \right \}$ is well defined. 
\label{defintion: stationary normal process}
\end{definition}

\begin{lemma}[Berman's lemma]
Let $\zeta_1, \dots, \zeta_n$ and $\tilde{\zeta}_1, \dots, \tilde{\zeta}_n $ be two sequences of Gaussian random variables with marginal $\mathcal{N}(0,1)$ distribution and covariances $\text{Cov} \left ( \zeta_i, \zeta_j \right ) = \Lambda_{ij}$ and $\text{Cov} \left ( \tilde{\zeta}_i, \tilde{\zeta}_j \right ) = \tilde{\Lambda}_{ij}$. Define $\rho_{ij} = \max \left ( \left | \Lambda_{ij} \right | , \left | \tilde{\Lambda}_{ij} \right | \right )$. For any real numbers $u_1, \dots, u_n $ the following holds. 
\begin{align*}
& \left | \mathbb{P} \left ( \zeta_j \leq u_j \mid 1 \leq j \leq n \right ) - \mathbb{P} \left ( \tilde{\zeta}_j \leq u_j \mid 1 \leq j \leq n \right ) \right | \\ 
& \leq \frac{1}{2 \pi} \sum_{1 \leq i < j \leq n} \left | \Lambda_{ij} - \tilde{\Lambda}_{ij} \right | \left ( 1 - \rho_{ij}^2 \right )^{-1/2} \exp \left (- \frac{ \frac{1}{2} \left ( u_i^2 + u_j^2 \right )}{1 + \rho_{ij}} \right )
\end{align*}

\label{lemma: normal comparison lemma}
\end{lemma}

\begin{proof}
See Theorem 4.2.1 in \cite{leadbetter2012extremes}.
\end{proof}

\begin{lemma}[Khintchine's lemma]
Let $\left \{ M_n \right \}_{n > 0}$ be a sequence of random variables and let $G$ be a non-degenerate distribution. If $\left \{ (c_n, d_n) \right \}_{n > 0}$ are scaling and centring sequences such that $\left ( M_n - c_n \right ) / d_n \rightarrow G$ then for any alternative sequences $\left \{ (c_n', d_n') \right \}_{n > 0}$ satisfying $d_n / d_n' \sim 1$ and $\left ( c_n - c_n' \right ) / d_n = o(1)$ we also have that $\left ( M_n - c_d' \right ) / d_n' \rightarrow G$. 
\label{lemma: convergence with other scaling sequcnes}
\end{lemma}

\begin{proof}
See Theorem 1.2.3 in \cite{leadbetter2012extremes}. 
\end{proof}

\begin{lemma}[Pickand's lemma, continuous version]
Let $\left \{ \xi (t) \right \}_{t>0}$ be a stationary Gaussian process with index $\alpha \in (0,2]$ and local structure $C_\xi > 0$. There is a constant $H_\alpha > 0$ such that for any compact $K \subset \mathbb{R}^+$ the following holds. 
\begin{equation*}
\mathbb{P} \left ( \sup_{t \in K} \left \{ \xi \left ( t \right ) \right \} > u \right ) \sim H_\alpha C_{\xi}^{1/\alpha} \left | K \right | u^{2/\alpha - 1} \phi \left ( u \right ) 
\end{equation*}
Moreover the values $H_1 = 1$ and $H_2 = 1 / \sqrt{\pi}$ are known explicitly. 

\label{lemma: high excursion probability - continuous}
\end{lemma}

\begin{proof}
See Theorem 9.15 in \cite{piterbarg2015twenty}, and Remark 12.2.10 in \cite{leadbetter2012extremes} for the values of $H_\alpha$. 
\end{proof}

\begin{lemma}[Pickand's lemma, discrete version]
Let $\left \{ \xi (t) \right \}_{t>0}$ be a stationary Gaussian process with index $\alpha \in ( 0, 2 ]$ and local structure $C_\xi > 0$. If $q \rightarrow 0$ and $u \rightarrow \infty$ in such a way that $u^{2/\alpha} q \rightarrow a > 0$ the following holds for any compact $K \subset \mathbb{R}^+$.
\begin{equation*}
\mathbb{P} \left ( \sup_{t \in K \cap \mathbb{Z}q} \left \{ \xi \left ( t \right ) \right \} > u \right ) \sim  F_\xi \left ( a \right ) \left | K \right | u^{2/\alpha-1} \phi \left ( u \right ) 
\end{equation*}
The function $F_\xi \left ( \cdot \right )$ is defined as follows
\begin{equation*}
F_\xi \left ( a \right ) = \lim_{T \rightarrow \infty} \frac{1}{T} \mathbb{E} \left [ \exp \left ( \sup_{s \in [0,T] \cap a \mathbb{Z}} Z \left ( s \right ) \right ) \right ].
\end{equation*}
Where $\left \{ Z ( s ) \right \}_{s > 0}$ is a stationary Gaussian process with first and second moments as follows
\begin{gather*}
\mathbb{E} \left ( Z (s) \right ) = - C_\xi \left | s \right |^\alpha, \\
\text{Cov} \left ( Z (s_1), Z (s_2) \right ) = C_\xi \left | s_1 \right |^\alpha + C_\xi \left | s_2 \right |^\alpha - C_\xi \left | s_1 - s_2 \right |^\alpha. 
\end{gather*}
\label{lemma: high excursion probability - discrete}
\end{lemma}

\begin{proof}
See Lemma 12.2.1 in \cite{leadbetter2012extremes}. 
\end{proof}

\begin{lemma} 
Let $\left \{ B \left ( t \right ) \right \}_{t>0}$ be standard Brownian motion and define the function $F \left ( \cdot \right )$ as follows. 

\begin{equation*}
F \left ( x \right ) =  \lim_{T \rightarrow \infty} \frac{1}{T} \mathbb{E} \left [ \exp \left ( \sup_{s \in [0,T] \cap x \mathbb{Z}} \left \{ B (s) - s/2 \right \} \right )  \right ]
\end{equation*}

(i) For $x > 0$ it holds that $F \left ( x \right ) = p^2_\infty \left ( x \right ) / x$ where $p_\infty \left ( \cdot \right )$ is defined as follows. 
\begin{equation*}
p_\infty \left ( x \right ) = \exp \left ( - \sum_{k=1}^\infty \frac{1}{k} \bar{\Phi} \left ( \sqrt{kx / 4} \right ) \right )
\end{equation*}

(ii) Putting $G \left ( y \right ) = \left ( 1 / y \right ) F \left ( C / y \right )$ for any fixed $C > 0$ it holds that $G \left ( y \right ) \sim 1/2y$ as $y \rightarrow \infty$. 

\label{lemma: F fuction}
\end{lemma}

\begin{proof}
See Theorem 7.2 and Corollary 3.18 respectively in \cite{kabluchko2007extreme}. 
\end{proof}


\subsection{Intermediate results}

\begin{lemma}

Let $\left \{ B \left ( t \right ) \right \}_{t>0}$ be standard Brownian motion and define the process $\left \{ \xi \left ( t \right ) \right \}_{t>0}$ as follows. 

\begin{align*}
& \xi \left ( l \right ) = \left \{ \left ( \frac{1}{p+2} \right ) \sum_{i=0}^{p+1} \binom{p+1}{i}^2 \right \} ^ {-1/2} \sum_{j=0}^{p+1} \left ( -1 \right ) ^ {p+1-j} \binom{p+1}{j} \mathcal{Y}_{l,j} \\ 
& \mathcal{Y}_{l,j} = \left [  B \left ( l + \frac{j+1}{p+2} \right ) - B \left ( l + \frac{j}{p+2} \right ) \right ]
\end{align*}

(i) The process $\left \{ \xi \left ( l \right ) \right \}_{l>0}$ is the continuous time analogue of $\frac{1}{\sigma} D_{l,w}^p \left ( Y \right )$ under Assumption \ref{assumption: gaussian noise} and the null of no change points, in the sense that for a given scale $w$ the following holds.
\begin{equation*}
\left \{ \frac{1}{\sigma} D_{l,w}^p \left ( Y \right ) \mid 1 \leq l \leq n - w \right \} \stackrel{d}{=} \left \{ \xi \left ( l / w \right ) \mid 1 \leq l \leq n - w \right \} 
\end{equation*}

(ii) According to Definition \ref{defintion: stationary normal process} the process is locally stationary with index $\alpha = 1$ and local structure $C_p$ defined as follows. 
\begin{equation*}
C_p = \left ( p + 2 \right ) \left ( 1 + \frac{ \sum_{j=1}^{p+1} \binom{p+1}{j} \binom{p+1}{j-1} }{ \sum_{j=0}^{p+1} \binom{p+1}{j} ^ 2 } \right )
\end{equation*}

\label{lemma: continuous time analogue}
\end{lemma}

\begin{proof}
Part (i) can be verified by inspection. To show part (ii) note that for all $l > 0$ we have $\mathbb{E} \left ( \xi \left ( l \right ) \right ) = 0$ and $\mathbb{E} \left ( \xi ^ 2 \left ( l \right ) \right ) = 1$, so it remains to calculate the covariance between $\xi \left ( l \right ) $ and $\xi \left ( l + s_l \right )$ for $| s_l | \rightarrow 0$. First, taking $s_l > 0$ we have the following. 
\begin{align*}
\text{Cov} \left ( \xi (l), \xi (l + s_l) \right ) & = \left ( \left ( \frac{1}{p+2} \right ) \sum_{i=0}^{p+1} \binom{p+1}{i}^2 \right ) ^{-1} \sum_{j=0}^{p+1} \sum_{k=0}^{p+1} \binom{p+1}{j} \binom{p+1}{k} \left ( -1 \right ) ^ {j + k} \text{Cov} \left ( \mathcal{Y}_{l,j}, \mathcal{Y}_{l + s_l, k} \right ) \\
& = \left ( \left ( \frac{1}{p+2} \right ) \sum_{i=0}^{p+1} \binom{p+1}{i}^2 \right ) ^ {-1} \left \{ \sum_{j=0}^{p+1} \binom{p+1}{j}^2 \text{Cov} \left ( \mathcal{Y}_{l,j},\mathcal{Y}_{l + s_l,j} \right ) + \dots \right . \\
& \dots + \left . \sum_{j=1}^{p+1} \left ( -1 \right ) \binom{p+1}{j} \binom{p+1}{j-1} \text{Cov} \left ( \mathcal{Y}_{l,j},\mathcal{Y}_{l + s_l,j-1} \right ) \right \}
\end{align*}
Using the fact that $\text{Cov} \left ( B (l_1), B (l_1) \right ) = \min \left ( l_1, l_2 \right )$ gives the following. 
\begin{align*}
& \text{Cov} \left ( \mathcal{Y}_{l,j}, \mathcal{Y}_{l + s_l,j} \right ) = \frac{1}{p+2} - s_l \\ 
& \text{Cov} \left ( \mathcal{Y}_{l,j}, \mathcal{Y}_{l + s_l,j-1} \right ) = s_l
\end{align*}
Therefore for $s_l \rightarrow 0$ with $s_l > 0$ we have the following. 
\begin{equation*}
\text{Cov} \left ( \xi (l), \xi (l + s_l) \right ) =  1 - \left ( p + 2 \right ) \left ( 1 + \frac{ \sum_{j=1}^{p+1} \binom{p+1}{j} \binom{p+1}{j-1} }{ \sum_{j=0}^{ p+1} \binom{p+1}{j} ^ 2 } \right ) s_l
\end{equation*}
The same calculations can be repeated for the case $s_l < 0$ and so ultimately we have that $\text{Cov} \left ( \xi (l), \xi (l + s_l) \right ) = 1 - C_p |s_l|$ as $|s_l| \rightarrow 0$. 

\end{proof}

\begin{lemma}
Consider the problem of testing for the presence of a change point on the interval $I = \left \{ 1, \dots , m \right \}$ where $m$ satisfies $(p+2) \delta \leq m < (p+2) (\delta + 1)$ for some integer $\delta > 1$. If the interval contains a single change point at location $\delta$ with change sizes $\Delta_{0}, \dots, \Delta_{p}$ then the test 
\begin{equation*}
T_{1,m}^{\lambda} = \boldsymbol{1} \left \{ |D_{1,(p+2) \delta}^p \left ( \boldsymbol{Y} \right )| > \lambda \right \}
\end{equation*}
with threshold $\lambda = \widehat{\tau} \times \bar{\lambda}$, for some $\bar{\lambda} > 0$, will detect the change on the event
\begin{equation}
\left \{ L^{\widehat{\tau}}_{\mathcal{G} \left ( W , a \right )} \left ( \boldsymbol{\zeta} \right ) \leq \bar{\lambda} \right \} \cap \left \{ \widehat{\tau} < 2 \tau \right \}
\label{equation: tau bound event}
\end{equation}
as long Assumption \ref{assumption: One prominent jump} is satisfied and $\delta'$ satisfied the inequality
\begin{align*}
\delta > n^{\frac{2p^*}{2p^*+1}} \left ( \frac{16 C_{p,p^*}^2 \tau^2 \bar{\lambda}^2}{\left | \Delta_{p^*} \right |^2}\right ) ^ {\frac{1}{2p^* + 1}},
\end{align*}
where 
\begin{equation*}
C_{p,p^*} = 2^{p^*+2}(p^*+2)\sqrt{\sum_{i=0}^{p+1}\binom{p+1}{i}^2}.
\end{equation*}
\label{lemma: change point detection on an interval}
\end{lemma}

\begin{proof}
By the linearity of the difference operator and the triangle inequality the change will be detected if the following occurs. 
\begin{equation}
\left | D_{1,m}^p \left ( \boldsymbol{f} \right ) \right | > \left | D_{1,m}^p \left ( \boldsymbol{\zeta} \right ) \right | + \lambda 
\label{equation: D trinagle inequality bound}
\end{equation}
Moreover on (\ref{equation: tau bound event}) we must have that $\left | D_{1,m}^p \left ( \boldsymbol{\zeta} \right ) \right | + \lambda < 4 \tau \bar{\lambda}$. Writing $B_k$ for the $k$-th Bernoulli number we have the following by Faulhaber's formula for any integers $p > 0$ and $\delta > 1$. 
\begin{align}
\frac{1}{\delta'} \sum_{t=1}^{\delta'} \left ( 1 - t / \delta' \right ) ^ p & = (\delta')^{-(p+1)} \sum_{s=1}^{\delta'-1} s^p \nonumber \\ 
& = \left ( \frac{1}{p+1} \right ) \left ( \frac{\delta'-1}{\delta'} \right ) ^ {p+1} \sum_{k=0}^{p} \binom{p+1}{k} B_k \left ( \delta' - 1 \right ) ^ {-k} \nonumber \\ 
& \geq \left ( \frac{1}{p+1} \right ) \left ( \frac{\delta'-1}{\delta'} \right ) ^ {p+1} \nonumber \\ 
& \geq \frac{1}{2^{p+1}(p+1)}.
\label{equation: g1 bound}
\end{align}
Using the above along with Assumption \ref{assumption: One prominent jump} and the fact that the test statistic (\ref{equation: local test statistic}) is invariant to the addition of arbitrary degree $p$ polynomials we have the following.  
\begin{align}
\left | D_{1,m}^p \left ( \boldsymbol{f} \right ) \right | & = \left \{ \delta \sum_{i=0}^{p+1} \binom{p+1}{i}^2 \right \} ^ {-\frac{1}{2}}\left | \sum_{j=0}^p \Delta_j \sum_{t=1}^\delta \left ( \frac{t}{n} - \frac{\delta}{n} \right )^j \right | \nonumber \\ 
& \geq \sqrt{\delta} \left | \Delta_{p^*} \right | \left ( \frac{\delta}{n} \right ) ^ {p^*} \left [ \frac{\frac{1}{\delta} \sum_{t=1}^\delta \left ( 1 - \frac{t}{\delta} \right ) ^ {p^*}}{\sqrt{\sum_{i=0}^{p+1} \binom{p+1}{i}^2}} \right ] - \sum_{\substack{0 \leq j \leq p \\ j \neq p^*}} \sqrt{\delta} \left | \Delta_j \right | \left ( \frac{\delta}{n} \right ) ^ j \left [ \frac{\frac{1}{\delta} \sum_{t=1}^\delta \left ( 1 - \frac{1}{\delta} \right )^j}{{\sqrt{\sum_{i=0}^{p+1} \binom{p+1}{i}^2}}} \right ] \nonumber \\ 
& \geq C_{p,p^*}^{-1} \left | \Delta_{p^*} \right | \delta^{\frac{2p^*+1}{2}} n^{-p^*} \label{equation: D f bound}
\end{align}
Therefore combining (\ref{equation: D trinagle inequality bound}) and (\ref{equation: D f bound}) we have that on the event (\ref{equation: tau bound event}) the change will be detected if $C_{p,p^*}^{-1} \left | \Delta_{p^*} \right | \delta^{\frac{2p^*+1}{2}} n^{-p^*} > 4 \tau \bar{\lambda}$, and the desired result follows by rearranging. 
\end{proof}

\begin{theorem}
Put $w = \left \lfloor c \log (n) \right \rfloor$ for some constant $c > 0$ and introduce maximum of the local test statistics (\ref{equation: local test statistic}) appropriately standardised and restricted to scales $w$ as follows. 
\begin{equation*}
M_{c\log(n)}^\sigma \left ( \mathbf{Y} \right ) = \max \left \{ \frac{1}{\sigma} D_{l,w}^p \left ( \mathbf{Y} \right ) \mid 1 \leq l \leq n - w \right \}    
\end{equation*}
Then under Assumption \ref{assumption: gaussian noise} and the null of no change points for any fixed $x \in \mathbb{R}$ the following holds, where $\mathfrak{a}_{n}$ and $\mathfrak{b}_{n}$ are defined as in Theorem \ref{theorem: Gaussian tightness of normalized maximum}. 
\begin{equation*}
\mathbb{P} \left ( \mathfrak{a}_{n} M_{c\log(n)}^\sigma \left ( \mathbf{Y} \right ) - \mathfrak{b}_{n} \leq x \right ) \sim \exp \left ( - \left ( \frac{2 C_p}{c} \right ) F \left ( \frac{2 C_p}{c} \right  ) e^{-x} \right ) 
\end{equation*}
\label{theorem: distribution on optimal scales}
\end{theorem}

\begin{proof}
Omitting dependence on $x$ introduce the following notation. 
\begin{equation*}
\mathfrak{u}_n = \sqrt{2 \log (n)} + \left ( -\frac{1}{2} \log \log (n) - \log \left ( 2 \sqrt{\pi} \right ) + x \right ) / \sqrt{2 \log (n)}
\end{equation*}
For some $\rho \in \left ( 0 , 1 \right )$ we decompose the index set $\left \{ 1, \dots, n \right \}$ into disjoint blocks $A_0, B_0, A_1, B_1, \dots$ respectively of size $w$ and $w^\rho$ defined as follows.
\begin{align*}
& A_i = \left \{ l \mid i \left (w + w^\rho \right ) <  l \leq (i+1) w + i w^\rho  \right \} \\ 
& B_i = \left \{ l \mid (i+1) w + i w^\rho < l \leq (i+1)(w + w^\rho) \right \}
\end{align*}
The proof proceeds in three steps. 

\textbf{STEP 1:} we first show that the behaviour of small blocks is asymptotically unimportant for the maximum. Putting $\mathcal{B}_n = \cup_i B_i$ and using the fact $|\mathcal{B}_n| \sim n w^\rho / (w + w^\rho)$ and $\mathfrak{u}_n^2 = 2 \log (n) - \log \log (n) + \mathcal{O} \left ( 1 \right )$ the following holds. 
\begin{align*}
\mathbb{P} \left ( \max_{l \in \mathcal{B}_n} \left \{ \frac{1}{\sigma} D_{l,w}^p \left ( \mathbf{Y} \right ) \right \} > \mathfrak{u}_n \right ) & \leq \sum_{l \in \mathcal{B}_n} \mathbb{P} \left ( \frac{1}{\sigma} D_{l,w}^p \left ( \mathbf{Y} \right ) > \mathfrak{u}_n \right ) \\ 
& = \left | \mathcal{B}_n \right | \bar{\Phi} \left ( \mathfrak{u}_n \right ) \\
& \cleq \frac{w^\rho}{w + w^\rho} 
\end{align*}

\textbf{STEP 2:} next we show that the any dependence between larger blocks is asymptotically unimportant for the the maximum. Write 
\begin{equation*}
\Lambda_{l_1, l_2} = \text{Cov} \left ( \frac{1}{\sigma} D_{l_1,w}^p \left ( \mathbf{Y} \right ), \frac{1}{\sigma} D_{l_2,w}^p \left ( \mathbf{Y} \right ) \right ),    
\end{equation*}
and let $\sigma^{-1}\tilde{D}_{l,w}^p \left ( \mathbf{Y} \right ) $ be random variables with the same marginal distributions as $\sigma^{-1}D_{l, w}^p \left ( \mathbf{Y} \right ) $ and covariances as shown below. 
\begin{equation*}
\tilde{\Lambda}_{l_1,l_2} = 
\begin{cases}
\Lambda_{l_1,l_2} & l_1 \in A_{i_1}, l_2 \in A_{i_2} \text{ with } i_1 = i_2 \\
0 & \text{else}
\end{cases}
\end{equation*}
For any $l_1, l_2$ write $j_{1,2} = \left | \left \{ l_1, \dots, l_1 + w - 1  \right \} \cap \left \{ l_2, \dots, l_2 + w - 1  \right \} \right |$ and put $\Lambda_{l_1, l_2} = \Lambda_{j_{1,2}}$. Writing $\mathcal{A}_n = \cup_i A_i$ and using Lemma \ref{lemma: normal comparison lemma} we have the following. 
\begin{align}
& \left | \mathbb{P} \left ( \max_{l \in \mathcal{A}_n} \left \{ \frac{1}{\sigma} D_{l, w}^p \left ( \mathbf{Y} \right ) \right \} \leq \mathfrak{u}_n \right ) - \mathbb{P} \left ( \max_{l \in \mathcal{A}_n} \left \{ \frac{1}{\sigma} \tilde{D}_{l, w}^p \left ( \mathbf{Y} \right ) \right \} \leq \mathfrak{u}_n \right ) \right | \label{equation: normal comparison bound} \\
& \leq \frac{1}{2 \pi} \sum_{\substack{l_1 \in A_{i},  l_2 \in A_{j} \\ i \neq j}} \left | \Lambda_{l_1, l_2} - \tilde{\Lambda}_{l_1, l_2} \right | \left ( 1 - \Lambda_{l_1, l_2}^2 \right ) ^ {-1/2} \exp \left ( - \frac{\mathfrak{u}_n^2}{ 1 + \Lambda_{l_1, l_2}} \right ) \nonumber \\ 
& \cleq \sum_{i = 0}^{ |\mathcal{A}_n| / |A_0|} \sum_{\substack{l_1 \in A_i \\ l_2 \in A_{i+1}}} \left | \Lambda_{l_1, l_2} - \tilde{\Lambda}_{l_1, l_2} \right | \left ( 1 - \Lambda_{l_1, l_2}^2 \right ) ^ {-1/2} \exp \left ( - \frac{\mathfrak{u}_n^2}{ 1 + \Lambda_{l_1, l_2}} \right ) \nonumber \\ 
& \cleq \frac{\left | \mathcal{A}_n \right |}{|A_0|} \sum_{l = 1}^{|A_0|} \sum_{j=1}^l \Lambda_j \left ( 1 - \Lambda_j^2 \right )^{-1/2} \exp \left ( - \frac{2 \log (n) - \log \log (n)}{ 1 + \Lambda_j} \right ) \nonumber \\ 
& \cleq \log (n) \frac{\left | \mathcal{A}_n \right |}{|A_0|} \sum_{l = 1}^{|A_0|} \sum_{j=1}^l \Lambda_j \left ( 1 - \Lambda_j^2 \right )^{-1/2} \exp \left ( - \frac{2 \log (n)}{1 + \Lambda_j} \right ) \nonumber
\end{align}
Note that for some fixed $K > 0$ depending on $p$ it must hold that $\Lambda_{j} \leq \min\left ( j K,  w - w^\rho \right ) / w$. Therefore the first term after the double sum can be bounded as follows. 
\begin{align}
\Lambda_j \left ( 1 - \Lambda_j^2 \right )^{-1/2} & \leq \Lambda_j \left ( 1 - \Lambda_j \right )^{-1/2} \nonumber \\
& \leq \min \left (jK, w - w^\rho \right ) / \sqrt{\left ( w - \min \left (jK, w - w^\rho \right ) \right ) w} \nonumber \\
& \leq \min \left (jK, w - w^\rho \right ) / \sqrt{w}
\label{equation: double lambda bound}
\end{align}
For the exponential term put $2 / (1 + \Lambda_j) = 1 + \delta_j$. The following holds. 
\begin{align}
\delta_j & = \left ( 1 - \Lambda_j \right ) / \left ( 1 + \Lambda_j \right ) \nonumber \\
& \geq \left ( w - \min \left ( jK, w - w^\rho \right ) \right ) / \left ( w + \min \left ( jK, w - w^\rho \right ) \right ) \nonumber \\ 
& \geq \left ( w - \min \left ( jK, w - w^\rho \right ) \right ) / 2 w
\label{equation: exponential term bound}
\end{align}
Therefore substituting (\ref{equation: double lambda bound}) and (\ref{equation: exponential term bound}) into (\ref{equation: normal comparison bound}) we obtain the following. 
\begin{align}
(\ref{equation: normal comparison bound}) & \cleq \frac{\sqrt{\log(n)}}{n} \frac{|\mathcal{A}_n|}{|A_0|} \sum_{j=1}^{l} \min \left ( jK, w -w^\rho \right ) \left ( n ^{\frac{1}{2 w}} \right )^{- (w - \min \left ( jK, w -w^\rho \right ))} \nonumber \\ 
& = \frac{\sqrt{\log(n)}}{n} \frac{|\mathcal{A}_n|}{|A_0|} \left \{ \sum_{l=1}^{\left \lfloor |A_0| / K \right \rfloor} \sum_{j=1}^l jK \left ( n ^{\frac{1}{2 w}} \right )^{- (w - jK)} + \dots \right . \nonumber \\ 
& \dots + \left . \sum_{l= \left \lfloor |A_0| / K \right \rfloor + 1}^{|A_0|} \sum_{j=1}^l \left ( w -w^\rho \right ) \left ( n ^{\frac{1}{2 w}} \right )^{w^\rho} \right \}
\label{equation: split sums and bound}
\end{align}
The first sum in (\ref{equation: split sums and bound}) can be bounded as follows.  
\begin{align}
\sum_{l=1}^{\left \lfloor |A_0| / K \right \rfloor} \sum_{j=1}^l jK \left ( n ^{\frac{1}{2 w}} \right )^{- (w - jK)} & \cleq n^{-1/2} \int_{1}^{\left \lfloor |A_0| / K \right \rfloor + 1} \int_{1}^{y+1} x \left ( n^{\frac{1}{2w}} \right )^{Kx} \mathrm{d} x \mathrm{d} y \nonumber \\ 
& \cleq w n^{-w^{-(1-\rho)}/2}
\label{equation: bound on first sum}
\end{align}
The second sum in (\ref{equation: split sums and bound}) can be bounded as follows.  
\begin{align}
\sum_{l= \left \lfloor |A_0| / K \right \rfloor + 1}^{|A_0|} \sum_{j=1}^l \left ( w -w^\rho \right ) \left ( n ^{\frac{1}{2 w}} \right )^{w^\rho} & \cleq w n^{-w^{-(1-\rho)}/2} \sum_{l= \left \lfloor |A_0| / K \right \rfloor + 1}^{|A_0|} \left ( l \right ) \nonumber \\
& \cleq w^3 n^{-w^{-(1-\rho)}/2}
\label{equation: bound on second sum}
\end{align}
Finally plugging (\ref{equation: bound on first sum}) and (\ref{equation: bound on second sum}) into (\ref{equation: normal comparison bound}) and using the fact that $|\mathcal{A}_n| / | A_0 | \sim n / (w + w^\rho)$ we obtain the following for some $C > 0$ depending on $\rho$ as long as $n$ is sufficiently large. 
\begin{align*}
(\ref{equation: normal comparison bound}) & \cleq \frac{\sqrt{\log (n)}}{n} \frac{|\mathcal{A}_n|}{|A_0|} \left \{ n^{-w^{-(1-\rho)}/2} \left ( w + w^3 \right ) \right \} \\ & \cleq \log^{5/2} (n) n^{-w^{-(1-\rho)}/2} \\
& \cleq \exp \left ( -C \log ^ \rho (n) \right )
\end{align*}

\textbf{STEP 3:} we now prove Theorem \ref{theorem: distribution on optimal scales}. Using Lemma \ref{lemma: high excursion probability - discrete} and part (i) of Lemma \ref{lemma: continuous time analogue} and noting that $\mathfrak{u}_n^2 / w \sim 2 / c$ gives the following for any $i = 0, \dots, \left | \mathcal{A}_n \right | -1$. 
\begin{equation}
\mathbb{P} \left ( \max_{l \in A_i} \left \{ \frac{1}{\sigma} D_{l,w}^p \left ( \mathbf{Y} \right ) \right \} > \mathfrak{u}_n \right ) \sim \left ( \frac{w}{n} \right ) \left ( \frac{2 C_p}{c} \right ) F \left ( \frac{2 C_p}{c} \right  ) e^{-x}
\label{equation: M log n limit}
\end{equation}
The following inequality is evident.
\begin{equation*}
\mathbb{P} \left ( M_{c\log (n)}^\sigma \left ( \mathbf{Y} \right ) \leq \mathfrak{u}_n \right ) \leq  \mathbb{P} \left ( \max_{l \in \mathcal{A}_n}  \left \{ \frac{1}{\sigma} D_{l,w}^p \left ( \mathbf{Y} \right ) \right \} \leq \mathfrak{u}_n \right ) 
\end{equation*}
Therefore (\ref{equation: M log n limit}), the results of step 2, and that $|\mathcal{A}_n| / |A_0| \sim n / w$ imply the following. 
\begin{align*}
\lim_{n \rightarrow \infty} \mathbb{P} \left ( M_{c \log (n)}^\sigma \left ( \mathbf{Y} \right ) \leq \mathfrak{u}_n \right ) & \leq \lim_{n \rightarrow \infty} \left \{ \mathbb{P} \left ( \max_{l \in \mathcal{A}_n} \left \{ \frac{1}{\sigma} \tilde{D}_{l , w}^p \left ( \mathbf{Y} \right ) \right \} \leq \mathfrak{u}_n \right ) + \mathcal{O} \left ( \exp \left ( -C\log^\rho (n) \right ) \right ) \right \} \\
& = \lim_{n \rightarrow \infty} \left ( 1 - \left ( \frac{w}{n} \right ) \left ( \frac{2 C_p}{c} \right ) F \left ( \frac{2 C_p}{c} \right  ) e^{-x} \right ) ^ {|\mathcal{A}_n| / |A_0|} \\
& = \exp \left ( - \left ( \frac{2 C_p}{c} \right ) F \left ( \frac{2 C_p}{c} \right  ) e^{-x} \right ) 
\end{align*}
Going the other way the following inequality is also evident. 
\begin{equation*}
\mathbb{P} \left ( M_{c\log (n)}^\sigma \left ( \mathbf{Y} \right ) \leq \mathfrak{u}_n \right ) \geq  \mathbb{P} \left ( \max_{l \in \mathcal{A}_n}  \left \{ \frac{1}{\sigma} D_{l,w}^p \left ( \mathbf{Y} \right ) \right \} \leq \mathfrak{u}_n \right ) - \mathbb{P} \left ( \max_{l \in \mathcal{B}_n}  \left \{ \frac{1}{\sigma} D_{l,w}^p \left ( Y \right ) \right \} > \mathfrak{u}_n \right )
\end{equation*}
Using (\ref{equation: M log n limit}) and the results of Steps 1 and 2 gives the following. 
\begin{align*}
\lim_{n \rightarrow \infty} \mathbb{P} \left ( M_{c \log (n)}^\sigma \left ( \mathbf{Y} \right ) \leq \mathfrak{u}_n \right ) & \geq \lim_{n \rightarrow \infty} \left \{ \mathbb{P} \left ( \max_{l \in \mathcal{A}_n} \left \{ \frac{1}{\sigma} \tilde{D}_{l , w}^p \left ( Y \right ) \right \} \leq \mathfrak{u}_n \right ) \dots \right . \\
& \left . \dots - \mathcal{O} \left ( \exp \left ( -C\log^\rho (n) \right ) \right ) - \mathcal{O} \left ( \frac{w^\rho}{w+w^\rho} \right ) \right \} \\
& = \lim_{n \rightarrow \infty} \left ( 1 - \left ( \frac{w}{n} \right ) \left ( \frac{2 C_p}{c} \right ) F \left ( \frac{2 C_p}{c} \right  ) e^{-x} \right ) ^ {|\mathcal{A}_n| / |A_0|} \\
& = \exp \left ( - \left ( \frac{2 C_p}{c} \right ) F \left ( \frac{2 C_p}{c} \right  ) e^{-x} \right )
\end{align*}
Therefore, the theorem is proved. 
\end{proof}

\subsection{Proof of Theorem \ref{theorem: Gaussian tightness of normalized maximum}}

\begin{proof}
Given the result in part (i), part (ii) follows immediately from Lemma \ref{lemma: convergence with other scaling sequcnes}. For the proof of part (i) write $k_n = \left \lfloor \log_a (W) \right \rfloor$ and for some $A > 0$ introduce the restrictions of the $a$-adic grid defined in (\ref{equation: a-adic grid}) to scales no larger than $Wa^A$. 
\begin{align*}
& \mathcal{G}_{-} \left ( A \right ) = \left \{ \left ( l , w \right ) \in \mathbb{N}^2 \mid w \in \mathcal{W}_{-} (A) , 1 \leq l \leq n - w \right \} \\
& \mathcal{W}_{-} \left ( A \right ) = \left \{ w = \left \lfloor a^k \right \rfloor \mid k_n \leq k \leq k_n + A \right \}
\end{align*}
Introduce also the restriction of (\ref{equation: a-adic grid}) to scales strictly larger than $Wa^A$. 
\begin{align*}
& \mathcal{G}_{+} \left ( A \right ) = \left \{ \left ( l , w \right ) \in \mathbb{N}^2 \mid w \in \mathcal{W}_{+} (A) , 1 \leq l \leq n - w \right \} \\
& \mathcal{W}_{+} \left ( A \right ) = \left \{ w = \left \lfloor a^k \right \rfloor \mid k_n + A < k \leq \left \lfloor \log_a (n / 2) \right \rfloor \right \} 
\end{align*}
The proof proceeds in four steps. 

\textbf{STEP 1:} we first show that the behaviour of the tests statistic on large scales is asymptotically unimportant for the maximum. Making use of lemma \ref{lemma: high excursion probability - continuous} we have the following. 
\begin{align*}
& \mathbb{P} \left ( \max_{(l,w) \in \mathcal{G}_{+} (A)} \left \{ \frac{1}{\sigma} D_{l,w}^p \left ( \mathbf{Y} \right ) \right \} > \mathfrak{u}_n \right ) \\
& \leq \sum_{k = k_n + A} ^ {\left \lfloor \log_{a} (n / 2) \right \rfloor} \sum_{i=0}^{\left \lfloor n / a^k \right \rfloor -1} \mathbb{P} \left ( \max \left \{ \frac{1}{\sigma} D_{l, \left \lfloor a^k \right \rfloor}^p \left ( \mathbf{Y} \right ) \mid i \times \left \lfloor a^k \right \rfloor < l \leq (i+1) \times \left \lfloor a^k \right \rfloor \right \} > \mathfrak{u}_n \right ) \\ 
& \leq \sum_{k = k_n + A} ^ {\left \lfloor \log_{a} (n / 2) \right \rfloor} \left ( \frac{n}{a^k} \right ) \mathbb{P} \left ( \sup_{t \in [0,1)} \left \{ \xi \left ( t \right ) \right \} > \mathfrak{u}_n \right ) \\
& \cleq \sum_{k = k_n + A} ^ {\left \lfloor \log_{a} (n / 2) \right \rfloor} \left ( \frac{n}{a^k} \right ) \mathfrak{u}_n e^{-\mathfrak{u}_n^2/2} \\ 
& \cleq \frac{a^{-A}}{1 - a^{-1}}
\end{align*}
Finally, sending $A \rightarrow \infty$ the claim is proved. 

\textbf{STEP 2:} next we show that for any fixed $A$ the dependence between maxima occurring over different scales in $\mathcal{W}_{-} (A)$ is asymptotically unimportant for the overall maximum. Write
\begin{equation*}
\Lambda_{l_1,w_1,l_2,w_2} = \text{Cov} \left ( \frac{1}{\sigma} D_{l_1,w_1}^p ( \mathbf{Y} ), \frac{1}{\sigma} D_{l_2,w_2}^p ( \mathbf{Y} ) \right ),
\end{equation*}
and let $\sigma^{-1} \tilde{D}_{l,w}^{(p)} \left ( \mathbf{Y} \right )$ be random variables with the same marginal distribution as $\sigma^{-1} D_{l,w}^p \left ( \mathbf{Y} \right )$ and covariance as shown below. 
\begin{equation*}
\tilde{\Lambda}_{l_1, l_2, w_1, w_2} = 
\begin{cases}
\Lambda_{l_1, l_2, w_1, w_2}  & \text{if } w_1 = w_2 \\
0 & \text{else} 
\end{cases}
\end{equation*}
Note that for each $a > 1$ there will be a $\Lambda_a \in (0,1)$ depending only on $a$ such that for any $w_1 \neq w_2$ and all permissible $l_1, l_2$ it holds that $\Lambda_{l_1,w_1,l_2,w_2} \leq \Lambda_a$. Therefore using Lemma \ref{lemma: normal comparison lemma} we have the following. 
\begin{align*}
& \left | \mathbb{P} \left ( \max_{(l,w) \in \mathcal{G}_{-} (A)} \left \{ \frac{1}{\sigma} D_{l,w}^p \left ( \mathbf{Y} \right ) \right \} \leq \mathfrak{u}_n \right ) - 
\mathbb{P} \left ( \max_{(l,w) \in \mathcal{G}_{-} (A)} \left \{ \frac{1}{\sigma} \tilde{D}_{l,w}^p \left ( \mathbf{Y} \right ) \right \} \leq \mathfrak{u}_n \right ) \right | \\
& \leq \frac{1}{2\pi} \sum_{\substack{w_1, w_2 \in \mathcal{W}_{-} \left ( A \right ) \\ w_1 \neq w_2}} \sum_{\substack{1 \leq l_1 \leq n - w_1 \\ 1 \leq l_2 \leq n - w_2}} \left | \Lambda_{\substack{l_1 , w_1 \\ l_2, w_2}} - \tilde{\Lambda}_{\substack{l_1 , w_1 \\ l_2, w_2}} \right | \left ( 1 - \Lambda^2_{\substack{l_1, w_1 \\ l_2, w_2}} \right ) ^ {-1/2} \exp \left ( - \frac{\mathfrak{u}_n^2}{1 + \Lambda^2_{\substack{l_1, w_1 \\ l_2, w_2}}} \right ) \\ 
& \cleq \sum_{\substack{w_1, w_2 \in \mathcal{W}_{-} \left ( A \right ) \\ w_1 \neq w_2}} \sum_{\substack{1 \leq l_1 \leq n - w_1 \\ 1 \leq l_2 \leq n - w_2 \\ |l_1 - l_2| < \max \left ( w_1, w_2 \right )}} \Lambda_{\substack{l_1 , w_1 \\ l_2, w_2}} \left ( 1 - \Lambda_{\substack{l_1 , w_1 \\ l_2, w_2}}^2 \right ) ^ {-1/2} \exp \left ( - \frac{\mathfrak{u}_n^2}{1 + \Lambda^2_{\substack{l_1, w_1 \\ l_2, w_2}}} \right )  \\ 
& \cleq \log (n) \sum_{\substack{w_1, w_2 \in \mathcal{W}_{-} (A) \\ w_1 \neq w_2}} \sum_{\substack{1 \leq l_1 \leq n - w_1 \\ 1 \leq l_2 \leq n - w_2 \\ |l_1 - l_2| < \max \left ( w_1, w_2 \right )}} \left ( \frac{\Lambda_a}{ \sqrt{1 - \Lambda_a^2}} \right ) \exp \left ( - \frac{2 \log (n)}{1 + \Lambda_a} \right ) \\
& \cleq \left ( 1 + A \right ) ^ 2 a^A \log^2 (n) \times n^{-\frac{1-\Lambda_a}{1 + \Lambda_a}}
\end{align*}
Since $\Lambda_a < 1$ the statement is proved.

\textbf{STEP 3:} we now show that if we pass to a sub-sequence of $n$'s on which the quantity $b_{n} = a^{\left \lfloor \log_{a} (W) \right \rfloor} / W$ converges to some constant $b$ the sequence of normalised maxima 
\begin{equation}
\left \{ \mathfrak{a}_n M_{\mathcal{G}(W,a)}^\sigma \left ( \mathbf{Y} \right ) - \mathfrak{b}_{n} \mid n \in \mathbb{N} \right \}
\label{equation: sequence of normalised maxima}
\end{equation}
converges weakly to a Gumbel distribution. On such a sub-sequence for each $j \in \mathbb{N}$ we have that $a^{k_n + j} \sim a^j b d \times \log (n)$. Therefore from Theorem \ref{theorem: distribution on optimal scales} we have the following.  
\begin{equation*}
\mathbb{P} \left ( \max_{1 \leq l \leq n - \left \lfloor a^{k_n + j} \right \rfloor} \left \{ \frac{1}{\sigma} D_{l,\left \lfloor a^{k_n + j} \right \rfloor}^p \left ( \mathbf{Y} \right ) \right \} \leq \mathfrak{u}_n \right ) \sim \exp \left ( - \left ( \frac{2 C_{p}}{a^j b d} \right ) F \left ( \frac{2 C_{p}}{a^j b d} \right ) e^{-\tau} \right ) 
\label{equation: fixed j subseq limit}
\end{equation*}
The following inequality is evident.
\begin{equation*}
\mathbb{P} \left ( M_{\mathcal{G}(W,a)}^\sigma \left ( \mathbf{Y} \right ) \leq \mathfrak{u}_n \right ) \leq \mathbb{P} \left ( \max_{(l,w) \in \mathcal{G}_{-}(A)} \left \{ \frac{1}{\sigma} D_{l,w}^p \left ( \mathbf{Y} \right ) \right \} \leq \mathfrak{u}_n \right )
\end{equation*}
Therefore (\ref{equation: fixed j subseq limit}) and the result from step 2 imply the following.
\begin{equation*}
\limsup_{n \rightarrow \infty} \mathbb{P} \left ( M_{\mathcal{G}(W,a)}^\sigma \left ( \mathbf{Y} \right ) \leq \mathfrak{u}_n \right ) \leq \exp \left ( - \sum_{j=0}^\infty \left ( \frac{2 C_{p}}{a^j b d} \right ) F \left ( \frac{2 C_{p}}{a^j b d} \right ) e^{-x} \right )
\end{equation*}
Note that because $a>1$ by part (ii) of Lemma \ref{lemma: F fuction} the above sum converges. Going the other way the following inequality is also evident.  
\begin{align*}
\mathbb{P} \left ( M_{\mathcal{G}(W,a)}^\sigma \left ( \mathbf{Y} \right ) \leq \mathfrak{u}_n \right ) \geq \mathbb{P} \left ( \max_{(l,w) \in \mathcal{G}_{-}(A)} \left \{ \frac{1}{\sigma} D_{l,w}^p \left ( \mathbf{Y} \right ) \right \} \leq \mathfrak{u}_n \right ) - \mathbb{P} \left ( \max_{(l,w) \in \mathcal{G}_{+}(A)} \left \{ \frac{1}{\sigma} D_{l,w}^p \left ( \mathbf{Y} \right ) \right \} > \mathfrak{u}_n \right )
\end{align*}
Therefore (\ref{equation: fixed j subseq limit}) and the result from steps 1 and 2 imply the following. 
\begin{equation*}
\liminf_{n \rightarrow \infty} \mathbb{P} \left ( M_{\mathcal{G}(W,a)}^\sigma \left ( \mathbf{Y} \right ) \leq \mathfrak{u}_n \right ) \geq \exp \left ( - \sum_{j=0}^\infty \left ( \frac{2 C_{p}}{a^j b d} \right ) F \left ( \frac{2 C_{p}}{a^j b d} \right ) e^{-x} \right )
\end{equation*}
Therefore, the statement is proved. 

\textbf{STEP 4:} we now prove the result in part (i). Since $b_{n}$ may have any sub-sequential limit between $1/a$ and $1$ it follows from step 4 that the sequence of random variables (\ref{equation: sequence of normalised maxima}) is tight. Using part (i) of Lemma \ref{lemma: F fuction} the constants in (\ref{equation: gaussian limit constants}) are easily recognised as the largest and smallest constants which may appear in the extreme value limit.  

\end{proof}

\subsection{Proof of Theorem \ref{theorem: tightness of non-gaussian dependent maximum}}

\begin{proof} 
With $W$ satisfying Assumption \ref{assumption: minimum segment} and omitting dependence on $x$ introduce the following notation. 
\begin{equation*}
\mathfrak{u}_{n,W} = \sqrt{2 \log \left ( n / W \right ) } + \left ( \frac{1}{2} \log \log (n / W) - \log (\sqrt{\pi}) + x \right ) / \sqrt{2 \log \left ( n / W \right ) }
\end{equation*}

We first investigate the be behaviour of local test statistics (\ref{equation: local test statistic}) restricted to a particular scale of the order $\mathcal{O} (W)$ under the null of no change points. For some $c > 0$ put $w = \left \lfloor c W \right \rfloor$, and write
\begin{equation*}
M_{cW}^\tau \left ( \mathbf{Y} \right ) = \max \left \{ \frac{1}{\tau} D_{l,w}^p \left ( \mathbf{Y} \right ) \mid 1 \leq l \leq n - w \right \}.
\end{equation*}
Putting $\mathbf{B} = \left ( B (1), \dots, B (n) \right )'$, where $\left \{ B (t) \right \}_{t>0}$ is the process introduced in Assumption \ref{assumption: strong approximation}, making use of Assumption \ref{assumption: strong approximation} the following holds. 
\begin{equation}
M_{n,W}^\tau \left ( \mathbf{Y} \right ) = \max \left \{ D_{l,w}^p \left ( \mathbf{B} \right ) \mid 1 \leq l \leq n - w \right \} + \mathcal{O}_\mathbb{P} \left ( \sqrt{n^{\frac{2}{2+\nu}} / W} \right )
\label{equation: strong approx to maximum}
\end{equation}
Moreover, using Lemma \ref{lemma: high excursion probability - continuous} and arguing as in the proof of Theorem \ref{theorem: distribution on optimal scales} the following holds.  
\begin{align}
\mathbb{P} \left ( M_{cW}^1 \left ( \mathbf{B} \right ) \leq \mathfrak{u}_{n,W} \right ) \nonumber & \sim \prod_{i=0}^{\left \lfloor n / w \right \rfloor} \mathbb{P} \left ( \max \left \{ D_{l,w}^p \left ( \mathbf{B} \right ) \mid i \times w < l \leq (i+1) \times w \right \} \leq \mathfrak{u}_{n,W} \right ) \nonumber \\ 
& \sim \left [ 1 - \mathbb{P} \left ( \sup_{l \in [0,1)} \left \{ \xi \left ( l \right ) \right \} > \mathfrak{u}_{n,W} \right ) \right ] ^ {\left \lfloor n / w \right \rfloor} \nonumber \\
& \sim \exp \left ( - \frac{C_p}{c} e^{-x} \right ) 
\label{equation: WP max distribution}
\end{align}
Therefore, combining (\ref{equation: strong approx to maximum}) and (\ref{equation: WP max distribution}) and arguing as in the proof of Theorem \ref{theorem: Gaussian tightness of normalized maximum}, we immediately have that
\begin{equation*}
\mathbb{P} \left ( M_{cW}^\tau \left ( \mathbf{Y} \right ) \leq \mathfrak{u}_{n,W} \right ) \sim \exp \left ( - \frac{C_{p}}{c} e^{-x} \right ).
\end{equation*} 
On a sub-sequence of $n$'s for which the quantity $b_{n} = a^{\left \lfloor \log_a (W) \right \rfloor} / W$ converges to some constant $b$, arguing as in the proof of Theorem \ref{theorem: Gaussian tightness of normalized maximum}, we therefore have under the null of no change points that
\begin{equation*}
\mathbb{P} \left ( M_{\mathcal{G} (W,a)}^\tau \left ( \mathbf{Y} \right ) \leq \mathfrak{u}_{n,W} \right ) \rightarrow \exp \left ( - \left ( \frac{b^{-1} C_{p}}{1-a^{-1}} \right ) e^{-x} \right ).
\end{equation*}
However, it is again clear that $b_{n}$ can have any sub-sequential limit between $a^{-1}$ and $1$, so part (i) of the theorem is proved. Part (ii) again follows from Lemma \ref{lemma: convergence with other scaling sequcnes}. 
\end{proof}
\subsection{Proof of Lemma \ref{lemma: MAD consistency}}

\begin{proof}

Write $m = \left ( n - p - 1 \right ) / \left ( p + 1 \right )$ and $c_\Phi = \Phi \left ( 2 \Phi^{-1} \left ( 3 / 4 \right ) \right ) - 3/4$. For some $\varepsilon > 0$ not depending on $n$ put $A_\varepsilon = \left ( 3 / c_\Phi + \sqrt{9/c_\Phi^2 + 2 \varepsilon} \right ) / 2 $, and therefore define 
\begin{equation}
\delta = \frac{A_\varepsilon \sigma}{c_\Phi} \left ( \frac{1}{\sqrt{m}} \vee \frac{N}{m} \right ). 
\label{equation: delta}
\end{equation}
We will show that with $n$ sufficiently large 
\begin{equation}
\mathbb{P} \left ( \left | \hat{\sigma}_\text{MAD} - \sigma \right | > \delta \right ) \leq 2 (p+1) e^{-\varepsilon},
\label{equation: MAD prob bound}
\end{equation} 
which implies the desired result. For simplicity assume $n - (p+1)$ is a multiple of $(p+1)$ and introduce the following sets:
\begin{align*}
& I_j = \left \{p+1 \leq t \leq n \mid (t+j) \hspace{0.25em} \mathrm{mod} \hspace{0.25em} (p+1) = 0 \right \} \\
& I_\eta = \cup_{k=1}^N \left \{ \eta_k, \dots, \eta_k + (p+1) \right \} \\
& I_{j,1} = I_j \setminus I_\eta \\
& I_{j,2} = I_j \cap I_\eta. 
\end{align*}
Introducing also the random variables $B_{t}^{\delta} = \boldsymbol{1} \left \{ |X_t| > \Phi^{-1}(3/4) \sqrt{\sum_{i=0}^{p+1} \binom{p+1}{i}^2} \left [ \sigma + \delta \right ] \right \}$ and put $p_\delta = \mathbb{E} \left ( B_t^\delta \mid t \not\in I_\eta \right )$. The following holds via Hoeffding's inequality 
\begin{align}
\mathbb{P} \left ( \widehat{\sigma}_\text{MAD} - \sigma > \delta \right ) & = \mathbb{P} \left ( \frac{\text{median} \left \{ |X_{p+1}|, \dots, |X_n| \right \}}{\Phi^{-1} \left ( 3/4 \right ) \sqrt{\sum_{i=0}^{p+1} \binom{p+1}{i}^2}} > \sigma + \delta \right ) \nonumber \\ 
& \leq \sum_{j=0}^p \mathbb{P} \left ( \sum_{t \in I_{j,1}} B_t^\delta + \sum_{t \in I_{j,2}} B_t^\delta > \frac{n - (p+1)}{2(p+1)} \right ) \nonumber \\ 
& \leq \sum_{j=0}^p \mathbb{P} \left ( \sum_{t \in I_{j,1}} \left ( B_t^\delta - p_\delta \right ) > \frac{n - (p+1)}{2(p+1)} - \left | I_{j,2} \right | - p_\delta \left | I_{j,1} \right | \right ) \nonumber \\ 
& \leq \left ( p + 1 \right ) \exp \left ( - 2m \left [ \left ( 1/2 - p_\delta \right )^2 + \left ( N / m\right )^2 - 2 \left ( 1/2 - p_\delta \right ) \left ( N / m\right )  \right ] ^ 2 \right ). \label{equation: Hoeffding prob bound}
\end{align}
We now bound $p_\delta$ from above and from below. For the lower bound, putting $Z \sim \mathcal{N} (0,1)$ we have that 
\begin{align}
p_\delta & = \mathbb{P} \left ( |Z| > \Phi^{-1} (3/4) \left [ 1 + \delta / \sigma \right ] \right ) \nonumber \\ 
& = 2 \left ( 1 - \int_{-\infty}^{\Phi^{-1} \left ( 3 / 4 \right )} \phi \left ( x \left [ 1 + \delta / \sigma \right ] \right ) \mathrm{d} x \left [ 1 + \delta / \sigma \right ] \right ) \nonumber \\
& \geq 2 \left ( 1 - \Phi \left ( \Phi ^ {-1} (3/4) \right ) \left [ 1 + \delta / \sigma \right ] \right ) \nonumber \\ 
& = 1/2 - (3/2) \times (\delta / \sigma), \label{equation: p delta lower bound}
\end{align}
which holds because for all $\alpha > 1$ and any $x \in \mathbb{R}$ it holds that $\phi(\alpha x) \leq \phi (x) $. For the upper bound write $f (x) = \Phi \left ( \Phi^{-1} \left ( 3/4 \right ) \left ( 1 + x \right ) \right )$ for $x \in [0,1]$. Then using the facts that (i) $f(0) = 3/4$, (ii) $f(1) = \Phi \left ( 2 \Phi^{-1} \left ( 3/4 \right ) \right )$, (iii) $\Phi$ is concave on $[0,1]$, and (iv) $\delta / \sigma \leq 1$ for $n$ sufficiently large, we obtain that
\begin{equation}
p_\delta = 2 \left ( 1 - \Phi \left ( \Phi^{-1} \left ( 3/4 \right ) \left [ 1 + \delta / \sigma \right ] \right )  \right ) \leq 1/2 - \frac{c_\Phi \delta}{\sigma}
\label{equation: p delta upper bound}
\end{equation}
for $n$ sufficiently large. Therefore plugging (\ref{equation: p delta lower bound}), (\ref{equation: p delta upper bound}) and (\ref{equation: delta}) into (\ref{equation: Hoeffding prob bound}) we obtain that 
\begin{align*}
& \mathbb{P} \left ( \widehat{\sigma}_\text{MAD} - \sigma > \delta \right ) \\
& \leq \left ( p + 1 \right ) \exp \left ( - 2 m \left [ A_\varepsilon^2 \left ( m^{-1/2} \vee N/m \right )^2 + (N/m)^2 - \frac{3A_\varepsilon}{c_\Phi} \left ( m^{-1/2} \vee N/m \right ) \left ( N / m \right ) \right ] \right ) \\ 
& \leq \left ( p + 1 \right ) \exp \left ( - 2m \left [ \left ( A_\varepsilon^2 - \frac{3 A_\varepsilon}{c_\Phi} \right ) \left ( m^{-1/2} \vee N/m \right )^2 \right ] \right ) \\
& \leq \left ( p + 1 \right ) \exp \left ( - \varepsilon \right ). 
\end{align*}
Similar arguments give identical bounds on the probability that $\hat{\sigma}_\text{MAD} - \sigma$ is smaller that a given $\delta$, which overall establishes (\ref{equation: MAD prob bound}). 
\end{proof}

\subsection{Proof of Lemma \ref{lemma: consistency of diff-sd estimator}} \label{section: proof of diff estimator}

\begin{proof}

Write $\gamma_i = \max_{1 \leq t \leq n} \mathbb{E} \left | \zeta_t / \sigma \right |^i$ for each $i = 2,3$ and put $\boldsymbol{D}_p = \boldsymbol{\tilde{D}}_p'\boldsymbol{\tilde{D}}_p$ where $\boldsymbol{\tilde{D}}_p$ is the $n \times n$ difference matrix such that each entry in the vector $\boldsymbol{D}_p \boldsymbol{x}$ is the $(p+1)$-th difference of the corresponding entry in the $n$-vector $\boldsymbol{x}$ scaled by
\begin{equation*}
1/\sqrt{\sum_{i=0}^{p+1}\binom{p+1}{i}^2}.
\end{equation*}
Writing $\boldsymbol{Y = f + \zeta}$ the equation below follows directly from equation (6) in \cite{dette1998estimating}. 
\begin{align*}
\mathbb{E} \left [ \left | \widehat{\sigma}^2_{\text{DIF}} - \sigma^2 \right |^2 \right ] & \leq \left [ \left ( \boldsymbol{f}' \boldsymbol{D}_p \boldsymbol{f} \right )^2 + 4 \sigma^2 \boldsymbol{f}' \boldsymbol{D}_p^2 \boldsymbol{f} 
 + 4 \boldsymbol{f}' \left ( \boldsymbol{D}_p \text{diag} \left ( \boldsymbol{D}_p \right ) \boldsymbol{1} \right ) \sigma^3 \gamma_3 + \dots  \right. \\
 & \dots + \left. \sigma^4 \text{trace} \left \{ \text{diag} \left ( \boldsymbol{D}_p \right ) ^ 2 \right \} \left ( \gamma_4 - 3 \right ) + 2 \sigma^4 \text{trace} \left ( \boldsymbol{D}_p^2 \right ) \right] / \left ( n - p - 1 \right ) ^ 2
\end{align*}
Since the noise terms have bounded fourth moment and function $f_\circ (\cdot)$ is assumed to be bounded the following must hold. 
\begin{gather*}
\sigma^4 \text{trace} \left \{ \text{diag} \left ( \boldsymbol{D}_p \right ) ^ 2 \right \} \left ( \gamma_4 - 3 \right ) + 2 \sigma^4 \text{trace} \left ( \boldsymbol{D}_p^2 \right ) = \mathcal{O} \left ( n \right ) \\
\left ( \boldsymbol{f}' \boldsymbol{D}_p \boldsymbol{f} \right )^2 + 4 \sigma^2 \boldsymbol{f}' \boldsymbol{D}_p^2 \boldsymbol{f} + 4 \boldsymbol{f}' \left ( \boldsymbol{D}_p \text{diag} \left ( \boldsymbol{D}_p \right ) \boldsymbol{1} \right ) \sigma^3 \gamma_3 = \mathcal{O} \left ( N^2 \right )
\end{gather*}
It therefore follows that
\begin{equation*}
\mathbb{E} \left [ \left | \widehat{\sigma}^2_{\text{DIF}} - \sigma^2 \right |^2 \right ] \leq \mathcal{O} \left ( \frac{1}{n} \vee \frac{N^2}{n^2} \right ),
\end{equation*}
and as such the desired result follows by Chebyshev's inequality. 

\end{proof}

\subsection{Proof of Lemma \ref{lemma: consistency of LRV estimator}}

\begin{proof}
Write $\boldsymbol{\bar{Y}} = \left ( \bar{Y}_{1,W'}, \dots, \bar{Y}_{\left \lfloor n / W' \right \rfloor, W'} \right )'$ and let $\boldsymbol{\bar{f}}$ and $\boldsymbol{\bar{\zeta}}$ be defined analogously. Let $\boldsymbol{D_p}$ be as defined in the proof of the last lemma, with its dimensions suitably adjusted. Finally put $m = \left \lfloor n / W' \right \rfloor - (p+1)$. We can therefore write $\hat{\tau}_\text{DIF}^2 = \frac{1}{mW}\boldsymbol{\bar{Y}'D_p\bar{Y}}$, and the absolute difference between our estimator and the truth can be bounded as follows:
\begin{align*}
\left | \hat{\tau}_\text{DIF}^2 - \tau^2 \right | & = \left | \frac{1}{mW'} \boldsymbol{\left ( \bar{f} + \bar{\zeta} \right ) ' D_p \left ( \bar{f} + \bar{\zeta} \right )} - \tau^2 \right | \\ 
& \cleq \left | \frac{1}{mW'}\boldsymbol{\bar{\zeta} ' D_p \bar{\zeta}} - \frac{1}{mW'} \mathbb{E} \left ( \boldsymbol{\bar{\zeta} ' D_p \bar{\zeta}} \right ) \right | + \left | \frac{1}{mW'}\mathbb{E} \left ( \boldsymbol{\bar{\zeta} ' D_p \bar{\zeta}} \right ) - \tau^2 \right | \\ 
& \hspace{2em}+ \frac{1}{mW'} \left | \boldsymbol{\bar{f}'D_p\bar{f}} \right | + \frac{1}{mW'} \left | \boldsymbol{\bar{f}'D_p\bar{\zeta}} \right | \\
& = T_1 + T_2 + T_3 + T4.
\end{align*}
We now bound each of the terms in turn. Introducing the notation 
\begin{equation*}
\psi_{p,j} = \left ( -1 \right )^{p+1-j} \binom{p+1}{j} / \sqrt{\sum_{i=0}^{p+1}\binom{p+1}{i}^2}.
\end{equation*}
We can therefore write
\begin{align*}
& \frac{1}{mW'}\boldsymbol{\bar{\zeta} ' D_p \bar{\zeta}} \\
& = \frac{1}{m} \sum_{s=p+2}^{\left \lfloor n / W' \right \rfloor} \left ( \sum_{j=0}^{p+1} \psi_{p,j} \left ( \bar{\zeta}_{s-j,W'} / \sqrt{W'} \right ) \right )^2 \\
& = \frac{1}{m} \sum_{s=p+2}^{\left \lfloor n / W' \right \rfloor} \left [ \sum_{j=0}^{p+1} \psi_{p,j}^2 \left ( \bar{\zeta}_{s-j,W'} / \sqrt{W'} \right )^2 + \sum_{\substack{k \neq l \\ 0 \leq k , l \leq p+1}} \psi_{p,k} \psi_{p,l} \left ( \bar{\zeta}_{s-k,W'} / \sqrt{W'} \right ) \left ( \bar{\zeta}_{s-l,W'} / \sqrt{W'} \right ) \right ]. 
\end{align*}
From which it follows that
\begin{align*}
\frac{1}{mW'} \mathbb{E} \left ( \boldsymbol{\bar{\zeta} ' D_p \bar{\zeta}} \right ) & = \sum_{j=0}^{p+1} \psi_{p,j}^2 \left ( \gamma_0 + 2 \sum_{h=1}^{W'-1} \left ( 1 - \frac{h}{W'} \right ) \gamma_h \right ) \\
& + \sum_{\substack{k \neq l \\ 0 \leq k , l \leq p+1}} \psi_{p,k} \psi_{p,l} \left ( \gamma_{W' |k-l|} + 2 \sum_{h=1}^{W'-1} \left ( 1 - \frac{h}{W'} \right ) \gamma_{W' |k-l| + h} \right ).
\end{align*}
Using these facts term $T_1$ can be bounded as follows
\begin{align*}
T_1 & \leq \left | \frac{1}{m} \sum_{s=p+2}^{\left \lfloor n / W' \right \rfloor} \sum_{j=0}^{p+1} \psi_{p,j}^2 \left ( \left ( \bar{\zeta}_{s-j,W'} / \sqrt{W'} \right )^2 - \gamma_0 - 2 \sum_{h=1}^{W'-1} \left ( 1 - \frac{h}{W'} \right ) \gamma_h \right ) \right | \\
& + \left | \frac{1}{m} \sum_{s=p+2}^{\left \lfloor n / W' \right \rfloor} \sum_{\substack{k \neq l \\ 0 \leq k , l \leq p+1}} \psi_{p,k} \psi_{p,l} \left ( \left ( \bar{\zeta}_{s-k,W} / \sqrt{W'} \right ) \left ( \bar{\zeta}_{s-l,W'} / \sqrt{W'} \right ) \right . \right . \\
& \hspace{2em} - \left . \left . \gamma_{W' |k-l|} - 2 \sum_{h=1}^{W'-1} \left ( 1 - \frac{h}{W'} \right ) \gamma_{W' |k-l| + h} \right ) \right | \\
& = T_{1,1} + T_{1,2}.
\end{align*}
For the first term we have that
\begin{align*}
T_{1,1} & = \left | \frac{1}{m} \sum_{s=p+2}^{\left \lfloor n / W' \right \rfloor} \sum_{j=0}^{p+1} \psi_{p,j}^2 \left ( \frac{1}{W} \sum_{t=W'(s-j-1)+1}^{W'(s-j)} \zeta_t^2 \right . \right . \\ 
& \left . \left . \hspace{2em} + \frac{2}{W'} \sum_{h=1}^{W'-1} \sum_{t=W'(s-j-1)+1}^{W'(s-j)-h} \zeta_{t} \zeta_{t+h} - \gamma_0 - 2 \sum_{h=1}^{W'-1} \left ( 1 - \frac{h}{W'} \right ) \gamma_h\right ) \right | \\
& = \left | \frac{1}{m} \sum_{s=p+2}^{\left \lfloor n / W' \right \rfloor} \sum_{j=0}^{p+1} \psi_{p,j}^2 \left ( \frac{1}{W'} \sum_{t=W'(s-j-1)+1}^{W'(s-j)} \left ( \zeta_t^2 - \gamma_0 \right ) \right . \right . \\
& \left . \left . \hspace{2em} + \sum_{h=1}^{W'-1} \frac{1}{\left ( W' - h \right )} \sum_{t=W'(s-j-1)+1}^{W'(s-j)-h} \left ( 1 - \frac{h}{W'} \right ) \left ( \zeta_{t} \zeta_{t+h} - \gamma_h \right ) \right ) \right | \\
& \leq \sum_{j=0}^{p+1} \psi_{p,j}^2 \left \{ \left | \frac{1}{mW'} \sum_{s=p+2}^{\left \lfloor n / W' \right \rfloor} \sum_{t=W'(s-j-1)+1}^{W'(s-j)} \left ( \zeta_t^2 - \gamma_0 \right ) \right | \right . \\ 
& \left . \hspace{2em} + \sum_{h=1}^{W'-1} \left | \frac{1}{m\left ( W' - h \right )} \sum_{s=p+2}^{\left \lfloor n / W' \right \rfloor} \sum_{t=W'(s-j-1)+1}^{W'(s-j)-h} \left ( \zeta_{t} \zeta_{t+h} - \gamma_h \right ) \right | \right \} \\
& = \sum_{j=0}^{p+1} \psi_{p,j}^2 \left \{ \mathcal{O}_\mathbb{P} \left ( \frac{1}{\sqrt{mW'}} \right ) + \sum_{h=1}^{W'} \mathcal{O}_\mathbb{P} \left ( \frac{1}{\sqrt{m \left ( W' - h \right )}} \right ) \right \} \\
& \leq \mathcal{O}_\mathbb{P} \left ( \frac{W'}{\sqrt{n}} \right ). 
\end{align*}
Where in the last line we have used the fact that $m \sim n / W'$ along with the fact that
\begin{equation*}
\sum_{h=1}^{W'-1} \frac{1}{\sqrt{n \left ( 1 - \frac{h}{W'} \right )}} < \frac{1}{\sqrt{n}} \left ( \int_{1}^{W'-1} \frac{1}{\sqrt{1-\frac{x}{W'}}} \mathrm{d}x + \sqrt{W'}\right ) = \frac{2 W'}{\sqrt{n}} \left ( 1 + o(1) \right ). 
\end{equation*}
Arguing analogously we likewise have that $T_{1,2} \leq \mathcal{O}_\mathbb{P} \left ( \frac{W'}{\sqrt{n}} \right ) $. For the second term we have that
\begin{align*}
T_2 & = \left | \gamma_0 + 2 \sum_{h=1}^{W'-1} \left ( 1 - \frac{h}{W'} \right ) \gamma_h \right . \\
& \hspace{2em} \left . + \sum_{\substack{k \neq l \\ 0 \leq k , l \leq p+1}} \psi_{p,k} \psi_{p,l} \left ( \gamma_{W' |k-l|} + 2 \sum_{h=1}^{W'-1} \left ( 1 - \frac{h}{W'} \right ) \gamma_{W' |k-l| + h} \right ) - \gamma_0 - 2 \sum_{h=1}^\infty \gamma_h \right | \\
& \leq 2 \left | \sum_{h=1}^{W'-1} \left ( 1 - \frac{h}{W'} \right ) \gamma_h - \left \{ \sum_{h=1}^{W'-1} + \sum_{h=W'}^\infty \right \} \gamma_h \right | + 2 \sum_{\substack{k \neq l \\ 0 \leq k , l \leq p+1}} \psi_{p,k} \psi_{p,l} \left | \sum_{h=0}^{W'-1} \gamma_{W' |k-l| + h} \right | \\
& \leq 2 \sum_{h=1}^{W'-1} \frac{h}{W'} \left | \gamma_h \right | + 2 \sum_{h=W'}^\infty \left | \gamma_h \right | + 2 \sum_{\substack{k \neq l \\ 0 \leq k , l \leq p+1}} \psi_{p,k} \psi_{p,l} \sum_{h=0}^{W'-1} \left | \gamma_{W' |k-l| + h} \right | \\
& < \frac{2}{W'} \left ( \sum_{h=1}^{\infty} h \left | \gamma_h \right | + \sum_{\substack{k \neq l \\ 0 \leq k , l \leq p+1}} \psi_{p,k} \psi_{p,l} \sum_{h=0}^{W'-1} \left ( W' |k-l| + h \right ) \left | \gamma_{W' |k-l| + h} \right | \right ) \\
& = \mathcal{O} \left ( W'^{-1} \right ).
\end{align*}
For the third term we have that $T_3 \leq \mathcal{O} \left ( \frac{NW'^2}{n} \right )$ and for the fourth term we likewise have that $T_4 \leq \mathcal{O} \left ( \frac{NW'^2}{n} \right )$. Combining the bounds on terms $T_1$, $T_2$, $T_3$, and $T_4$ the stated result follows. 
\end{proof}

\subsection{Proof of Theorem \ref{theorem: optimality of change point detection}}

\begin{proof}
With slight abuse of notation write $I \in \mathcal{G} \left ( W , a \right )$ if $I = \left \{ l, \dots, l + w - 1 \right \}$, where $(l,w) \in \mathcal{G} \left ( W , a \right )$. For each $k = 1, \dots, N$ introduce the set of intervals
\begin{equation*}
\mathcal{I}_k = \left \{ I \in \mathcal{G} \left ( W , a \right ) \mid \eta_k \in I, \left \lfloor \frac{\left | I \cap \left \{ 1, \dots, \eta_k \right \} \right |}{p+2} \right \rfloor = (p+1) \left \lfloor \frac{\left | I \cap \left \{ \eta_{k}+1, \dots, n \right \} \right |}{p+2} \right \rfloor \right \}.
\end{equation*}
Moreover assume that 
\begin{equation*}
\delta_k > 2 a \left ( p + 2 \right ) \left ( W \vee n^{\frac{2p^*_k}{2p^*_k + 1}} \left ( \frac{16C_{p,p^*_k}^2 \tau^2 \lambda_\alpha^2}{\left | \Delta_{p^*_k,k} \right |^2} \right ) ^ {\frac{1}{2p^*_k + 1}} \right ), \hspace{2em} k = 1, \dots, N.
\end{equation*}
Since $\lambda_\alpha^2 = \mathcal{O} \left ( \log (n) \right )$ for any fixed $\alpha$ and either of threshold (\ref{equation: Gaussian FWE control lambda}) or threshold (\ref{equation: generic FWE control lambda}), this assumption can be seen to correspond to condition (\ref{equation: spacing condition}) in Theorem \ref{theorem: optimality of change point detection}. For ease of reading introduce the notation
\begin{equation*}
V_k^\alpha \left ( n \right ) = n^{\frac{2p^*_k}{2p^*_k + 1}} \left ( 16C_{p,p^*}^2 \tau^2 \lambda_\alpha^2 / \left | \Delta_{p^*_k,k} \right |^2 \right ) ^ {\frac{1}{2p^*_k + 1}}, \hspace{2em} k = 1, \dots, N.
\end{equation*}
Due to lemma \ref{lemma: change point detection on an interval}, testing for a change point on an interval $I' \in \mathcal{I}_k$ using (\ref{equation: local tests}) with threshold $\lambda_\alpha$ the $k$-th change point will be detected as long as $\left | I' \right | > (p+1) V_k^\alpha \left ( n \right )$ on the event  

\begin{equation}
\left \{ L_{\mathcal{G} \left ( W , a \right )}^{\widehat{\tau}} \left ( \boldsymbol{\zeta} \right ) \leq \lambda_\alpha \right \} \cap \left \{ \widehat{\tau} < 2 \tau \right \}. 
\label{equation: noise small tau well estimated}
\end{equation}

Therefore, there must be an interval $I'' \in \mathcal{I}_k$ with $\left | I'' \right | < a \left ( p + 2 \right ) \left ( W \vee V_k^\alpha \left ( n \right ) \right )$ on which the $k$-th change can be detected. By the assumption on the $\delta$'s and the above discussion, the shortest interval in $\mathcal{G} \left ( W , a \right )$ on which the $k$-th chaneg point can be detected will not overlap with the shortest intervals on which the $(k-1)$-th and $(k+1)$-th changes will be detected. Finally, on the event (\ref{equation: noise small tau well estimated}) no test carried out on a sub-interval which are free from change points will spuriously reject. Therefore, events $E_3^*$, $E_4^*$, and $E_5^*$ are verified. 
\end{proof}

\subsection{Proof of Lemma \ref{lemma: sSIC consistency}}

\begin{proof}
We must show that $\text{sSIC} (p') > \text{sSIC} (p)$ for all $p' \neq p$ in the set $\left \{ \underline{p}, \dots, \overline{p} \right \}$. We begin with the case $p' > p$ for which we have that
\begin{align*}
& \text{sSIC} (p') - \text{sSIC} (p) \\ 
& \hspace{4em} = \frac{n}{2} \log \left ( 1 - \frac{\hat{\sigma}^2_{p} - \hat{\sigma}^2_{p'}}{\hat{\sigma}^2_p} \right ) + \left [ \left ( \hat{N}_{p'} + 1 \right ) \left ( p' + 1 \right ) - \left ( \hat{N}_p + 1 \right ) \left ( p + 1 \right ) \right ] \log^\alpha (n) \\
& \hspace{4em} := T_1 + T_2.
\end{align*}
Observe that by Corollary \ref{corollary: large sample consistency} on a set with probability $1 + o(1)$ we will have that $\hat{N}_{p'} = \hat{N}_p = N$. Therefore, the fact that the $\zeta$'s are Gaussian combined with the $\ell_2$ risk of constrained least squares spline estimators, which can be found for example in \cite{shen2022phase}, guarantee that on a set with probability $1+o(1)$ we will have that $ | \hat{\sigma}^2_{p'} - \sigma^2 | \cleq n^{-1} \log (n) $ for each $p' \geq p$. Consequently
\begin{align*}
T_1 \cgeq -\frac{n}{2} \left ( \hat{\sigma}^2_{p} - \hat{\sigma}^2_{p'} \right ) / \hat{\sigma}^2_p \geq - \frac{n}{2} \left( \left | \hat{\sigma}^2_p - \sigma^2 \right | +  \left | \hat{\sigma}^2_{p'} - \sigma^2 \right | \right ) / \hat{\sigma}^2_p \cgeq - \log (n). 
\end{align*}
Again by Corollary \ref{corollary: large sample consistency} we have that with high probability
\begin{equation*}
T_2 = \left ( N + 1 \right ) \left ( p' - p \right ) \log^\alpha (n) \gg \log (n).  
\end{equation*}
Consequently, for $n$ sufficiently larger we have that with high probability $\text{sSIC} (p') - \text{sSIC} (p) > 0$ for $p' > p$. Next we consider the case $p' < p$ for which we have that
\begin{align}
& \text{sSIC} (p') - \text{sSIC} (p) \nonumber \\ 
& \hspace{4em} = \frac{n}{2} \log \left ( \frac{\hat{\sigma}^2_{p'}}{\hat{\sigma}^2_p} \right ) + \left [ \left ( \hat{N}_{p'} + 1 \right ) \left ( p' + 1 \right ) - \left ( \hat{N}_p + 1 \right ) \left ( p + 1 \right ) \right ] \log^\alpha (n) \nonumber \\
& \hspace{4em} := T_1 + T_2.
\label{equation: sSIC bound 2} 
\end{align}
By condition (iii) on a high probability set we must have that $T_1$ is negative and of the order $\mathcal{O} \left ( n \log (n) \right )$, while $\hat{N}_{p'}$ will be of the order $\mathcal{O} ( n / \log (n) )$. Therefore, since $\alpha > 1$ we are done. Since $(\overline{p} - \underline{p}) = \mathcal{O} (1)$ a union bound argument is sufficient to establish that with $n$ sufficiently large, on a high probability set, $\text{sSIC} (p') > \text{sSIC} (p)$ for all $p' \neq p$. This completes the proof. 
\end{proof}

\subsection{Remarks on Assumption \ref{assumption: One prominent jump}} \label{remark: assumption 3}

We remark that \ref{assumption: One prominent jump} was made for ease of technical exposition, and although it does not seem straightforward to relax the assumption in full generality we conjecture that Algorithm \ref{algorithm: binary segmentation over grids} is able to localize all change points at the optimal rate when the assumption is violated, albeit with different leading constants in (\ref{equation: spacing condition}). The reason for the claim is the following: Assumption \ref{assumption: One prominent jump} is made to avoid the possibility of signal cancellation, however examining the proof of Lemma \ref{lemma: change point detection on an interval} it can be seen that there are only $p$ values of $\delta'$ for which exact signal cancellation, and for any such $\delta'$ increasing or decreasing $\delta'$ by a constant will result in an interval of the same order for which no signal cancellation occurs.

Here we show that Algorithm \ref{algorithm: binary segmentation over grids} can localize change points at the optimal rate in the absence of Assumption \ref{assumption: One prominent jump} the when signal is piecewise linear. Moreover we provide some simulated examples of piecewise polynomial signals which violate Assumption \ref{assumption: One prominent jump} and show that the change points are still detected. 

\subsubsection{Relaxing Assumption \ref{assumption: One prominent jump} for piecewise linear signals signals}

Here we show that for piecewise linear signals Algorithm \ref{algorithm: binary segmentation over grids} is able to localize all changes at the minimax optimal rate when Assumption \ref{assumption: One prominent jump} does not hold, provided the remaining assumption in Theorem \ref{theorem: optimality of change point detection} hold. Without loss of generality we consider the case of a single change:
\begin{equation*}
f_\circ \left ( t / n \right ) = 
\begin{cases}
\alpha_0 + \alpha_1 \left ( t/n - \eta / n \right ) & \text{ if } t \leq \eta \\
\beta_0 + \beta_{1} \left ( t/n - \eta / n \right ) & \text{ if } t > \eta
\end{cases}. 
\end{equation*}
Therefore we will show that using the threshold $\lambda = \hat{\tau} \bar{\lambda}$, for some $\bar{\lambda} > 0$, on a high probability set the change can be detected on an interval of length at most
\begin{equation*}
C n^{\frac{2p^*}{2p^*+1}} \left ( 16 \tau^2 \bar{\lambda}^2 / \Delta_{p^*}^2 \right )^{\frac{1}{2p^*+1}},
\end{equation*}
where $C$ is a sufficiently large constant and $p^* \in \left \{ 0 , 1 \right \}$ is defined as in (\ref{equation: most prominent change}). If $\text{sign} (\alpha_0 - \beta_0) = \text{sign} (\alpha_1 - \beta_1)$ this can be shown precisely as in Lemma \ref{lemma: change point detection on an interval}. Therefore, we examine the setting in which $\text{sign} (\alpha_0 - \beta_0) \neq \text{sign} (\alpha_1 - \beta_1)$, for which there are three possible cases of interest: 
\begin{itemize}
\item Case I: $\Delta_0 = \Delta_1 \left ( \delta / n \right )$
\item Case II: $\Delta_0 > \Delta_1 \left ( \delta / n \right )$
\item Case III: $\Delta_0 < \Delta_1 \left ( \delta / n \right )$
\end{itemize}
Similar to Lemma \ref{lemma: change point detection on an interval}, without loss of generality we let $\delta'$ be an integer such that the change occurs at location $\delta'$ and put $m = (p+2) \delta'$. We therefore need to show that the statistic $| D_{1,m}^1 \left ( \boldsymbol{Y} \right ) |$ can detect the change point with high probability for an appropriately chosen $\delta'$. For ease of reading introduce the notation
\begin{align*}
& C_1 = 1 / \sqrt{\sum_{i=0}^2 \binom{2}{i}^2} \\
& g_{\delta'} = \frac{1}{\delta'} \sum_{t=1}^{\delta'} \left ( 1 - t / \delta' \right ) \text{ for } \delta' \in \mathbb{N}.  
\end{align*}

\textbf{Case I:} let $\delta'$ be an integer for which $\delta' < \delta / 2 $. Using the facts that $\Delta_1 / \Delta_0 = n / \delta$ and $g_{\delta'} < 1/2$ for all $\delta'$ we have that 
\begin{align*}
\left | D_{1,m}^1 \left ( \boldsymbol{f} \right ) \right | & \geq C_1 \sqrt{\delta'} \left ( \Delta_0 - g_{\delta'} \Delta_1 \left ( \delta' / n \right ) \right ) \\
& = C_1 \sqrt{\delta'} \left ( \Delta_0 - g_{\delta'} \Delta_0 \left ( \delta' / \delta \right ) \right ) \\
& \geq \frac{3C_1}{4} \sqrt{\delta'} \Delta_0 
\end{align*}
and the desired result follows by rearranging (\ref{equation: D trinagle inequality bound}). 

\textbf{Case II:} this can be treated similarly to Case I. 

\textbf{Case III:} note that there is a $\delta''$ for which $\Delta_0 = \Delta_1 \left ( \delta'' / n \right )$. We first consider the setting where $\delta'' < \left ( 2 / C_1 \right ) ^ 2 \left ( 16 \tau^2 \bar{\lambda}^2 / \Delta_1^2 \right ) ^ {1/3}$, in which case letting $\delta'$ be such that $\delta' > 24 \delta''$, and using the fact that $g_{\delta'} \geq 1/12$ for all $\delta' > 1$ by (\ref{equation: g1 bound}), we have that
\begin{align*}
\left | D_{1,m}^1 \left ( \boldsymbol{f} \right ) \right | & \geq C_1 \sqrt{\delta'} \left ( g_{\delta'} \Delta_1 \left ( \delta ' / n \right ) - \Delta_0 \right ) \\
& \geq \frac{C_1}{12} \sqrt{\delta'} \left ( \Delta_1 \left ( \delta' / n \right ) - 12 \Delta_0 \right ) \\
& \geq \frac{C_1}{24} \sqrt{\delta'} \Delta_1 \left ( \delta' / n \right ). 
\end{align*}
Therefore, rearranging (\ref{equation: D trinagle inequality bound}) and accounting for the facts that we must have $\delta' > 24 \delta''$ we obtain that the change will be detected as soon as
\begin{equation*}
\delta' \geq 24 \left ( 2 / C_1 \right ) ^ 2 \left ( 16 \tau^2 \bar{\lambda}^2 / \Delta_1^2 \right ) ^ {1/3}.
\end{equation*}
Finally we consider the case $\delta'' \geq \left ( 2 / C_1 \right ) ^ 2 \left ( 16 \tau^2 \bar{\lambda}^2 / \Delta_1^2 \right ) ^ {1/3}$. In this case, letting $\delta' \leq \delta''$  and using the fact that $\Delta_0 \geq \Delta_1 \left ( \delta' / n \right )$ for all such $\delta'$ we obtain that
\begin{align*}
\left | D_{1,m}^1 \left ( \boldsymbol{f} \right ) \right | & \geq C_1 \sqrt{\delta'} \left ( \Delta_0 - g_{\delta'} \Delta_1 \left ( \delta' / n \right ) \right ) \\
& \geq \frac{C_1}{2} \sqrt{\delta'} \Delta_0 \\
& \geq \frac{C_1}{2} \sqrt{\delta'} \Delta_1 \left ( \delta' / n \right ),
\end{align*}
and as in the previous cases the desired result follows by rearranging (\ref{equation: D trinagle inequality bound}).

\subsubsection{Examples of higher order polynomials which violate \ref{assumption: One prominent jump}}

Here we give simulated examples of higher order piecewise polynomial signals which violate Assumption \ref{assumption: One prominent jump}, and show that Algorithm \ref{algorithm: binary segmentation over grids} is still able to detect the change points in practice. Specifically we consider three piecewise quadratic signals with a single change point at location $\eta$: 
\begin{equation}
f_\circ \left ( t / n \right ) = 
\begin{cases}
\alpha_0 + \alpha_{1} \left ( t/n - \eta / n \right ) + \alpha_{2} \left ( t/n - \eta / n \right )^2 & \text{ if } t \leq \eta \\
\beta_0 + \beta_{1} \left ( t/n - \eta / n \right ) + \beta_{2} \left ( t/n - \eta / n \right )^2 & \text{ if } t > \eta 
\end{cases}
\label{equation: pcws quadratic signal}
\end{equation}
We consider three instances of (\ref{equation: pcws quadratic signal}) where in each case the sample size is $n = 500$, the change point occurs at location $\eta = n / 2$, and changes occur in two derivatives of different order in such a way that the changes work against each other in the sense that they have different signs and the signal strengths as measured by $\Delta_j \left ( \delta / n \right )^j$ for $j = 0,1,2$ exactly match. The three models are denoted by \texttt{M1}, \texttt{M2}, and \texttt{M3} and the values of the $\alpha$'s and $\beta$'s are given in Table \ref{table: assumption 3 models}. 
\begin{table}[!htbp]
\centering
\caption{Values of $\alpha$'s and $\beta$'s for three instances of (\ref{equation: pcws quadratic signal}) which violate Assumption \ref{assumption: One prominent jump} when the sample size is $n = 500$ and the change point occurs at location $\eta = n / 2$.}
\label{table: assumption 3 models}
\begin{tabular}{|c|c|c|c|c|c|c|}
\hline
 & $\alpha_0$ & $\beta_0$ & $\alpha_1$ & $\beta_1$ & $\alpha_2$ & $\beta_2$ \\
\hline
\texttt{M1} & $-1/2$ & $1/2$ & $-2$ & $2$ & $0$ & $0$ \\
\hline
\texttt{M2} & $0$ & $0$ & $6$ & $-6$ & $-12$ & $12$ \\
\hline
\texttt{M3} & $1/2$ & $-1/2$ & $0$ & $0$ & $-2$ & $2$ \\
\hline
\end{tabular}
\end{table}

We contaminate the signals with independent noise having marginal $\mathcal{N} \left ( 0, 0.5^2 \right )$ distribution and apply Algorithm \ref{algorithm: binary segmentation over grids} with parameter $\alpha = 0.1$. The results of this experiment, which was run with random seed \texttt{42} in \texttt{R}, are shown in Figure \ref{figure: signals violating assumption 3.3.3}. In all three cases Algorithm \ref{algorithm: binary segmentation over grids} returns a single interval which contains the true change point location. 

\begin{figure}[!htbp]
\centering
\caption{Piecewise polynomial signals which violate Assumption \ref{assumption: One prominent jump} with coefficients specified in Table \ref{table: assumption 3 models}, contaminated with i.i.d. Gaussian noise having standard deviation $\sigma = 0.5$ (left column). Intervals of significance with uniform $90\%$ coverage returned by our procedure (right column). Black dashed lines (\textbf{- - -}) represent underlying piecewise polynomial signal, light grey lines (\textcolor{gray}{\textbf{---}}) represent the observed data sequence, red shaded regions (\textcolor{transpred}{$\blacksquare$}) represent intervals of significance returned by our procedure.} 
\label{figure: signals violating assumption 3.3.3}
\begin{subfigure}[b]{0.4\textwidth}
\centering
\includegraphics[width=\textwidth]{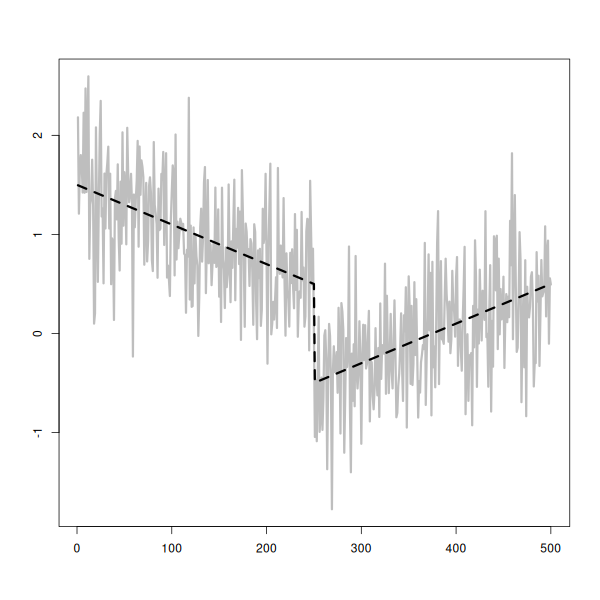}
\caption{\texttt{M1} + Gaussian noise}
\label{subfigure: d0-d1-signal}
\end{subfigure}
\begin{subfigure}[b]{0.4\textwidth}
\centering
\includegraphics[width=\textwidth]{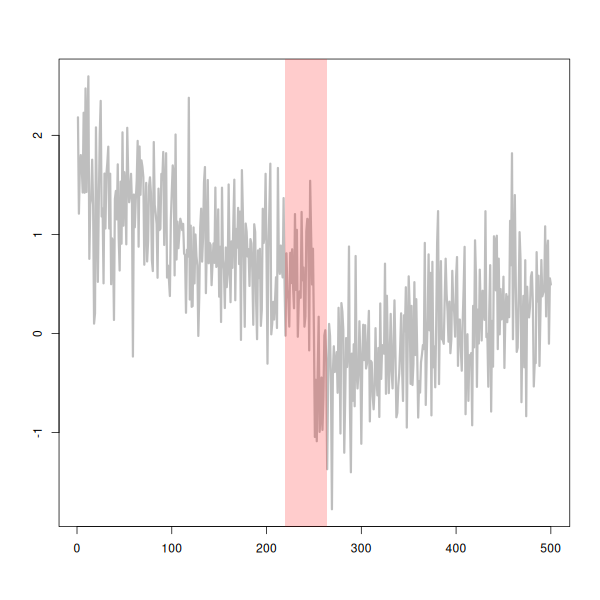}
\caption{intervals returned}
\end{subfigure}

\begin{subfigure}[b]{0.4\textwidth}
\centering
\includegraphics[width=\textwidth]{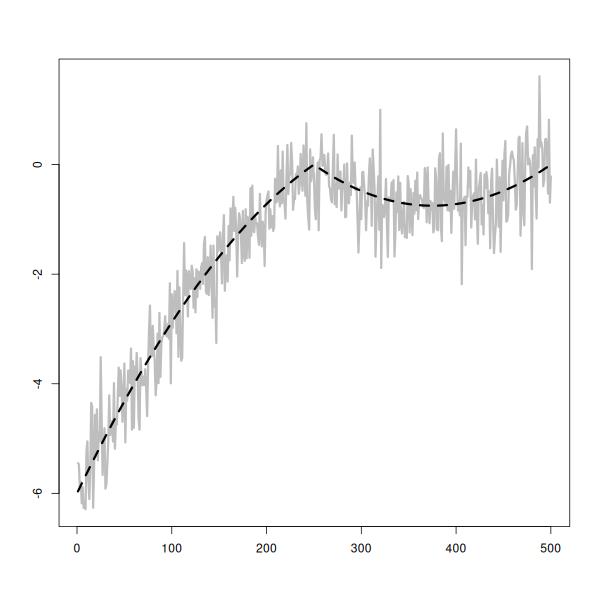}
\caption{\texttt{M2} + Gaussian noise}
\label{subfigure: waves}
\end{subfigure}
\begin{subfigure}[b]{0.4\textwidth}
\centering
\includegraphics[width=\textwidth]{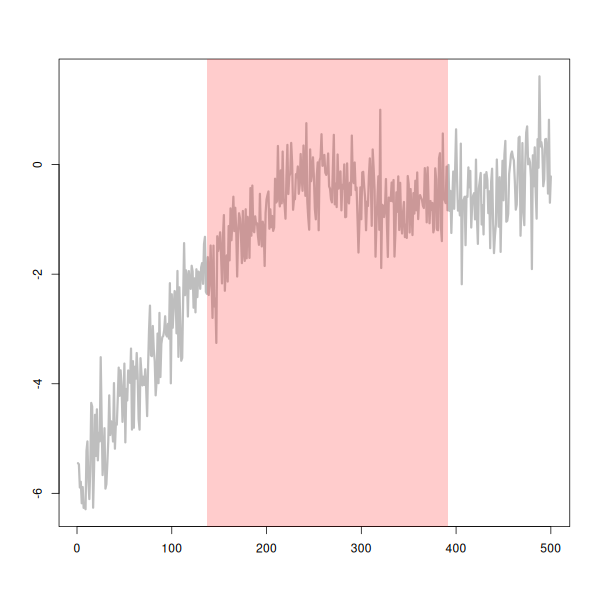}
\caption{intervals returned}
\end{subfigure}

\begin{subfigure}[b]{0.4\textwidth}
\centering
\includegraphics[width=\textwidth]{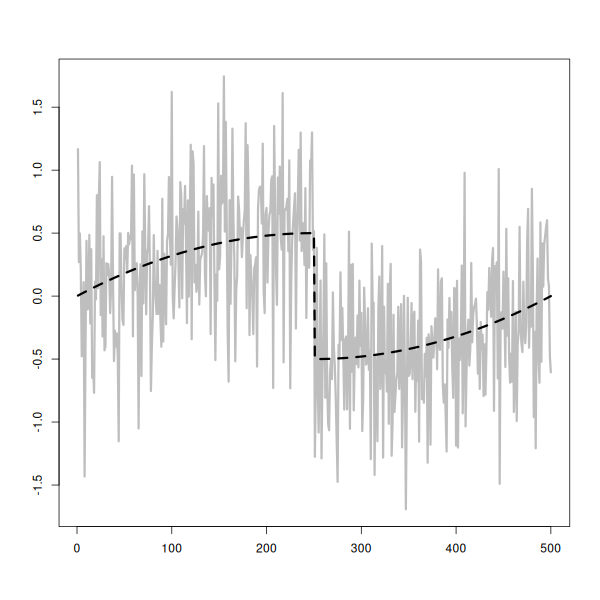}
\caption{\texttt{M3} + Gaussian noise}
\label{subfigure: hills}
\end{subfigure}
\begin{subfigure}[b]{0.4\textwidth}
\centering
\includegraphics[width=\textwidth]{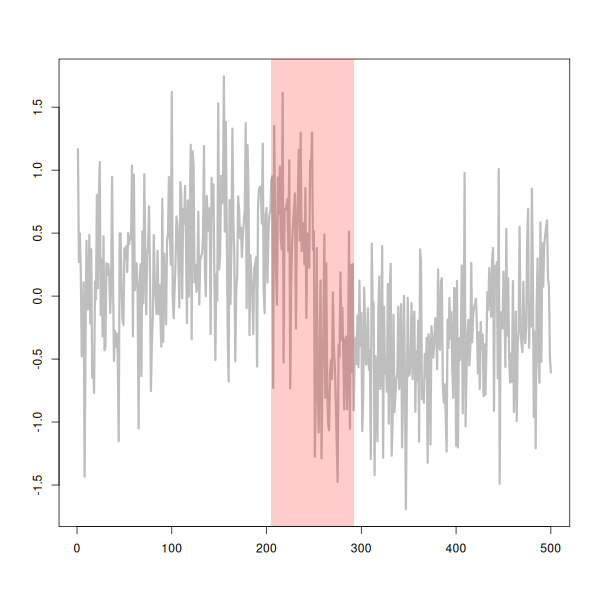}
\caption{intervals returned}
\end{subfigure}
\end{figure}

\section{Additional numerical illustrations} \label{section: additional simulation study}

To further investigate the coverage provided by our method in finite samples, in this section we reproduce the simulation study in Section \ref{section: coverage simulations} for signals of length $n \in \left \{ 100, 500, 1000, 2000 \right \}$. The results are shown in Tables \ref{table: extended simulation study}, and confirm that for a range of signal lengths our procedure continues to provide accurate coverage. 

\begin{table}[!htbp]
\caption{Proportion of times out of $100$ replications each method returned no intervals of significance when applied to a noise vector of length $n \in \left \{ 100, 500, 1000, 2000 \right \}$, as well as whether each method is theoretically guaranteed to provide correct coverage.} 
\label{table: extended simulation study}
\centering
\begin{subtable}{\linewidth}\centering
{\scalebox{0.9}{
\begin{tabular}{|l|l|c|c|c|}
  \hline
 &  & degree 0 & degree 1 & degree 2 \\ 
  \hline
 & DIF1-MAD & 0.93 & 0.92 & 0.94 \\ 
  n = 100 & DIF2-SD & 1.00 & 1.00 & 1.00 \\ 
   & DIF2-LRV & 0.98 & 0.92 & 0.95 \\ 
\hline
   & DIF1-MAD & 0.93 & 0.97 & 0.94 \\ 
  n = 500 & DIF2-SD & 1.00 & 1.00 & 1.00 \\ 
   & DIF2-LRV & 0.99 & 0.98 & 0.96 \\
\hline
   & DIF1-MAD & 0.93 & 0.92 & 0.93 \\ 
  n = 1000 & DIF2-SD & 1.00 & 1.00 & 0.99 \\ 
   & DIF2-LRV & 0.98 & 0.98 & 0.99 \\
\hline
   & DIF1-MAD & 0.88 & 0.93 & 0.95 \\ 
  n = 2000 & DIF2-SD & 1.00 & 0.99 & 0.99 \\ 
   & DIF2-LRV & 0.98 & 0.98 & 0.95 \\ 
   \hline
\end{tabular}
}}
\caption{Coverage on noise type \texttt{N1} with $\sigma = 1$.}
\end{subtable}%

\begin{subtable}{\linewidth}\centering
{\scalebox{0.9}{
\begin{tabular}{|l|l|c|c|c|}
  \hline
 &  & degree 0 & degree 1 & degree 2 \\ 
  \hline
 & DIF1-MAD & 0.88 & 0.64 & 0.81 \\ 
  n = 100 & DIF2-SD & 0.99 & 0.98 & 0.99 \\ 
   & DIF2-LRV & 1.00 & 0.91 & 0.91 \\ 
\hline
   & DIF1-MAD & 0.52 & 0.42 & 0.50 \\ 
  n = 500 & DIF2-SD & 0.98 & 0.97 & 0.99 \\ 
   & DIF2-LRV & 0.99 & 0.93 & 0.94 \\
\hline
   & DIF1-MAD & 0.43 & 0.28 & 0.39 \\ 
  n = 1000 & DIF2-SD & 0.94 & 0.99 & 0.97 \\ 
   & DIF2-LRV & 0.91 & 0.95 & 0.90 \\ 
\hline
   & DIF1-MAD & 0.23 & 0.19 & 0.18 \\ 
  n = 2000 & DIF2-SD & 0.99 & 0.97 & 0.96 \\ 
   & DIF2-LRV & 0.99 & 0.97 & 0.94 \\ 
   \hline
\end{tabular}
}}
\caption{Coverage on noise type \texttt{N2} with $\sigma = 1$.}
\end{subtable}%

\begin{subtable}{\linewidth}\centering
{\scalebox{0.9}{
\begin{tabular}{|l|l|c|c|c|}
  \hline
 &  & degree 0 & degree 1 & degree 2 \\ 
  \hline
 & DIF1-MAD & 0.72 & 0.71 & 0.65 \\ 
  n = 100 & DIF2-SD & 0.98 & 1.00 & 0.99 \\ 
   & DIF2-LRV & 0.95 & 0.94 & 0.91 \\ 
\hline
   & DIF1-MAD & 0.40 & 0.31 & 0.37 \\ 
  n = 500 & DIF2-SD & 0.99 & 0.97 & 1.00 \\ 
   & DIF2-LRV & 0.98 & 0.94 & 0.91 \\ 
\hline
   & DIF1-MAD & 0.23 & 0.26 & 0.37 \\ 
  n = 1000 & DIF2-SD & 0.98 & 0.96 & 0.98 \\ 
   & DIF2-LRV & 0.98 & 0.97 & 0.93 \\
\hline
   & DIF1-MAD & 0.17 & 0.15 & 0.15 \\ 
  n = 2000 & DIF2-SD & 0.96 & 0.98 & 0.99 \\ 
   & DIF2-LRV & 0.97 & 0.98 & 0.97 \\ 
   \hline
\end{tabular}
}}
\caption{Coverage on noise type \texttt{N3} with $\sigma = 1$.}
\end{subtable}

\end{table}

\begin{table}[!htbp]
\caption{Proportion of times out of $100$ replications each method returned no intervals of significance when applied to a noise vector of length $n \in \left \{ 100, 500, 1000, 2000 \right \}$, as well as whether each method is theoretically guaranteed to provide correct coverage.} 
\label{table: extended simulation study dep}
\centering
\begin{subtable}{\linewidth}\centering
{\scalebox{0.9}{
\begin{tabular}{|l|l|c|c|c|}
  \hline
 &  & degree 0 & degree 1 & degree 2 \\ 
  \hline
 & DIF1-MAD & 0.00 & 0.00 & 0.00 \\ 
  n = 100 & DIF2-SD & 0.00 & 0.00 & 0.00 \\ 
   & DIF2-LRV & 0.74 & 0.79 & 0.77 \\ 
\hline
   & DIF1-MAD & 0.00 & 0.00 & 0.00 \\ 
  n = 500 & DIF2-SD & 0.00 & 0.00 & 0.00 \\ 
   & DIF2-LRV & 0.92 & 0.87 & 0.91 \\ 
\hline
   & DIF1-MAD & 0.00 & 0.00 & 0.00 \\ 
  n = 1000 & DIF2-SD & 0.00 & 0.00 & 0.00 \\ 
   & DIF2-LRV & 0.95 & 0.89 & 0.94 \\ 
\hline
   & DIF1-MAD & 0.00 & 0.00 & 0.00 \\ 
  n = 2000 & DIF2-SD & 0.00 & 0.00 & 0.00 \\ 
   & DIF2-LRV & 0.98 & 0.98 & 0.97 \\ 
   \hline
\end{tabular}
}}
\caption{Coverage on noise type \texttt{N4} with $\sigma = 1$.}
\end{subtable}%

\begin{subtable}{\linewidth}\centering
{\scalebox{0.9}{
\begin{tabular}{|l|l|c|c|c|}
  \hline
 &  & degree 0 & degree 1 & degree 2 \\ 
  \hline
 & DIF1-MAD & 0.00 & 0.00 & 0.00 \\ 
  n = 100 & DIF2-SD & 0.00 & 0.00 & 0.00 \\ 
   & DIF2-LRV & 0.75 & 0.69 & 0.72 \\ 
\hline
   & DIF1-MAD & 0.00 & 0.00 & 0.00 \\ 
  n = 500 & DIF2-SD & 0.00 & 0.00 & 0.00 \\ 
   & DIF2-LRV & 0.89 & 0.82 & 0.88 \\ 
\hline
   & DIF1-MAD & 0.00 & 0.00 & 0.00 \\ 
  n = 1000 & DIF2-SD & 0.00 & 0.00 & 0.00 \\ 
   & DIF2-LRV & 0.98 & 0.90 & 0.89 \\ 
\hline
   & DIF1-MAD & 0.00 & 0.00 & 0.00 \\ 
  n = 2000 & DIF2-SD & 0.00 & 0.00 & 0.00 \\ 
   & DIF2-LRV & 0.94 & 0.98 & 0.98 \\ 
   \hline
\end{tabular}
}}
\caption{Coverage on noise type \texttt{N5} with $\sigma = 1$.}
\end{subtable}%

\begin{subtable}{\linewidth}\centering
{\scalebox{0.9}{
\begin{tabular}{|l|l|c|c|c|}
  \hline
 &  & degree 0 & degree 1 & degree 2 \\ 
  \hline
 & DIF1-MAD & 0.00 & 0.00 & 0.00 \\ 
  n = 100 & DIF2-SD & 0.00 & 0.00 & 0.00 \\ 
   & DIF2-LRV & 0.97 & 0.91 & 0.94 \\ 
\hline
   & DIF1-MAD & 0.00 & 0.00 & 0.00 \\ 
  n = 500 & DIF2-SD & 0.00 & 0.00 & 0.00 \\ 
   & DIF2-LRV & 0.99 & 1.00 & 0.99 \\ 
\hline
   & DIF1-MAD & 0.00 & 0.00 & 0.00 \\ 
  n = 1000 & DIF2-SD & 0.00 & 0.00 & 0.00 \\ 
   & DIF2-LRV & 0.95 & 0.98 & 0.95 \\ 
\hline
   & DIF1-MAD & 0.00 & 0.00 & 0.00 \\ 
  n = 2000 & DIF2-SD & 0.00 & 0.00 & 0.00 \\ 
   & DIF2-LRV & 0.99 & 0.97 & 0.99 \\ 
   \hline
\end{tabular}
}}
\caption{Coverage on noise type \texttt{N6} with $\sigma = 1$.}
\end{subtable}

\end{table}

\bibliography{paper-ref}
\bibliographystyle{alpha}

\end{document}